\pdfoutput = 1
\documentclass[11pt]{article}
\usepackage{fullpage}
\usepackage[english]{babel}
\usepackage[utf8x]{inputenc}
\usepackage[T1]{fontenc}
\usepackage{amsmath}
\usepackage{amssymb}
\usepackage{amsthm}
\usepackage{bbm}
\usepackage{bbold}
\usepackage{parskip}
\usepackage{thm-restate}
\usepackage{booktabs}
\usepackage{array}

\usepackage{hyperref}
\hypersetup{
	colorlinks=true,
	linkcolor=blue!70!black,
	citecolor=blue!70!black,
	urlcolor=blue!70!black
}
\usepackage{cleveref}
\usepackage{graphicx}

\usepackage[lined,boxed,ruled,norelsize,algo2e,linesnumbered]{algorithm2e}
\hypersetup{
	colorlinks=true,
	linkcolor=blue!70!black,
	citecolor=blue!70!black,
	urlcolor=blue!70!black
}

\newtheorem{lemma}{Lemma}

\newtheorem{fact}{Fact}

\newtheorem{definition}{Definition}
\newtheorem{corollary}{Corollary}
\newtheorem{proposition}{Proposition}
\newtheorem{remark}{Remark}

\newtheorem{assumption}{Assumption}

\newcommand{\defeq}{:=}

\newcommand{\norm}[1]{\left\lVert#1\right\rVert}
\newcommand{\norms}[1]{\lVert#1\rVert}
\newcommand{\normop}[1]{\left\lVert#1\right\rVert_{\textup{op}}}
\newcommand{\normf}[1]{\left\lVert#1\right\rVert_{\textup{F}}}

\newcommand{\normsop}[1]{\lVert#1\rVert_{\textup{op}}}
\newcommand{\normsf}[1]{\lVert#1\rVert_{\textup{F}}}

\newcommand{\inprod}[2]{\left\langle#1, #2\right\rangle}
\newcommand{\inprods}[2]{\langle#1, #2\rangle}
\newcommand{\eps}{\epsilon}
\newcommand{\lam}{\lambda}
\newcommand{\sig}{\sigma}

\DeclareMathOperator*{\argmin}{arg\,min}
\newcommand{\R}{\mathbb{R}}

\newcommand{\N}{\mathbb{N}}
\newcommand{\diag}[1]{\textbf{\textup{diag}}\left(#1\right)}
\newcommand{\half}{\frac{1}{2}}

\newcommand{\ind}{\mathbb{I}}

\newcommand{\proj}{\boldsymbol{\Pi}}

\newcommand{\Nor}{\mathcal{N}}

\newcommand{\Tr}{\textup{Tr}}
\newcommand{\opt}{\textup{OPT}}

\newcommand{\ball}{\mathbb{B}}
\newcommand{\id}{\mathbf{I}}

\newcommand{\dd}{\textup{d}}
\usepackage{xcolor}
\definecolor{burntorange}{rgb}{0.8, 0.33, 0.0}

\newcommand{\nnz}{\textup{nnz}}

\newcommand{\poly}{\textup{poly}}
\newcommand{\polylog}{\textup{polylog}}
\newcommand{\Par}[1]{\left(#1\right)}
\newcommand{\Brack}[1]{\left[#1\right]}
\newcommand{\Brace}[1]{\left\{#1\right\}}
\newcommand{\Abs}[1]{\left|#1\right|}

\newcommand{\oracle}{\mathcal{O}}

\newcommand{\mzero}{\mathbf{0}}
\newcommand{\Sym}{\mathbb{S}}
\newcommand{\PSD}{\Sym_{\succeq \mzero}}

\newcommand{\ma}{\mathbf{A}}
\newcommand{\mb}{\mathbf{B}}
\newcommand{\mc}{\mathbf{C}}
\newcommand{\md}{\mathbf{D}}

\newcommand{\mg}{\mathbf{G}}
\newcommand{\mh}{\mathbf{H}}
\newcommand{\mi}{\mathbf{I}}
\newcommand{\mj}{\mathbf{J}}
\newcommand{\mk}{\mathbf{K}}
\newcommand{\ml}{\mathbf{L}}
\newcommand{\mm}{\mathbf{M}}
\newcommand{\mn}{\mathbf{N}}

\newcommand{\mmp}{\mathbf{P}}
\newcommand{\mq}{\mathbf{Q}}
\newcommand{\mr}{\mathbf{R}}
\newcommand{\ms}{\mathbf{S}}
\newcommand{\mt}{\mathbf{T}}
\newcommand{\mmu}{\mathbf{U}}
\newcommand{\mv}{\mathbf{V}}

\newcommand{\my}{\mathbf{Y}}
\newcommand{\mz}{\mathbf{Z}}
\newcommand{\msig}{\boldsymbol{\Sigma}}
\newcommand{\mproj}{\boldsymbol{\Pi}}

\newcommand{\0}{\mathbf{0}}
\newcommand{\1}{\mathbf{1}}
\newcommand{\va}{\mathbf{a}}
\newcommand{\vb}{\mathbf{b}}
\newcommand{\vc}{\mathbf{c}}

\newcommand{\ve}{\mathbf{e}}
\newcommand{\vf}{\mathbf{f}}
\newcommand{\vg}{\mathbf{g}}

\newcommand{\vl}{\mathbf{l}}

\newcommand{\vq}{\mathbf{q}}
\newcommand{\vr}{\mathbf{r}}
\newcommand{\vs}{\mathbf{s}}
\newcommand{\vt}{\mathbf{t}}
\newcommand{\vu}{\mathbf{u}}
\newcommand{\vv}{\mathbf{v}}
\newcommand{\vw}{\mathbf{w}}
\newcommand{\vx}{\mathbf{x}}
\newcommand{\vy}{\mathbf{y}}
\newcommand{\vz}{\mathbf{z}}
\newcommand{\tva}{\tilde{\va}}
\newcommand{\vlam}{\boldsymbol{\lambda}}
\newcommand{\vsig}{\boldsymbol{\sigma}}
\newcommand{\vtau}{\boldsymbol{\tau}}

\newcommand{\vmu}{\boldsymbol{\mu}}
\newcommand{\bvx}{\bar{\vx}}
\newcommand{\Span}{\textup{Span}}
\newcommand{\rank}{\textup{rank}}
\newcommand{\conv}{\textup{conv}}

\newcommand{\calA}{\mathcal{A}}
\newcommand{\calB}{\mathcal{B}}

\newcommand{\calE}{\mathcal{E}}
\newcommand{\calF}{\mathcal{F}}

\newcommand{\calL}{\mathcal{L}}

\newcommand{\calP}{\mathcal{P}}

\newcommand{\calS}{\mathcal{S}}
\newcommand{\calT}{\mathcal{T}}

\newcommand{\calW}{\mathcal{W}}

\newcommand{\event}{\calE}
\newcommand{\Schur}{\textup{SC}}
\newcommand{\tma}{\widetilde{\ma}}
\newcommand{\tml}{\widetilde{\ml}}
\newcommand{\vdelta}{\boldsymbol{\delta}}
\newcommand{\hvdelta}{\hat{\vdelta}}
\newcommand{\tmv}{\mathcal{T}_{\textup{mv}}}
\newcommand{\OMDR}{\mathsf{OracleMDR}}
\newcommand{\SOCPack}{\mathsf{SOCPacking}}
\newcommand{\simiid}{\sim_{\textup{i.i.d.}}}
\newcommand{\simunif}{\sim_{\textup{unif.}}}
\newcommand{\net}{\mathcal{N}}
\newcommand{\cop}{C_{\textup{op}}}
\newcommand{\tmy}{\widetilde{\my}}
\newcommand{\bmy}{\overline{\my}}
\newcommand{\tvg}{\tilde{\vg}}
\newcommand{\step}{\boldsymbol{\delta}}

\newcommand{\codeStyle}[1]{{\bfseries #1} }
\newcommand{\codeInput}{\codeStyle{Input:}}	
\newcommand{\codeOutput}{\codeStyle{Output:}}	
\newcommand{\codeReturn}{\codeStyle{Return:}}

\usepackage{mathtools}

\title{Radial Isotropic Position via an Implicit Newton's Method}
\author{
Arun Jambulapati\thanks{Independent, \texttt{jmblpati@gmail.com}. Work completed while visiting the University of Texas at Austin.} \and 
Jonathan Li\thanks{University of Texas at Austin, \texttt{jli@cs.utexas.edu}}\and 
Kevin Tian\thanks{University of Texas at Austin, \texttt{kjtian@cs.utexas.edu}}
}
\date{}
\allowdisplaybreaks
\begin{document}
\pagenumbering{gobble}
\maketitle
\begin{abstract}
Placing a dataset $A = \{\va_i\}_{i \in [n]} \subset \R^d$ in \emph{radial isotropic position}, i.e., finding an invertible $\mr \in \R^{d \times d}$ such that the unit vectors $\{(\mr \va_i) \norms{\mr \va_i}_2^{-1}\}_{i \in [n]}$ are in isotropic position, is a powerful tool with applications in functional analysis, communication complexity, coding theory, and the design of learning algorithms. When the transformed dataset has a second moment matrix within a $\exp(\pm \eps)$ factor of a multiple of $\id_d$, we call $\mr$ an $\eps$-approximate \emph{Forster transform}.

We give a faster algorithm for computing approximate Forster transforms, based on optimizing an objective defined by Barthe \cite{Barthe98}. When the transform has a polynomially-bounded aspect ratio, our algorithm uses $O(nd^{\omega - 1}(\frac n \eps)^{o(1)})$ time to output an $\eps$-approximate Forster transform with high probability, when one exists. This is almost the natural limit of this approach, as even evaluating Barthe's objective takes $O(nd^{\omega - 1})$ time. Previously, the state-of-the-art runtime in this regime was based on cutting-plane methods, and scaled at least as $\approx n^3 + n^2 d^{\omega - 1}$.
We also provide explicit estimates on the aspect ratio in the \emph{smoothed analysis} setting, and show that our algorithm similarly improves upon those in the literature.

To obtain our results, we develop a subroutine of potential broader interest: a reduction from almost-linear time sparsification of graph Laplacians to the ability to support almost-linear time matrix-vector products. We combine this tool with new stability bounds on Barthe's objective to implicitly implement a box-constrained Newton's method \cite{CohenMTV17, Allen-ZhuLOW17}.
\end{abstract}
\thispagestyle{empty}
\newpage
\tableofcontents
\thispagestyle{empty}
\newpage
\pagenumbering{arabic}
\section{Introduction}\label{sec:intro}

Transforming a dataset $A = \{\va_i\}_{i \in [n]} \subset \R^d$ into a canonical representation enjoying a greater deal of regularity is a powerful idea that has had myriad applications throughout computer science, statistics, and related fields. Examples of common such representations include the following.
\begin{itemize}
	\item \textbf{Normalization:} Replacing each $\va_i$ with the unit vector $\tva_i \defeq \va_i \norm{\va_i}_2^{-1}$ in the same direction. Such a transformation exists whenever all of the $\{\va_i\}_{i \in [n]}$ are nonzero vectors.
	\item \textbf{Isotropic position:} Replacing each $\va_i$ with $\tva_i \defeq \mr \va_i$ for an invertible $\mr \in \R^{d \times d}$, such that $\sum_{i \in [n]} \tva_i \tva_i^\top = \id_d$. Such a transformation exists whenever the $\{\va_i\}_{i \in [n]}$ span $\R^d$.
\end{itemize}
Recently, a common generalization of both of these representations known as \emph{radial isotropic position} has emerged as a desirable data processing step in many settings. 

\begin{definition}[Radial isotropic position]\label{def:rip}
	Let $\vc \in (0, 1]^n$ satisfy $\norm{\vc}_1 = d$, and let $\eps \in (0, 1)$. We say that $\ma \in \R^{n \times d}$ with rows $\{\va_i^\top\}_{i \in [n]}$ is in \emph{$(\vc, \eps)$-radial isotropic position} (or, $(\vc, \eps)$-\emph{RIP}) if 
	\begin{equation}\label{eq:rip_def} \exp(-\eps) \id_d \preceq \sum_{i \in [n]} \vc_i \cdot \frac{\va_i\va_i^\top}{\norm{\va_i}_2^2} \preceq \exp(\eps) \id_d.\end{equation}
	If $\eps$ is omitted then $\eps = 0$ by default, and if $\vc$ is omitted then $\vc = \frac d n \1_n$ by default. For an invertible matrix $\mr \in \R^{d \times d}$, we say that $\mr$ \emph{is a $(\vc, \eps)$-Forster transform of $\ma$} if $\ma \mr^\top$ is in $(\vc, \eps)$-RIP:
	\begin{equation}\label{eq:scale_rip_def}
		\exp(-\eps) \id_d \preceq \sum_{i \in [n]} \vc_i \cdot \frac{\Par{\mr\va_i}\Par{\mr\va_i}^\top}{\norm{\mr\va_i}_2^2} \preceq \exp(\eps) \id_d.
	\end{equation}
\end{definition}

In other words, $\mr$ is a $\vc$-Forster transform of $\ma \in \R^{n \times d}$ representing the dataset $A = \{\va_i\}_{i \in [n]} \subset \R^d$, if the transformed-and-normalized vectors $\{(\mr \va_i)\norms{\mr \va_i}_2^{-1}\}_{i \in [n]}$ are in isotropic position. Note that $\norm{\vc}_1 = d$ in Definition~\ref{def:rip} is necessary as $\eps \to 0$, by taking traces of \eqref{eq:rip_def}. For example, $\vc = \frac d n \1_n$ induces an empirical second moment matrix with uniform weights. After applying a Forster transform, the new dataset then exhibits desirable properties that are useful in downstream applications.

Notably, although the concepts of radial isotropic position and Forster transforms first arose in early work on algebraic geometry \cite{GelfandGMS87} and functional analysis \cite{Barthe98}, they have since enabled many surprising results in algorithms and complexity. For example, Forster transforms played a pivotal role in breakthroughs spanning disparate areas such as communication complexity \cite{Forster02}, subspace recovery \cite{HardtM13}, coding theory \cite{DvirSW14}, frame theory \cite{HamiltonM19}, active and noisy learning of halfspaces \cite{HopkinsKLM20, DiakonikolasKT21, DiakonikolasTK23}, and robust statistics \cite{Cherapanamjeri24}. 

\paragraph{Alternate characterizations of RIP.} In fact, many of these results \cite{Forster02, HardtM13, DvirSW14, HopkinsKLM20, DiakonikolasKT21, DiakonikolasTK23} leverage alternative characterizations of RIP, which themselves have powerful implications. We mention two here as they are relevant to our development. 

First, \cite{GelfandGMS87, Barthe98} (see also \cite{CarlenLL04, HopkinsKLM20}) give a tight characterization of when a $\vc$-Forster transform of $\ma$ exists. In brief, one exists if and only if $\vc$ belongs to the \emph{basis polytope} of the independence matroid induced by $A$ (see Propositions~\ref{prop:scaling_subspace} and~\ref{prop:scaling_polytope}). A more intuitive and equivalent way of phrasing this result is that every $k$-dimensional linear subspace $V \subseteq \R^d$ must have\begin{equation}\label{eq:light_subspace_intro}\sum_{\substack{i \in [n] \\ \va_i \in V}} \vc_i \le k.\end{equation}
The necessity of \eqref{eq:light_subspace_intro} is straightforward: if there exists a ``heavy subspace'' containing too many points (according to their weight by $\vc$), then any transformation of the form $\va_i \to (\mr \va_i)\norm{\mr \va_i}_2^{-1}$ retains the existence of such a heavy subspace. This in turn rules out the possibility of $\vc$-RIP, because \eqref{eq:rip_def} cannot hold on the heavy subspace, simply by a trace argument. Further developments by e.g., \cite{CarlenLL04, HopkinsKLM20} showed that this is in fact the only barrier to Forster transforms.

Another dual viewpoint on Forster transforms is from the perspective of scaling the dataset to induce certain \emph{leverage scores}. More precisely, \cite{DadushR24} observed that finding $\vs \in \R^{n}_{> 0}$ such that
\begin{equation}\label{eq:scaling_lev_intro}\vtau\Par{\ms \ma} = \vc,\text{ where } \ms \defeq \diag{\vs},\end{equation}
implies that $\mr = (\ma^\top \ms^2 \ma)^{-\half}$ is a $\vc$-Forster transform of $\ma$. We recall this result in Lemma~\ref{lem:lev_rip}. Here, $\vtau(\ma) \in \R^n$ denotes the \emph{leverage scores} of a full-rank matrix $\ma \in \R^{n \times d}$ (defined in \eqref{eq:leverage}), a standard notion of the relative importance of points in a dataset. Thus, while Forster transforms are \emph{right scalings} $\mr \in \R^{d \times d}$ putting $\ma$ in isotropic position, we can equivalently find a \emph{left scaling} $\vs \in \R^n_{> 0}$ that balances $\ma$'s rows to have target leverage scores $\vc$ of our choice.

\paragraph{Computing an approximate Forster transform.} The goal of our work is designing efficient algorithms for computing a $(\vc, \eps)$-Forster transform of $\ma \in \R^{n \times d}$, whenever one exists. This goal is inspired by advancements in the complexity of simpler, but related, dataset transformation problems called \emph{matrix balancing and scaling}, for which \cite{CohenMTV17, Allen-ZhuLOW17} achieved nearly-linear runtimes in well-conditioned regimes. Indeed, as the list of applications of radial isotropic position grows, so too does the importance of designing efficient algorithms for finding them.

Given the algorithmic significance of Forster transforms, it is perhaps surprising that investigations of their computational complexity are relatively nascent. Previous strategies for obtaining polynomial-time algorithms can largely be grouped under two categories.

\textit{Optimizing Barthe's objective.} In a seminal work on radial isotropic position \cite{Barthe98}, Barthe observed that the minimizer of the convex objective $f: \R^n \to \R$ defined as
\begin{equation}\label{eq:barthe_intro}f(\vt) \defeq -\inprod{\vc}{\vt} + \log\det\Par{\sum_{i \in [n]} \exp\Par{\vt_i} \va_i \va_i^\top},\end{equation}
induces a $\vc$-Forster transform whenever it exists. Concretely, letting $\vt^\star$ minimize $f$, we have that $\vs \defeq \exp(\vt^\star)$, where $\exp$ is applied entrywise, satisfies \eqref{eq:scaling_lev_intro} (cf.\ Proposition~\ref{prop:barthe}). Thus, to obtain a Forster transform it is enough to efficiently optimize \eqref{eq:barthe_intro}. 

This approach was followed by \cite{HardtM13} (see also discussion in \cite{HamiltonM19, Cherapanamjeri24}), who proceeded via cutting-plane methods (CPMs), and \cite{ArtsteinKS20}, who used first-order methods (e.g., gradient descent). 
However, it is somewhat challenging to quantify the accuracy needed in solving \eqref{eq:barthe_intro} to induce a $(\vc, \eps)$-Forster transform for $\eps > 0$, because Barthe's objective is not strongly convex. 
For example, combining Lemmas B.6, B.9 of \cite{HardtM13} with Corollary 4 of \cite{HamiltonM19} gives an estimate of $\approx \eps \exp(-nd)$ additive error sufficing. Our work drastically improves this estimate  (cf.\ Lemma~\ref{lem:termination}), showing it is enough to obtain an additive error that is polynomial in $\eps$, and $\min_{i \in [n]} \vc_i$. 

An optimistic bound on \cite{HardtM13}'s runtime scales as\footnote{We use $\omega < 2.372$ to denote the exponent of the square matrix multiplication runtime \cite{AlmanDWXXZ25}. For all notation used throughout the paper, see Section~\ref{ssec:notation}.} $\approx n^2 d^{\omega - 1} + n^3$ (using the state-of-the-art CPM \cite{JiangLSW20}), where additional $\poly(n, d)$ factors are saved using our improved error bounds. Incomparably, \cite{ArtsteinKS20} gave runtimes for first-order methods depending polynomially on either the inverse target accuracy $\frac 1 \eps$ (and hence precluding high-accuracy solutions), or the inverse strong convexity of Barthe's objective, which is data-dependent but can lose $\exp(d)$ factors or worse.

\textit{Iterative scaling methods.} Instead of optimizing Barthe's objective, \cite{DiakonikolasTK23, DadushR24} recently gave alternative approaches that either iteratively refine a right-scaling $\mr \in \R^{d \times d}$ to satisfy \eqref{eq:scale_rip_def}, or refine a left-scaling $\vs \in \R^n_{> 0}$ to satisfy \eqref{eq:scaling_lev_intro}. These algorithms have the advantage of running in \emph{strongly polynomial time}, i.e., the number of arithmetic operations needed only depends on $n$ and $\frac 1 \eps$, rather than problem conditioning notions such as bit complexity. Designing strongly polynomial time algorithms is an interesting and important goal in its own right. For instance, \cite{DiakonikolasTK23} was motivated by the connection of Forster transforms to learning halfspaces with noise \cite{DiakonikolasKT21}, a robust generalization of linear programming, which is a basic problem for which strongly polynomial time algorithms are unknown. In a different direction, \cite{DadushR24} showed that a strongly polynomial algorithm for matrix scaling by \cite{LinialSW00} could be adapted to Forster transforms.

Unfortunately, the resulting runtimes from these direct iterative methods that sidestep Barthe's objective are quite large. For example, \cite{DiakonikolasTK23} claim a runtime of at least $\approx n^5 d^{11} \eps^{-5}$, and a recent improvement in \cite{DadushR24} still requires at least $\approx n^4 d^{\omega - 1}\log(\frac 1 \eps)$ time. 

\textit{Outlook.} There are a few other approaches to polynomial-time computation of approximate Forster transforms based on more general formulations of the problem, see e.g., \cite{Allen-ZhuGLOW18, StraszakV19}. We discuss these algorithms in more detail in Section~\ref{ssec:related}, but note that they appear to lack explicit runtime bounds, and we believe they are subsumed by those described thus far.

In summary, existing methods for computing $(\vc, \eps)$-Forster transforms have runtimes at least $\approx n^2 d^{\omega -1 } + n^3$ (weakly polynomial) or $\approx n^4 d^{\omega - 1}$ (strongly polynomial). On the other hand, for related problems such as matrix scaling, 
near-optimal runtimes are known in well-conditioned regimes, via structured optimization methods that more faithfully capture the geometry of relevant objectives \cite{CohenMTV17, Allen-ZhuLOW17}.
This state of affairs prompts the natural question: can we obtain substantially faster algorithms for computing a $(\vc, \eps)$-Forster transform?

\subsection{Our results}\label{ssec:results}

Our main contribution is to design such algorithms, primarily specialized to two settings which we call the \emph{well-conditioned} and \emph{smoothed analysis} regimes. We note that the distinction between well-conditioned and poorly-conditioned instances is a common artifact of fast algorithms for scaling problems, see e.g., discussions in \cite{CohenMTV17, Allen-ZhuLOW17, BurgisserLNW20}. In particular, analogous works to ours for matrix scaling and balancing \cite{CohenMTV17, Allen-ZhuLOW17} obtain nearly-linear runtimes in well-conditioned regimes, and polynomial runtime improvements in others.

For a fixed pair $\ma \in \R^{n \times d}$ and $\vc \in (0, 1]^n$ satisfying $\norm{\vc}_1=d$, we use the following notion of conditioning for the associated problem of computing a $\vc$-Forster transform of $\ma$.

\begin{assumption}\label{assume:simplify}
	For $f$ defined in \eqref{eq:barthe_intro}, there is $\vt^\star \in \argmin_{\vt \in \R^n} f(\vt)$ satisfying $\norm{\vt^\star}_\infty \le \log(\kappa)$. 
\end{assumption}

To justify this, recall that $\vt^\star \in \argmin_{\vt \in \R^n} f(\vt)$ induces the optimal left scaling, in the sense of \eqref{eq:scaling_lev_intro}, via $\vs(\vt^\star) = \exp(\half \vt)$, entrywise.
Further, Barthe's objective is invariant to translations by $\1_n$:
\begin{equation}\label{eq:invariant_ones}
	\begin{aligned}
		f\Par{\vt + \alpha \1_n} &= -\inprod{\vc}{\vt + \alpha \1_n} + \log\det\Par{\mz\Par{\vt + \alpha\1_n}} \\
		&= -\inprod{\vc}{\vt} - \alpha d + \log\det\Par{\mz(\vt)} + \log\det\Par{\exp(\alpha)\id_d} = f(\vt).
	\end{aligned}
\end{equation}

Thus, we can always shift any minimizing $\vt^\star$ so that its extreme coordinates average to $0$, which achieves the tightest $\ell_\infty$ bound on $\vt^\star$ via shifts by $\1_n$. This shows that $\kappa$ in Assumption~\ref{assume:simplify} is the ratio of the largest and smallest entries of the optimal scaling $\vs \in \R^n_{>0}$ achieving \eqref{eq:scaling_lev_intro}.

\paragraph{Radial isotropic position.} We now state our main result on computing Forster transforms.

\begin{restatable}{theorem}{restatemain}\label{thm:main}
	Let $\ma \in \R^{n \times d}$, $\vc \in (0, 1]^n$ satisfy Assumption~\ref{assume:simplify}, and let $\delta, \eps \in (0, 1)$. There is an algorithm that computes $\mr$, a $(\vc, \eps)$-Forster transform of $\ma$, with probability $\ge 1 - \delta$, in time
	\[O\Par{nd^{\omega - 1}\log\Par{\kappa}\Par{\frac{n\log(\kappa)}{\delta\eps\vc_{\min}}}^{o(1)}}, \text{ where } \vc_{\min} \defeq \min_{i \in [n]} \vc_i.\]
\end{restatable}

In the well-conditioned regime where $\kappa = \poly(n)$, Theorem~\ref{thm:main} improves upon the state-of-the-art runtimes for radial isotropic position by a factor of $\approx \max(n, n^2 d^{1 - \omega})$, up to a subpolynomial overhead in problem parameters.\footnote{As discussed earlier, to our knowledge, even that CPMs \cite{HardtM13} obtain runtimes of $\approx n^2d^{\omega - 1} + n^3$ for well-conditioned instances was unknown previously. This is enabled by our improved error tolerance analysis in Lemma~\ref{lem:termination}.} Moreover, Theorem~\ref{thm:main}
approaches natural limits for computing Forster transforms. For example, using current techniques, it takes $\approx nd^{\omega - 1}$ time to perform basic relevant operations such as evaluating Barthe's objective \eqref{eq:barthe_intro}, or verifying that a given right scaling $\mr \in \R^{d \times d}$ or left scaling $\vs \in \R^n_{> 0}$ places a dataset in radial isotropic position.

Interestingly, we prove Theorem~\ref{thm:main} by adapting the box-constrained Newton's method of \cite{CohenMTV17, Allen-ZhuLOW17} to Barthe's objective, discussed further in Section~\ref{ssec:techniques}. In fact, such an approach had been considered previously by \cite{CelisKV20} for a more general problem of computing \emph{maximum-entropy distributions}. We discuss this related result in more detail in Section~\ref{ssec:related}, but briefly mention that based on the runtime analysis in \cite{CelisKV20}, the resulting complexity is significantly slower than CPMs, and we require several new structural observations and algorithmic insights to improve the efficiency of this approach. Indeed, as one example, it is perhaps surprising that Theorem~\ref{thm:main}'s runtime depends linearly on $n$, given that it is a second-order method: merely writing down the Hessian of Barthe's objective takes $n^2$ time, which dominates the runtime of Theorem~\ref{thm:main} for $n \gg d$. 

We discuss the key algorithmic innovation that enables our result subsequently, but mention here that its use leads to the subpolynomial overheads in Theorem~\ref{thm:main}. By using explicit Hessian evaluations, we obtain an alternate runtime of
\[O\Par{n^2 d^{\omega - 2} \log\Par{\kappa}\polylog\Par{\frac{n\log(\kappa)}{\delta \eps \vc_{\min}}}},\]
as described more formally in Lemma~\ref{lem:grad_compute} and Remark~\ref{rem:log_delta_eps}, which yields (improved) polylogarithmic dependences on $\frac 1 \delta$, $\frac 1 \eps$, and $\frac 1 {\vc_{\min}}$, at the cost of a
multiplicative overhead of $\approx \frac n d$.

\paragraph{Sparsification via matrix-vector products.} Our fastest runtimes are obtained by using a technical tool of potential independent interest that we develop. To explain its relevance to our setting, while the Hessian of Barthe's objective $\nabla^2 f$ is $n \times n$ and fully dense (a formula is given in Fact~\ref{fact:barthe_derivs}), its structure is appealing in several regards. While we do not know how to compute $\nabla^2 f$ faster than in $ \approx n^2 d^{\omega - 2}$ time, we can access it via matrix-vector products in $O(nd^{\omega - 1})$ time (cf.\ Lemma~\ref{lem:grad_compute}). In addition, $\nabla^2 f$ is actually a \emph{graph Laplacian}, i.e., it belongs to a family of matrices that have enabled many powerful algorithmic primitives, such as \emph{spectral sparsification}. For example, breakthroughs by \cite{SpielmanS11, SpielmanT14} show that any $n \times n$ graph Laplacian $\ml$ admits constant-factor spectral approximations with only $\approx n$ nonzero entries.

In this work, we add a new primitive to the graph Laplacian toolkit. We consider the following problem, which to our knowledge has not been explicitly studied before: given an (implicit) Laplacian $\ml$ accessible only via a matrix-vector product oracle, how many queries are needed to produce an (explicit) spectral sparsifier of $\ml$? The sparsifier can then be used as a preconditioner, enabling faster second-order methods. Our main result to this end is the following.

\begin{restatable}{theorem}{restateimplicitsparsify}\label{thm:implicit_sparsify}
	Let $\ml$ be an $n \times n$ graph Laplacian, and let $\oracle: \R^n \to \R^n$ be an oracle that returns $\ml \vv$ on input $\vv \in \R^n$. Let $\delta \in (0, 1)$, $\Delta \in (0, \Tr(\ml))$, and let $\mproj \defeq \id_n - \frac 1 n \1_n \1_n^\top$ be the projection matrix to the subspace of $\R^n$ orthogonal to $\1_n$. There is an algorithm that takes as inputs $(\oracle, \delta, \Delta)$ and with probability $\ge 1 - \delta$, it returns $\tml$, an $n \times n$ graph Laplacian satisfying
	\begin{equation}\label{eq:sparse_quality}\begin{gathered}\ml + \Delta \mproj \preceq \tml \preceq \Par{\frac{n\Tr(\ml)}{\Delta \delta}}^{o(1)}\Par{\ml + \Delta \mproj},\; \nnz(\tml) = n \cdot \Par{\frac{n\Tr(\ml)}{\Delta\delta}}^{o(1)},\end{gathered}\end{equation}
	using $(\frac{n\Tr(\ml)}{\Delta\delta})^{o(1)} $ queries to $\oracle$, and $n \cdot (\frac{n\Tr(\ml)}{\Delta\delta})^{o(1)} $ additional time.
\end{restatable}

For $\delta = \poly(\frac 1 n)$ and $\poly(n)$-well conditioned graph Laplacians, Theorem~\ref{thm:implicit_sparsify} produces a spectral sparsifier of a Laplacian $\ml$ using $n^{o(1)}$ matrix-vector products and $n^{1 + o(1)}$ additional time. The approximation quality of the sparsifier is somewhat poor, i.e., $n^{o(1)}$, but in algorithmic contexts (such as that of Theorem~\ref{thm:main}), this is sufficient for use as a low-overhead preconditioner. 

We believe Theorem~\ref{thm:implicit_sparsify} may be of independent interest to the graph algorithms and numerical linear algebra communities, as it enhances the flexibility of existing Laplacian-based tools; we discuss its connection to known results in more depth in Section~\ref{ssec:related}. We are optimistic that its use can extend the reach of fast second-order methods for combinatorially-structured optimization problems.

\paragraph{Assumption~\ref{assume:simplify} in the smoothed regime.} Our third main contribution is to provide explicit bounds on the problem conditioning $\kappa$ in Assumption~\ref{assume:nd}, for ``beyond worst-case'' inputs $\ma$. We specialize our result to the \emph{smoothed analysis} setting, a well-established paradigm for beyond worst-case analysis in the theoretical computer science community \cite{SpielmanT04, Roughgarden20}. In our smoothed setting, we perturb entries of our input by Gaussian noise at noise level $\sig > 0$. This is a standard smoothed matrix model used in the study of linear programming algorithms \cite{SpielmanT04, SankarST06}. 

Here, we state the basic variant of our conditioning bound in the smoothed analysis regime.

\begin{restatable}{theorem}{restatesmoothed}\label{thm:smoothed}
	Let $\ma \in \R^{n \times d}$ have rows $\{\va_i\}_{i \in [n]}$ such that $\norm{\va_i}_2 = 1$ for all $i \in [n]$, let $\vc \defeq \frac d n \1_n$, let $\delta \in (0, 1)$, and let $\sig \in (0, \frac{\delta}{10nd})$. Let $\tma \defeq \ma + \mg$, where $\mg \in \R^{n \times d}$ has entries $\simiid \Nor(0, \sig^2)$. Then with probability $\ge 1 - \delta$, if $n > Cd$ where $C$ is any constant larger than $1$, Assumption~\ref{assume:simplify} holds for Barthe's objective $f$ defined with respect to $(\tma, \vc)$, where
	\[\log(\kappa) = O\Par{d\log\Par{\frac 1 {\sig}}}.\]
\end{restatable} 

That is, $\ma$ in Theorem~\ref{thm:smoothed} is a ``base worst-case instance'' that is smoothed into a more typical instance $\tma$, which our conditioning bound of $\kappa \approx (\frac 1 \sig)^{O(d)}$ applies to. 

The assumption in Theorem~\ref{thm:smoothed} that $\ma$ has unit norm rows is relatively mild; rescaling rows does not affect the (base) Forster transform problem, and our result still applies if row norms are in a $\poly(n)$ multiplicative range. Further, while Theorem~\ref{thm:smoothed} is stated for uniform marginals $\vc = \frac d n \1_n$, we show in Corollary~\ref{cor:nonuniform} that as long as the marginals $\vc$ are bounded away from $1$ entrywise by a constant, the conditioning estimate in Theorem~\ref{thm:smoothed} still holds for sufficiently large $n$. The requirement that $n > Cd$ for $C > 1$ is a minor bottleneck of our approach, discussed in Remark~\ref{rem:superlinear}.

We are aware of few explicit conditioning bounds for Forster transforms such as Theorem~\ref{thm:smoothed}, so we hope it (and techniques used in establishing it) become useful in future studies. Among conditioning bounds that exist presently, Lemma B.6 of \cite{HardtM13} (cf.\ discussion in Corollary 4, \cite{HamiltonM19}) shows that for $\ma$ with rows in \emph{general position}, we have $\log(\kappa) = O(n\log(\frac 1 D))$, where $D$ is the smallest determinant of a nonsingular $d \times d$ submatrix of $\ma$. In particular, $D$ can be inverse-exponential in $d$ (or worse) for poorly-behaved instances. A crude lower bound of $D \gtrsim \exp(-d^3)$ was provided in \cite{Cherapanamjeri24} for essentially the smoothed model we consider in Theorem~\ref{thm:smoothed}.

On the other hand, \cite{DiakonikolasTK23, DadushR24}, who respectively design strongly polynomial methods for iteratively updating an approximate Forster transform $\mr \in \R^{d \times d}$ or dual scaling $\vs \in \R^n_{> 0}$, bound related conditioning quantities. Both papers contain results (cf.\ Section 5, \cite{DiakonikolasTK23} and Section 4, \cite{DadushR24}) showing that any iterate has a ``nearby'' iterate in bounded precision, that does not significantly affect some potential function of interest. These results do not appear to directly have implications for Assumption~\ref{assume:simplify}, and provide rather large bounds on the bit complexity (focusing on worst-case instances). Nonetheless, exploring connections in future work could be fruitful.

Most relatedly, Theorem 1.5 of \cite{ArtsteinKS20} proves that for target marginals $\vc$ that are ``deep'' inside the basis polytope for independent sets of $\ma$'s rows, $\log(\kappa) \lesssim d$. However, \cite{ArtsteinKS20} does not give estimates on the deepness of marginals in concrete models, and indeed our approach to proving Theorem~\ref{thm:smoothed} is to provide such explicit bounds in the smoothed analysis regime.

Directly combining Theorems~\ref{thm:main} and~\ref{thm:smoothed} shows that for smoothed instances, the complexity of computing an approximate Forster transform is at most $\approx nd^{\omega}$, up to a subpolynomial factor. While it is worse than our well-conditioned runtime, our method in Theorem~\ref{thm:main} still improves upon state-of-the-art algorithms based on CPMs by a factor of $\approx \max(\frac n d, n^3 d^{-\omega})$ in this regime.

\paragraph{Computational model.} This paper works in the real RAM model, where we bound the number of basic arithmetic operations. Prior work on optimizing Barthe's objective \cite{HardtM13, ArtsteinKS20} also worked in this model, and we give a comparison on how these results are affected under finite-precision arithmetic in Appendix~\ref{app:numerical}. A more detailed investigation of the numerical stability of Forster transforms is an important direction for future work, but is outside our scope.

\subsection{Our techniques}\label{ssec:techniques}

In this section, we overview our approaches to proving Theorems~\ref{thm:main},~\ref{thm:implicit_sparsify}, and~\ref{thm:smoothed}.

\paragraph{Optimizing Barthe's objective.} Our algorithm for optimizing Barthe's objective \eqref{eq:barthe_intro} is a variant of the \emph{box-constrained Newton's method} of \cite{CohenMTV17, Allen-ZhuLOW17}, originally developed for approximate matrix scaling and balancing. We were inspired to use this tool by noticing similarities between the derivative structure of Barthe's objective (Fact~\ref{fact:barthe_derivs}) and the \emph{softmax} function, which can be viewed as the one-dimensional case of Barthe's objective. Previously, the softmax function was known to be \emph{Hessian stable} in the $\ell_\infty$ norm (Definition~\ref{def:hessian_stable}), enabling local optimization oracles that can be implemented via Newton's method \cite{CarmonJJJLST20} over $\ell_2$ or $\ell_\infty$ norm balls. 

It is much more challenging to prove that Barthe's objective is Hessian stable, as the proof in \cite{CarmonJJJLST20} does not naturally extend to non-commuting variables. Nonetheless, we give a different proof inspired by Kadison's inequality in operator algebra \cite{Kadison52} to establish Hessian stability of Barthe's objective in Section~\ref{ssec:barthe_stable}. Our Hessian stability bound directly improves the best previously known in the literature, due to Lemma D.2 in \cite{CelisKV20}, by a factor of $n$, making it dimension-independent. This reflects in a multiplicative $O(n)$ savings in our final runtime.

We complement this result in Section~\ref{ssec:terminate} with bounds on the additive error on Barthe's objective required to obtain a $(\vc, \eps)$-Forster transform, for a tolerance $\eps > 0$. By using the leverage score characterization \eqref{eq:scaling_lev_intro} of exact Forster transforms, and performing a local perturbation analysis at the optimizer, we show $\poly(\eps, \min_{i \in [n]} \vc_i)$ error suffices. This significantly sharpens prior error  bounds from \cite{HardtM13, Cherapanamjeri24}, which scaled exponentially in a polynomial of the problem parameters.

With these stability bounds in place, the rest of Section~\ref{sec:newton} makes small modifications to the \cite{CohenMTV17} analysis. We show that by using fast matrix multiplication, each Hessian can be computed in $\approx n^2 d^{\omega - 2}$ time (Lemma~\ref{lem:grad_compute}), and that box-constrained Newton steps can be efficiently implemented using the constrained optimization methods from \cite{ChenPW21}. This gives our basic runtime in Remark~\ref{rem:log_delta_eps}, which is sped up via sparsifiers provided by Theorem~\ref{thm:implicit_sparsify}.

\paragraph{Implicit sparsification.} We next describe our approach to proving Theorem~\ref{thm:implicit_sparsify}, a reduction from $n^{o(1)}$-approximate sparsification to $n^{o(1)}$ accesses of an implicit Laplacian via matrix-vector products. Our algorithm is an adaptation and extension of previous work \cite{JambulapatiLMSST23}, which gave an algorithm for recovering a $(1+\epsilon)$-spectral sparsifier of a graph $G$ on $n$ vertices, using $O(\polylog(n))$ matrix-vector products with $G$'s (pseudo)inverse Laplacian matrix, and $\approx n^2$ extra time. 

The approach taken by \cite{JambulapatiLMSST23} was to reduce the problem to a more general setting known as \emph{matrix dictionary recovery}. Here, we are given matrix-vector access to $\mb$ (e.g., an implicit Laplacian), and a dictionary of $\{ \ma_i \}_{i \in [m]} \in \PSD^{n \times n}$ with the promise that there exist weights $\vw^\star \in \R^n_{\ge 0}$ such that $\sum_{i \in [m]} \vw^\star_i \ma_i = \mb$. Our goal is to compute weights $\vw \in \R^m_{\ge 0}$ such that $\mb \preceq \sum_{i \in [m]} \vw_i \ma_i \preceq C \mb$ for some approximation factor $C > 1$. This was done in \cite{JambulapatiLMSST23} by reducing the two-sided recovery problem to a small number of one-sided \emph{packing} semidefinite programs (SDPs):
\begin{equation}\label{eq:packing_intro}
	\min_{\substack{\vx \geq 0 \\ \sum_{i \in [m]} \vx_i \ma_i \preceq \mb}} \vc^\top \vx.
\end{equation}
Then, \cite{JambulapatiLMSST23} applies packing SDP solvers from the literature \cite{Allen-ZhuLO16, PengTZ16, JambulapatiLT20} to solve these problems. Combining this with a homotopy scheme yields their full algorithm.

A key limitation of this approach in the graph setting is simply the size of the matrix dictionary: in particular, we have one entry per candidate edge.
For graphs on $n$ vertices, maintaining an internal representation requires manipulating potentially-dense graphs on $m \approx n^2$ edges. This appears hard to bypass, without a priori information on which edges exist in our implicit Laplacian.

We circumvent this issue by enforcing that our intermediate Laplacians, while remaining dense combinations of our matrix dictionary, have a greater deal of structure that enables key primitives such as \emph{sparsification} and \emph{matrix-vector multiplication} in $n^{1 + o(1)}$ time. We introduce a family of such compatible graphs, with a representation property we call \emph{sum-of-cliques} (Definition~\ref{def:soc}). The key challenge is designing solvers for packing SDPs \eqref{eq:packing_intro} that take few iterations, while maintaining that the solver's iterates are compatible with our sum-of-cliques machinery.

We make a crucial observation that standard matrix multiplicative weights (MMW)-based solvers for packing SDPs require gradient computations of the form
\begin{equation}\label{eq:grad_intro}\inprod{\mm}{\ml_e} = \norm{\mm^{\half }(\ve_u - \ve_v)}_2^2 \text{ for all } e = (u, v) \in [n] \times [n], \end{equation}
where $\ml_e \defeq (\ve_u - \ve_v)(\ve_u - \ve_v)^\top$ is an edge Laplacian, and $\mm$ is a response matrix given by the MMW updates. These gradients are then used to update a current iterate.

By noticing that gradient entries \eqref{eq:grad_intro} can be viewed as distance computations between columns of $\mm^{\half}$, we apply metric embedding and hashing tools to coarsely discretize our gradients in a way that induces appropriate clique structures. In Definition~\ref{def:step}, we further isolate sufficient conditions for this coarse discretization scheme to implement a long-step packing SDP solver, while maintaining that iterates are unions of $n^{o(1)}$ sums-of-cliques.  We combine these pieces in Section~\ref{ssec:packing} to give our main algorithmic innovation: an $n^{o(1)}$-approximate solver for \eqref{eq:packing_intro} that returns a Laplacian sparsifiable in $n^{1 + o(1)}$ time, when the constraint $\mb$ is an unknown graph Laplacian. 

We finally show in the rest of Section~\ref{sec:laplacian} that the remaining pieces of \cite{JambulapatiLMSST23}, i.e., the two-sided to one-sided reduction and the homotopy method, are compatible with our new packing solver.

\paragraph{Smoothed analysis of conditioning.} To prove Theorem~\ref{thm:smoothed}, our main result in the smoothed analysis setting, in Definition~\ref{def:deep} we first extend an approach of \cite{ArtsteinKS20} that defines a notion of \emph{deepness} of marginal vectors $\vc$ inside the basis polytope induced by $\ma$'s independent row subsets. As we recall in Section~\ref{ssec:deep}, \cite{ArtsteinKS20} argues that if $\vc$ has deepness of $\eta = \Omega(1)$, then we can obtain a conditioning bound of $\log(\kappa) \approx d$ in Assumption~\ref{assume:simplify}. In the case of $\vc = \frac d n \1_n$, this roughly translates to a robust variant of \eqref{eq:light_subspace_intro} that says: for all subspaces $E \subseteq \R^d$ of dimension $k$, at most a $\approx \frac k d$ fraction of $\ma$'s rows (after smoothing by Gaussian noise) should lie at distance $\poly(\frac 1 n)$ from $E$. The rest of Section~\ref{sec:smoothed} proves this deepness result for smoothed matrices $\tma = \ma + \mg$. 

The key challenge is to avoid union bounding over a net of all possible subspaces $E$; for $\dim(E) = \Theta(d)$, this na\"ive approach would require taking $n \gtrsim d^2$ samples, as nets of $\Theta(d)$-dimensional subspaces have cardinality $\approx \exp(d^2)$. We instead show in Lemma~\ref{lem:small_sv} that deepness is implied by submatrices of $\tma$ of appropriate size (dictated by the subspace dimension $k$) having at least $k + 1$ large singular values, allowing us to apply union bounds to a smaller number of \emph{data-dependent} subspaces. We combine this observation with singular value estimates from the random matrix theory literature to prove Theorem~\ref{thm:smoothed}. Our argument requires some casework on the subspace dimension; we handle wide and near-square submatrices in Section~\ref{ssec:wide} and tall submatrices in Section~\ref{ssec:tall}.

\subsection{Related work}\label{ssec:related}

\paragraph{Forster transforms via maximum entropy.} An alternative characterization of Forster transforms was followed by \cite{StraszakV19}, who studied certain \emph{maximum-entropy distribution} representations of specified marginals $\vc$ with respect to an index set $\calS$, which we briefly explain for context. In our setting of finding a $\vc$-Forster transform of $\ma \in \R^{n \times d}$, the index set $\calS$ consists of all $S \subseteq [n]$ with $|S| = d$, and the underlying $\pi(S)$ is the \emph{determinantal measure} with $\pi(S) \propto \det([\ma^\top \ma]_{S : S})$. Then, Section 8.3 of \cite{StraszakV19} applies the Cauchy-Binet formula to show that
\begin{equation}\label{eq:max_entropy}
	\begin{aligned}
		\min_{\vt \in \R^n} -\inprod{\vc}{\vt} + \log\det\Par{\sum_{i \in [n]} \exp\Par{\vt_i} \va_i \va_i^\top} &= \min_{\vt \in \R^n} \log\Par{\sum_{S \in \calS} \exp\Par{\inprod{\1_S - \vc}{\vt}} \det\Par{\Brack{\ma^\top\ma}_{S:S}}} \\
		&= \min_{\vt \in \R^n} \log\Par{\sum_{S \in \calS} \pi\Par{S}\exp\Par{\inprod{\1_S - \vc}{\vt}} } + Z, \\
		\text{where } Z &\defeq \log\Par{\sum_{S \in \calS} \det\Par{\Brack{\ma^\top \ma}_{S:S}}},
	\end{aligned}
\end{equation}
and the starting expression is Barthe's objective \eqref{eq:barthe_intro}. This shows that computing Forster transforms falls within the framework of Section 7 in \cite{StraszakV19}, which exactly gives polynomial-time algorithms for optimizing functions in the form of the ending expression above. The runtime of \cite{StraszakV19} is not explicit, and we believe it is significantly slower than more direct approaches, e.g., CPMs \cite{HardtM13}.
Similarly, \cite{BurgisserLNW20} develop interior-point methods for solving maximum entropy optimization problems of the form \eqref{eq:max_entropy}, but with runtimes scaling as $\poly(|\calS|)$, which in our setting is $\approx n^d$. 

More closely related to our work, \cite{CelisKV20} also consider a box-constrained Newton's method for solving maximum-entropy distributions over the hypercube $\{0, 1\}^n$ (which generalizes Barthe's objective, due to \eqref{eq:max_entropy}). They prove that in this setting, maximum-entropy optimization problems  \eqref{eq:max_entropy} are $(r, O(rn))$-Hessian stable for all $r > 0$ with respect to $\norm{\cdot}_\infty$, which is a factor $n$ worse than Proposition~\ref{prop:hessian_stable_barthe}. Moreover, they do not provide an explicit runtime for solving box-constrained subproblems (relying on black-box quadratic programming solvers), and do not analyze the required accuracy for termination (e.g., Lemma~\ref{lem:termination}). We believe their approach, as implemented for Barthe's objective, results in a much slower algorithm than CPMs; indeed, their stated runtime in a preprint \cite{CelisKYV19} scales at least as $\approx n^{4.5}$. To obtain our runtime improvements, we require several novel structural insights on Barthe's objective and its interaction with optimization methods, as well as new algorithmic primitives (e.g., Theorem~\ref{thm:implicit_sparsify}) capable of harnessing said structure.

\paragraph{Forster transforms via operator scaling.} In another direction, \cite{GargGOW17} discovered a nontrivial connection between computing Forster transforms (phrased in an equivalent way of computing Brascamp-Lieb constants, see Proposition 1.8, \cite{GargGOW17}) and a related problem called \emph{operator scaling}. This implies that polynomial-time algorithms for operator scaling \cite{GargGOW16, GargGOW17, IvanyosQS17, Allen-ZhuGLOW18} apply to our problem as well. However, none of the aforementioned algorithms for operator scaling have an explicitly specified polynomial, and a crude analysis results in fairly substantial blowups. Also, some of these algorithms have more explicit and stronger variants analyzed in \cite{DiakonikolasTK23, DadushR24}, so  we believe they are subsumed by our existing discussion.

More generally, there is an active body of research on generalizations of operator scaling and Forster transforms \cite{GargGOW17, Franks18, BurgisserFGOWW18, BurgisserFGOWW19}, for which several state-of-the-art results are via variants of Newton's method, broadly defined. It would be interesting to explore if the ideas developed in this paper could extend to those settings as well.

\paragraph{Reductions between graph primitives.} Our work on implicit sparsification (Theorem~\ref{thm:implicit_sparsify}) fits into a line of work that aims to characterize which fast (i.e., $n^{1 + o(1)}$-time) graph primitives imply others by reduction. This theme was explicitly considered by \cite{AlmanCSS20} (see also related work by \cite{Quanrud21}), who studied these primitives for graphs implicitly defined by low-dimensional kernels. Among the three primitives of (1) fast matrix-vector multiplication, (2) fast spectral sparsification, and (3) fast Laplacian system solving, it was known previously that (3) reduces to (1) and (2) \cite{SpielmanT04}, and that (1) reduces to (2) and (3) \cite{AlmanCSS20}. Our work makes progress on this reduction landscape, as it shows (2) reduces to (1) (and hence, (3) also reduces to (1)). We mention that Theorem 5 in \cite{JambulapatiLMSST23} gives a related, but slower, $\approx n^2$-time reduction from (2) to (3).

Our work is also thematically connected to prior work on spectral sparsification under weak graph access, e.g., in streaming and dynamic settings \cite{KapralovLMMS17, AbrahamDKKP16}. Specifically, several of the rounding-via-sketching tools used to prove Theorem~\ref{thm:implicit_sparsify} are inspired by \cite{KapralovMMMNST20}. Their result is incomparable to ours, as we are unable to directly access a sketch of the graph, so we instead use these tools to speed up an optimization method rather than identify the sparsifier in one shot.
\section{Preliminaries}\label{sec:prelims}

In Section~\ref{ssec:notation}, we give notation used throughout the paper, as well as several key linear algebraic definitions. In Section~\ref{ssec:rip}, we provide preliminaries on the radial isotropic position scaling problem that we study, as well as Barthe's objective \cite{Barthe98} used in computing Forster transforms.

\subsection{Notation}\label{ssec:notation}

\paragraph{General notation.} We denote matrices in capital boldface and vectors in lowercase boldface throughout. By default, vectors are $d \times 1$ matrices. We let $\0_d$ and $\1_d$ denote the all-zeroes and all-ones vectors in $\R^d$, and $\0_{m \times n}$ is the all-zeroes $m \times n$ matrix. If $S \subseteq [d]$ and a dimension $d$ is clear from context, we let $\1_S$ be the $0$-$1$ indicator vector of $S$. We let $\vv \circ \vw$ denote the entrywise product of vectors $\vv, \vw$ of the same dimension. For $n \in \N$ we define $[n] \defeq \{i \in \N \mid i \le n\}$. For $p \ge 1$ and $p = \infty$ we let $\norm{\cdot}_p$ denote the $\ell_p$ norm of a vector argument, as well as the Schatten-$p$ norm of a matrix argument. For $\bvx \in \R^d$ and $r > 0$, we let $\ball_p(\bvx, r) \defeq \{\vx \in \R^d \mid \|\bvx - \vx\|_p \le r\}$ denote the $\ell_p$ norm ball of radius $r$ centered at $\bvx$; if $\bvx$ is unspecified then $\bvx = \0_d$ by default.
We let $\Nor(\vmu, \msig)$ denote the multivariate Gaussian with specified mean and covariance. When $d$ is clear from context, $\ve_i \in \R^d$ for $i \in [d]$ denotes the $i^{\text{th}}$ standard basis vector. For $S \subset \R^d$, we use $\conv(S)$ to denote the convex hull of $S$. For an event $\event$, we let $\ind_\event$ denote the $0$-$1$ indicator variable. 

For a multilinear form $\mt \in \R^{d_1 \times \ldots \times d_\ell}$ and vectors $\vv_i \in \R^{d_i}$ for all $i \in [\ell]$, we denote 
\[\mt\Brack{\vv_1, \ldots, \vv_\ell} \defeq \sum_{\substack{i_1 \in [d_1] \\ \ldots \\ i_\ell \in [d_\ell]}} \mt_{i_1\ldots i_\ell} [\vv_1]_{i_1} \ldots [\vv_\ell]_{i_\ell}.\]
For example, when $\mt$ is an $n \times d$ matrix, $\mt[\vu,\vv] = \sum_{(i, j) \in [n] \times [d]} \mt_{ij}\vu_i\vv_j = \vu^\top \mt \vv$.

\paragraph{Matrices.} The $d \times d$ identity matrix is denoted $\id_d$. The span of a set of vectors $\{\va_i\}_{i \in [n]}$ is denoted $\Span(\{\va_i\}_{i \in [n]})$; when $\ma$ is a matrix, we overload $\Span(\ma)$ to mean the span of its columns. We similarly use $\rank(\{\va_i\}_{i \in [n]})$, $\rank(\ma)$ to denote the dimension of the aforementioned subspaces. 

We denote the $j^{\text{th}}$ column of $\ma \in \R^{n \times d}$ by $\ma_{:j} \in \R^n$ for all $j \in [d]$, and we similarly denote the $i^{\text{th}}$ row of $\ma$ (viewed as a column vector) by $\ma_{i:} \in \R^{d}$ for all $i \in [n]$. For row and column subsets $S \subseteq [n]$, $T \subseteq [d]$, the corresponding submatrix is denoted $\ma_{S:T}$; $S = [n]$ and $T = [d]$ by default if excluded. We denote the Frobenius and ($2 \to 2$) operator norms of a matrix argument by $\normf{\cdot}$ and $\normop{\cdot}$. For a vector $\vv \in \R^d$ we let $\diag{\vv}$ be the $d \times d$ diagonal matrix with $\vv$ along the diagonal. We use $\nnz(\mm)$ to denote the number of nonzero entries of a matrix $\mm$, and we use $\tmv(\mm)$ to mean the time required to compute $\mm \vv$ for an arbitrary vector $\vv$ of appropriate dimension. We denote the projection matrix onto a subspace $E \subseteq \R^d$ by $\mproj_E \in \PSD^{d \times d}$.

We denote the set of symmetric $d \times d$ matrices by $\Sym^{d \times d}$, and the positive semidefinite (PSD) cone by $\PSD^{d \times d} \subset \Sym^{d \times d}$. We equip $\PSD^{d \times d}$ with the Loewner partial ordering $\preceq$.  We define the induced seminorm of $\mm \in \PSD^{d \times d}$ by $\norms{\vv}_{\mm}^2 \defeq \vv^\top \mm \vv$. For $\mm \in \Sym^{d \times d}$ we let $\mm^\dagger$ denote its pseudoinverse, which satisfies $\mm\mm^\dagger = \mm^\dagger \mm$ is the projection matrix onto $\Span(\mm)$. The eigenvalues of $\mm \in \Sym^{d \times d}$ are denoted $\vlam(\mm) \in \R^d$, where our convention is to order $\vlam_1(\mm) \ge \vlam_2(\mm) \ge \ldots \ge \vlam_d(\mm)$. We similarly denote the (monotone nonincreasing) singular values of $\ma \in \R^{n \times d}$ by $\vsig(\ma) \in \R^{\min(n, d)}$.

The trace of $\mm \in \Sym^{d \times d}$ is denoted $\Tr(\mm)$. For $\ma, \mb \in \R^{n \times d}$, we define the matrix inner product by $\inprod{\ma}{\mb} \defeq \Tr(\ma^\top \mb) = \sum_{(i, j) \in [n] \times [d]} \ma_{ij}\mb_{ij}$. We use the following notion of multiplicative approximation between PSD matrices: for $\ma, \mb \in \PSD^{d \times d}$ and $\eps > 0$, we write
\begin{equation}\label{eq:approx_mat}\ma \approx_\eps \mb \iff \exp\Par{-\eps} \mb \preceq \ma \preceq \exp\Par{\eps}\mb.\end{equation}

For nonnegative scalars $a, b \in \R_{\ge 0}$ we write $a \approx_\eps b$ to be the $1$-d specialization of \eqref{eq:approx_mat}, and $\vu \approx_\eps \vv$ is an entrywise definition for nonnegative vectors $\vu, \vv \in \R^d_{\ge 0}$.

We define the \emph{leverage scores} of a matrix $\ma \in \R^{n \times d}$ with rows $\{\va_i^\top\}_{i \in [n]}$ by
\begin{equation}\label{eq:leverage}\vtau_i(\ma) \defeq \va_i^\top \Par{\ma^\top \ma}^\dagger \va_i.\end{equation}
Leverage scores are a common measure of the relative importance of rows for preserving spectral information in numerical linear algebra, and are a crucial concept used throughout this paper. We summarize several basic facts about leverage scores here, see e.g., \cite{CohenLMMPS15} for proofs.

\begin{fact}\label{fact:lev}
For all $\ma \in \R^{n \times d}$, we have $\vtau(\ma) \in [0, 1]^{n}$, and $\sum_{i \in [n]} \vtau_i(\ma) = \rank(\ma)$.
\end{fact}

Finally, we let $\omega < 2.372$ denote the exponent of the square matrix multiplication runtime \cite{AlmanDWXXZ25}.

\paragraph{Optimization.} For $k$-times differentiable $f: \R^n \to \R$, we let $\nabla^k f$ denote the $k^{\text{th}}$ derivative tensor of $f$, e.g., $\nabla f$ and $\nabla^2 f$ are the gradient and Hessian of $f$. We use the following notion of multiplicative stability for analyzing Newton's method, patterned off \cite{CohenMTV17, KarimireddySS18, CarmonJJJLST20}.

\begin{definition}[Hessian stability]\label{def:hessian_stable}
We say that twice-differentiable $f: \R^n \to \R$ is $(r, \eps)$-\emph{Hessian stable with respect to norm $\norm{\cdot}$} if for all $\vx, \vy \in \R^n$ with $\norm{\vx - \vy} \le r$, we have following \eqref{eq:approx_mat} that
\[\nabla^2 f(\vx) \approx_\eps \nabla^2 f(\vy).\]
\end{definition}

We specifically will use the $\norm{\cdot}_\infty$ case of Definition~\ref{def:hessian_stable} to design our algorithms, via toolkits provided by \cite{CohenMTV17}, who called this property ``second-order robustness,'' and \cite{ChenPW21}.

\subsection{Radial isotropic position}\label{ssec:rip}

Our main contribution is an algorithm to scale a feasible matrix $\ma \in \R^{n \times d}$ to approximately satisfy a strong linear algebraic condition known as \emph{radial isotropic position} (Definition~\ref{def:rip}). This condition was introduced to the theoretical computer science community by \cite{Barthe98} to study inequalities in functional analysis, but a similar definition arose earlier in algebraic geometry \cite{GelfandGMS87}. 

To briefly demystify Definition~\ref{def:rip}, let $\ma \in \R^{n \times d}$ with rows $\{\va_i^\top\}_{i \in [n]}$ have $\rank(\ma) = d$ (so $n \ge d$). Then, it is well-known that $\ma$ can be scaled by invertible $\mr \in \R^{d \times d}$ so that 
\begin{equation}\label{eq:c_isotropy}\mr \ma^\top \diag{\vc} \ma \mr^\top = \sum_{i \in [n]} \vc_i \Par{\mr \va_i}\Par{\mr \va_i}^\top = \id_d. \end{equation}
Indeed, choosing $\mr = (\ma^\top \diag{\vc} \ma)^{-1/2}$ suffices. The condition \eqref{eq:c_isotropy} is sometimes referred to as being scaled to be in $\vc$-isotropic position, and there are natural $\eps$-approximate generalizations.

Similarly, as long as all $\va_i \neq \0_d$, there is a diagonal scaling $\ms$ so that $\ms\ma$ has unit-norm rows: let 
\begin{equation}\label{eq:radial}\ms = \diag{\vs} \text{ where } \vs_i = \frac 1 {\norm{\va_i}_2} \text{ for all } i \in [n] \implies \norm{\Brack{\ms \ma}_{i:}}_2 = 1 \text{ for all } i \in [n]. \end{equation}
Each of the transformations \eqref{eq:c_isotropy} and \eqref{eq:radial} is used in many applications to improve the regularity of a point set given by viewing the rows of $\ma$ as points in $\R^d$. The purpose of $\vc$-radial isotropic position (Definition~\ref{def:rip}) is to give a Forster transform matrix $\mr \in \R^{d \times d}$ inducing a scaling $\ms = \diag{\vs}$ via $\vs_i^{-1} = \norm{\mr \va_i}_2$ for all $i \in [n]$, such that the left-and-right scaled matrix $\ms \ma \mr^\top$ simultaneously has unit-norm rows, and is in $\vc$-isotropic position. Our goal is to efficiently approximate $\mr$.

\paragraph{Existence of Forster transform.} Not all point sets admit a $\vc$-Forster transform. A sequence of works \cite{GelfandGMS87, Barthe98, CarlenLL04, DvirSW14, HopkinsKLM20} gave two useful characterizations of feasibility.

\begin{proposition}[Lemma 4.19, \cite{HopkinsKLM20}]\label{prop:scaling_subspace}
Given a point set $A \defeq \{\va_i\}_{i \in [n]} \subset \R^d$ and $\vc \in (0, 1]^n$ satisfying $\norm{\vc}_1 = d$, the following conditions are equivalent where $\ma \in \R^{n \times d}$ has rows $A$.
\begin{enumerate}
    \item For any $\eps > 0$ there exists $\mr \in \R^{d \times d}$, a $(\vc, \eps)$-Forster transform of $\ma$.
    \item For every $k \in [d]$, every $k$-dimensional linear subspace $V \subseteq \R^d$ satisfies
    \begin{equation}\label{eq:light_subspace}\sum_{\substack{i \in [n] \\ \va_i \in V}} \vc_i \le k.\end{equation}
\end{enumerate}
\end{proposition}

One direction of Proposition~\ref{prop:scaling_subspace} is straightforward: if a $k$-dimensional subspace $V$ is too ``heavy'' (i.e., \eqref{eq:light_subspace} is violated) then there still exists a heavy subspace under any transform $\mr$. Taking the trace of both sides of the definition \eqref{eq:scale_rip_def} restricted to this heavy subspace yields a contradiction for sufficiently small $\eps$. In the case of $\vc = \frac d n \1_n$, \eqref{eq:light_subspace} simply translates to no $k$-dimensional subspace containing more than $\frac{k}{d} \cdot n$ of the points. One simple way for this condition to hold is if $A$ is in \emph{general position}. Note that by taking $V = \Span(A)$, \eqref{eq:light_subspace} implies that $n \ge d$ and $\ma$ has full rank.

The other direction is significantly more challenging, and \cite{HopkinsKLM20} gives an iterative construction based on decompositions with respect to the \emph{basis polytope} of $A = \{\va_i\}_{i \in [n]}$. This is the central object in the next characterization we will use, so we define it here.

\begin{definition}[Basis polytope]\label{def:basis_polytope}
Consider a point set $A = \{\va_i\}_{i \in [n]}$.
Let $\calB \subseteq 2^{[n]}$ be the set of subsets $B \subseteq [n]$ such that $\{\va_i\}_{i \in B}$ is a basis of $\R^d$, i.e., it is a linearly-independent set that spans $\R^d$. Letting $\1_B \in \{0, 1\}^n$ denote the $0$-$1$ indicator vector of each $B \in \calB$, we let 
\[\calP(A) \defeq \conv\Par{\Brace{\1_B}_{B \in \calB}}\]
denote the \emph{basis polytope} corresponding to the independent set matroid induced by $A$. 
\end{definition}

Another convenient reformulation of the necessary and sufficient condition in Proposition~\ref{prop:scaling_subspace} was given by \cite{CarlenLL04}; see also \cite{Barthe98, DvirSW14, HopkinsKLM20} for interpretations of this condition.

\begin{proposition}[Theorem 4.4, \cite{CarlenLL04}]\label{prop:scaling_polytope}
The conditions in Proposition~\ref{prop:scaling_subspace} hold iff $\vc \in \calP(A)$.
\end{proposition}

\paragraph{Scaling via leverage scores.} There is a primal-dual viewpoint of Definition~\ref{def:rip}, as described in our derivation \eqref{eq:c_isotropy}, \eqref{eq:radial}. In particular, $\vc$-RIP can equivalently be viewed as being induced by a pair $(\mr, \vs)$, where the diagonal scaling $\vs$ is an implicit function of the Forster transform $\mr$.

We may ask if this correspondence goes the other direction; are there conditions on $\vs \in \R^n_{> 0}$ such that one can deduce $\mr \in \R^{d \times d}$ that scales $\ma$ to be in $\vc$-RIP? The following observation, patterned from \cite{DadushR24}, shows the answer is yes: it is enough for $\vtau(\ms\ma) = \vc$, where $\ms \defeq \diag{\vs}$.

\begin{lemma}\label{lem:lev_rip}
Given a point set $A \defeq \{\va_i\}_{i \in [n]} \subset \R^d$ and $\vc \in (0, 1]^n$ satisfying $\norm{\vc}_1 = d$, suppose \eqref{eq:light_subspace} holds. Then letting $\ma \in \R^{n \times d}$ have rows $A$, if some $\vs \in \R^n_{> 0}$ satisfies $\vtau(\ms\ma) = \vc$, where $\ms \defeq \diag{\vs}$, then
\begin{equation}\label{eq:lev_implies_rip}
\sum_{i \in [n]} \vc_i \cdot \frac{\Par{\mr\va_i}\Par{\mr\va_i}^\top}{\norm{\mr\va_i}_2^2} = \id_d \text{ for } \mr \defeq \Par{\ma^\top \ms^2 \ma}^{-\half}.
\end{equation}
More generally, for any $\eps > 0$, if $\vtau(\ms\ma) \approx_\eps \vc$, then
\begin{equation}\label{eq:lev_implies_approx_rip}
\sum_{i \in [n]} \vc_i \cdot \frac{\Par{\mr\va_i}\Par{\mr\va_i}^\top}{\norm{\mr\va_i}_2^2} \approx_{\eps} \id_d \text{ for } \mr \defeq \Par{\ma^\top \ms^2 \ma}^{-\half}.
\end{equation}
\end{lemma}
\begin{proof}
We prove \eqref{eq:lev_implies_approx_rip}, which implies \eqref{eq:lev_implies_rip} by taking $\eps \to 0$. Indeed, by using $\vtau(\ms\ma) \approx_\eps \vc$,
\begin{align*}
\vtau_i\Par{\ms \ma} = \vs_i^2 \norm{\mr \va_i}_2^2 \approx_\eps \vc_i \implies \frac{\vc_i}{\norm{\mr \va_i}_2^2} \approx_{\eps} \vs_i^2 \text{ for all } i \in [n].
\end{align*}
This directly implies that \eqref{eq:lev_implies_approx_rip} holds:
\begin{align*}
\sum_{i \in [n]} \vc_i \cdot \frac{\Par{\mr\va_i}\Par{\mr\va_i}^\top}{\norm{\mr\va_i}_2^2} \approx_\eps \sum_{i \in [n]} \vs_i^2 \Par{\mr\va_i}\Par{\mr\va_i}^\top = \mr\Par{\sum_{i \in [n]} \vs_i^2 \va_i\va_i^\top} \mr = \id_d.
\end{align*}
\end{proof}

\paragraph{Barthe's objective.} In the rest of the paper, we fix a point set $A \defeq \{\va_i\}_{i \in [n]} \subset \R^d$, that forms the rows of $\ma \in \R^{n \times d}$. We also let $\vc \in (0, 1]^n$ satisfy $\norm{\vc}_1 = d$, such that the condition \eqref{eq:light_subspace} holds. For some $\eps \in (0, 1)$, we will give an algorithm for computing a $(\vc, \eps)$-Forster transform of $\ma$.

To ease our exposition we fix the following notation throughout, for $\vt \in \R^n$:
\begin{equation}\label{eq:notation}
\begin{aligned}
\mz\Par{\vt} &\defeq \ma^\top \diag{\exp\Par{\vt}} \ma = \sum_{i \in [n]} \exp\Par{\vt_i} \va_i\va_i^\top, \\
\mr\Par{\vt} &\defeq \mz\Par{\vt}^{-\half} = \Par{\ma^\top \diag{\exp\Par{\vt}} \ma}^{-\half},\; \\
\ms\Par{\vt} &\defeq \diag{\vs(\vt)}, \text{ where } \vs(\vt) \defeq \exp\Par{\frac{\vt}{2}}, \\
\tva_i(\vt) &\defeq \mr\Par{\vt} \va_i,\text{ for all } i \in [n].
\end{aligned}
\end{equation}
In \eqref{eq:notation}, we let $\exp$ be applied to a vector argument entrywise. Note that all of the matrices and vectors in \eqref{eq:notation} correspond to those arising in our earlier discussion, after reparameterizing the problem by $\vt = 2\log(\vs)$ entrywise. This reparameterization becomes convenient shortly.

Our approach follows the seminal work \cite{Barthe98}, which gave an algorithmic proof of Propositions~\ref{prop:scaling_subspace},~\ref{prop:scaling_polytope}, by explicitly characterizing the scaling $\vs \in \R^n_{> 0}$ in Lemma~\ref{lem:lev_rip} such that $\vtau(\ms\ma) = \vc$ for $\ms \defeq \diag{\vs}$, by way of a $\vt \in \R^n$ that achieves $\vs = \vs(\vt)$. To explain, we first define Barthe's objective:
\begin{equation}\label{eq:barthe_obj}
f\Par{\vt} \defeq -\inprod{\vc}{\vt} + \log\det\Par{\mz\Par{\vt}}.
\end{equation}
Then Barthe's result can be stated as follows.

\begin{proposition}[Proposition 6, \cite{Barthe98}]\label{prop:barthe}
Following notation \eqref{eq:notation}, \eqref{eq:barthe_obj}, $f: \R^n \to \R$ is a convex function, and its minimizer is attained iff $\ma$, $\vc$ satisfy $\norm{\vc}_1 = d$ and the condition \eqref{eq:light_subspace}. Moreover, letting $\vt^\star \defeq \argmin_{\vt \in \R^n} f(\vt)$, $\mr(\vt^\star)$ is a $\vc$-Forster transform of $\ma$.
\end{proposition}

Proposition~\ref{prop:barthe} can be somewhat demystified by computing the derivatives of Barthe's objective. We introduce one additional piece of notation here:
\begin{equation}\label{eq:mi_def}
\mm_i(\vt) \defeq \vs_i(\vt)^2 \tva_i(\vt)\tva_i(\vt)^\top = \mr\Par{\vt}\Par{\exp\Par{\vt_i} \va_i\va_i^\top}\mr\Par{\vt},\text{ for all } i \in [n].
\end{equation}

\begin{fact}\label{fact:barthe_derivs}
Following notation \eqref{eq:notation}, \eqref{eq:barthe_obj}, we have for all $(i, j) \in [n] \times [n]$ that
\begin{equation}\label{eq:barthe_derivs}
\begin{aligned}
\nabla_i f\Par{\vt} &= -\vc_i + \Tr\Par{\mm_i(\vt)}, \\
\nabla^2_{ij} f\Par{\vt}
 &= \Tr\Par{\mm_i(\vt)}\ind_{i = j} - \Tr\Par{\mm_i(\vt)\mm_j(\vt)}.
\end{aligned}
\end{equation}
\end{fact}

From Fact~\ref{fact:barthe_derivs} we can glean several different parts of Proposition~\ref{prop:barthe}. For example, the fact that $\sum_{i \in [n]} \mm_i(\vt) = \mr(\vt) \mz(\vt) \mr(\vt) = \id_d$, combined with Kadison's inequality \cite{Kadison52} (see also Theorem 2.3.2, \cite{Bhatia07}), shows that for all vectors $\vv \in \R^d$,
\begin{align*}
\Par{\sum_{i \in [n]} \vv_i \mm_i(\vt)}^2 \preceq \sum_{i \in [n]} \vv_i^2 \mm_i(\vt).
\end{align*}
In particular, taking a trace of both sides above shows
\begin{align*}
\nabla^2 f(\vt)[\vv, \vv] = \Tr\Par{\sum_{i \in [n]} \vv_i^2 \mm_i(\vt)} - \Tr\Par{\Par{\sum_{i \in [n]} \vv_i \mm_i(\vt)}^2} \ge 0,
\end{align*}
which implies that $f$ is convex. Similarly, letting $\vt^\star$ minimize $f$, we have from \eqref{eq:barthe_derivs} that
\begin{equation}\label{eq:optimality}
\vc_i = \Tr\Par{\mm_i(\vt^\star)} = \vtau_i\Par{\ms\Par{\vt^\star} \ma} \text{ for all } i \in [n].
\end{equation}
Using the characterization of $\ms(\vt^\star)$ in \eqref{eq:optimality} as obtaining the leverage scores $\vc$, and applying Lemma~\ref{lem:lev_rip}, we have shown that $\mr(\vt^\star)$ is indeed a $\vc$-Forster transform of $\ma$, as stated in Proposition~\ref{prop:barthe}.
\section{Optimizing Barthe's Objective via Newton's Method}\label{sec:newton}

In this section, we give our algorithm for computing approximate Forster transforms. Our algorithm is a variant of the box-constrained Newton's method of \cite{CohenMTV17}, which solves box-constrained quadratics to optimize a Hessian-stable function in $\norm{\cdot}_\infty$ to high precision. 

We first make our key technical observation in Section~\ref{ssec:barthe_stable}: that Barthe's objective is Hessian-stable with respect to $\norm{\cdot}_\infty$. We then give a termination condition in Section~\ref{ssec:terminate} that suffices for $\vt \in \R^n$ to induce an $\eps$-Forster transform $\mr(\vt)$. In Section~\ref{ssec:box_newton}, we leverage our implicit Laplacian sparsification algorithm from Section~\ref{sec:laplacian} to implement the iteration of the \cite{CohenMTV17} Newton's method. We put all the pieces together in Section~\ref{ssec:main_proof} to give our main result.

Throughout this section, we fix a pair $\ma \in \R^{n \times d}$ and $\vc \in (0, 1]^n$ satisfying $\norm{\vc}_1 = d$ and \eqref{eq:light_subspace}. We follow the notation outlined in Section~\ref{ssec:rip}, in particular, \eqref{eq:notation}, \eqref{eq:barthe_obj}, and \eqref{eq:mi_def}. We also will state our results under the diameter bound in Assumption~\ref{assume:simplify}.

\subsection{Hessian stability of Barthe's objective}\label{ssec:barthe_stable}

In this section, we prove the following key structural result enabling our approach.

\begin{restatable}{proposition}{restatehsb}\label{prop:hessian_stable_barthe}
For all $r > 0$, $f$ is $(r, 2r)$-Hessian stable with respect to $\norm{\cdot}_\infty$.
\end{restatable}

A similar result to Proposition~\ref{prop:hessian_stable_barthe} was previously established for the softmax objective 
\[\vt \to \log\Par{\sum_{i \in [n]} \exp(\vt_i)},\]
in Lemma 14, \cite{CarmonJJJLST20}, using more elementary techniques, i.e., directly establishing that the softmax satisfies a third-order regularity property called \emph{quasi-self-concordance}. 

Due to complications arising from Barthe's objective being defined with respect to potentially non-commuting matrices, we follow a substantially different approach in this section. We need the following helper lemma, where we define the \emph{Schur complement}
\begin{equation}\label{eq:sc_def}
\Schur\Par{\mm, S} \defeq \mm_{S:S} - \mm_{S:S^c} \mm_{S^c: S^c}^\dagger \mm_{S^c: S},
\end{equation}
for any square matrix $\mm \in \R^{d \times d}$, subset $S \subseteq [d]$, and $S^c \defeq [d] \setminus S$.

\begin{lemma}\label{lem:schur_approx}
If $\mm, \mn \in \PSD^{d \times d}$ and $\mm \approx_\eps \mn$, then 
\[\Schur\Par{\mm, S} \approx_\eps \Schur\Par{\mn, S} \text{ for all } S \subseteq [d].\]
\end{lemma}
\begin{proof}
It suffices to show that if $\ma, \mb \in \PSD^{d \times d}$,
\begin{equation}\label{eq:schur_preserve}
\ma \succeq \mb \implies \Schur\Par{\ma, S} \succeq \Schur\Par{\mb, S}.
\end{equation}
The claim then follows by applying \eqref{eq:schur_preserve} with $(\ma, \mb) \gets (\exp(\eps) \mm, \mn)$ and $\gets (\exp(\eps)\mn, \mm)$, since $\Schur(\alpha \mm, S) = \alpha \Schur(\mm, S)$ for any scaling coefficient $\alpha \in \R$. 

We now establish \eqref{eq:schur_preserve}. It is well-known (see, e.g., Appendix A.5.5 of \cite{BoydV04}) that 
\begin{align*}
\vx^\top \Schur\Par{\ma, S} \vx = \min_{\vy \in \R^{S^c}} \begin{pmatrix} \vx \\ \vy \end{pmatrix}^\top \ma \begin{pmatrix} \vx \\ \vy \end{pmatrix}
\end{align*}
for all $S \subseteq [d]$, $\vx \in \R^S$. Now \eqref{eq:schur_preserve} follows from
\begin{align*}
\vx^\top \Schur\Par{\ma, S} \vx = \min_{\vy \in \R^{S^c}} \begin{pmatrix} \vx \\ \vy \end{pmatrix}^\top \ma \begin{pmatrix} \vx \\ \vy \end{pmatrix} \ge \min_{\vy \in \R^{S^c}} \begin{pmatrix} \vx \\ \vy \end{pmatrix}^\top \mb \begin{pmatrix} \vx \\ \vy \end{pmatrix} = \vx^\top \Schur\Par{\mb, S} \vx.
\end{align*}
\end{proof}

We are now ready to prove Proposition~\ref{prop:hessian_stable_barthe}.

\begin{proof}[Proof of Proposition~\ref{prop:hessian_stable_barthe}]
Throughout, fix $\vt, \vt' \in \R^n$ with $\norm{\vt - \vt'}_\infty \le r$. Our goal is to show
\[\nabla^2 f(\vt) \approx_{2r} \nabla^2 f(\vt').\]
We follow the notation \eqref{eq:notation}, \eqref{eq:mi_def}, and whenever the argument is dropped, it is implied to be at $\vt$; we will use a superscript $'$ whenever the argument is at $\vt'$. So, for example, $\mm_i \equiv \mm_i(\vt)$ and $\mm'_i \equiv \mm_i(\vt')$ for all $i \in [n]$. Also, to ease notation in this proof we define
\[\mc_i \defeq \exp\Par{\vt_i} \va_i\va_i^\top,\; \mc'_i \defeq \exp\Par{\vt'_i} \va_i\va_i^\top,\]
so that $\mm_i = \mr \mc_i \mr$ for all $i \in [n]$. We first claim that for all $i \in [n]$,
\begin{equation}\label{eq:block_stable}
\begin{pmatrix}
        \mc_i&\0&\cdots&\0&\mc_i&\0&\cdots&\0\\
        \0&\0&\cdots&\0&\0&\0&\cdots&\0\\
        \vdots&\vdots&\ddots&\vdots&\vdots&\vdots&\ddots&\vdots\\
        \0&\0&\cdots&\0&\0&\0&\cdots&\0\\
        \mc_i&\0&\cdots&\0&\mc_i&\0&\cdots&\0\\
        \0&\0&\cdots&\0&\0&\0&\cdots&\0\\
        \vdots&\vdots&\ddots&\vdots&\vdots&\vdots&\ddots&\vdots\\
        \0&\0&\cdots&\0&\0&\0&\cdots&\0\\
        \end{pmatrix} \approx_r \begin{pmatrix}
        \mc'_i &\0&\cdots&\0&\mc'_i &\0&\cdots&\0\\
        \0&\0&\cdots&\0&\0&\0&\cdots&\0\\
        \vdots&\vdots&\ddots&\vdots&\vdots&\vdots&\ddots&\vdots\\
        \0&\0&\cdots&\0&\0&\0&\cdots&\0\\
        \mc'_i &\0&\cdots&\0&\mc'_i &\0&\cdots&\0\\
        \0&\0&\cdots&\0&\0&\0&\cdots&\0\\
        \vdots&\vdots&\ddots&\vdots&\vdots&\vdots&\ddots&\vdots\\
        \0&\0&\cdots&\0&\0&\0&\cdots&\0\\
        \end{pmatrix},
\end{equation}
where both matrices in \eqref{eq:block_stable} have dimensions $(n + 1)d \times (n + 1)d$, and only the $(1, 1)$, $(i + 1, 1)$, $(1, i + 1)$, and $(i + 1, i + 1)$-indexed $d \times d$ blocks are nonzero. We can verify \eqref{eq:block_stable} by direct expansion with respect to a $2d$-dimensional test vector with blocks $\vx, \vy$, which reduces the claim to
\[(\vx + \vy)^\top \mc_i (\vx + \vy) \approx_r (\vx + \vy)^\top \mc'_i (\vx + \vy) \impliedby \mc_i \approx_r \mc'_i,\]
where the latter fact above follows from $\norm{\vt - \vt'}_\infty \le r$. Summing \eqref{eq:block_stable} for all $i \in [n]$ shows
\begin{align*}
\ml \approx_r \ml', \text{ where } \ml &\defeq \begin{pmatrix}
    \sum_{i\in[n]} \mc_i & \mc_1 & \mc_2 & \cdots & \mc_n \\
    \mc_1 & \mc_1 & \0 & \cdots & \0 \\
    \mc_2 & \0 & \mc_2 & \cdots & \0 \\
    \vdots & \vdots & \vdots & \ddots & \vdots \\
    \mc_n & \0 & \0 & \cdots & \mc_n \\
    \end{pmatrix},\\
\ml' &\defeq \begin{pmatrix}
    \sum_{i\in[n]} \mc'_i & \mc'_1 & \mc'_2 & \cdots & \mc'_n \\
    \mc'_1 & \mc'_1 & \0 & \cdots & \0 \\
    \mc'_2 & \0 & \mc'_2 & \cdots & \0 \\
    \vdots & \vdots & \vdots & \ddots & \vdots \\
    \mc'_n & \0 & \0 & \cdots & \mc'_n \\
    \end{pmatrix}.
\end{align*}
Next, using Lemma~\ref{lem:schur_approx} to take Schur complements of $\ml, \ml'$ onto the index set $[(n + 1)d] \setminus [d]$ shows 
\begin{align*}
\mk \approx_r \mk', \text{ where } \mk &\defeq \begin{pmatrix}
    \mc_1 - \mc_1 \mz^{-1} \mc_1 & \cdots & -\mc_1 \mz^{-1} \mc_n \\
    -\mc_2 \mz^{-1} \mc_1 & \cdots & -\mc_2 \mz^{-1} \mc_n \\
    \vdots & \ddots & \vdots \\
    -\mc_n \mz^{-1} \mc_1 & \cdots & \mc_n - \mc_n \mz^{-1} \mc_n \\
    \end{pmatrix},\\
\mk' &\defeq \begin{pmatrix}
    \mc'_1 - \mc'_1 (\mz')^{-1} \mc'_1 & \cdots & -\mc'_1 (\mz')^{-1} \mc'_n \\
    -\mc'_2 (\mz')^{-1} \mc'_1 & \cdots & -\mc'_2 (\mz')^{-1} \mc'_n \\
    \vdots & \ddots & \vdots \\
    -\mc'_n (\mz')^{-1} \mc'_1 & \cdots & \mc'_n - \mc'_n (\mz')^{-1} \mc'_n \\
    \end{pmatrix},
\end{align*}
where we used that $\mz = \sum_{i \in [n]} \mc_i$ and $\mz' = \sum_{i \in [n]} \mc'_i$. Finally, fix some vector $\vv \in \R^n$. Let
\begin{align*}
\mj \defeq \vv\vv^\top \otimes \mz^{-1} = \begin{pmatrix}
    \vv_1^2 \mz^{-1} & \vv_1 \vv_2 \mz^{-1} & \vv_1 \vv_3 \mz^{-1} & \cdots & \vv_1 \vv_n \mz^{-1} \\
    \vv_1 \vv_2 \mz^{-1} & \vv_2^2 \mz^{-1} & \vv_2 \vv_3 \mz^{-1} & \cdots & \vv_2 \vv_n \mz^{-1} \\
    \vv_1 \vv_3 \mz^{-1} & \vv_2 \vv_3 \mz^{-1} & \vv_3^2 \mz^{-1} & \cdots & \vv_3 \vv_n \mz^{-1} \\
    \vdots & \vdots & \vdots & \ddots & \vdots \\
    \vv_1 \vv_n \mz^{-1} & \vv_2 \vv_n \mz^{-1} & \vv_3 \vv_n \mz^{-1} & \cdots & \vv_n^2 \mz^{-1} \\
    \end{pmatrix},
\end{align*}
where $\otimes$ denotes the Kronecker product. Similarly define $\mj' \defeq \vv\vv^\top \otimes (\mz')^{-1}$. Because we established each $\mc_i \approx_r \mc'_i$, we also have $\mz \approx_r \mz'$ and thus $\mz^{-1} \approx_r (\mz')^{-1}$. By well-known properties of the Kronecker product (cf.\ Theorem 2.3, \cite{Schacke13}), we conclude that $\mj \approx_r \mj'$. We have thus shown:
\[\mk \approx_r \mk',\; \mj \approx_r \mj'.\]
Finally, it is standard that if $\ma \preceq \mb$ and $\mc \preceq \md$ then $\inprod{\ma}{\mc} \le \inprod{\mb}{\md}$, if all of $\ma, \mb, \mc, \md$ are PSD matrices of the same dimension. Thus, $\inprod{\mk}{\mj} \approx_{2r} \inprod{\mk'}{\mj'}$. The conclusion follows upon realizing
\begin{align*}\inprod{\mk}{\mj} &= \sum_{i \in [n]} \vv_i^2 \Tr\Par{\mz^{-1} \mc_i} - \sum_{(i, j) \in [n] \times [n]} \vv_i\vv_j \Tr\Par{\mz^{-1}\mc_i\mz^{-1}\mc_j} \\
&= \sum_{i \in [n]} \vv_i^2 \Tr\Par{\mm_i} - \sum_{(i, j) \in [n] \times [n]} \vv_i\vv_j \Tr\Par{\mm_i\mm_j} = \vv^\top \nabla^2 f(\vt)\vv, \end{align*}
and similarly, $\inprod{\mk'}{\mj'} = \vv^\top \nabla^2 f(\vt') \vv$, by 
comparing to Fact~\ref{fact:barthe_derivs} and using the cyclic property of trace. This establishes $\vv^\top \nabla^2 f(\vt) \vv \approx_{2r} \vv^\top \nabla^2 f(\vt') \vv$ for all $\vv \in \R^n$, as desired.
\end{proof}

\subsection{Termination condition}\label{ssec:terminate}

In this section, we quantify the suboptimality gap (with respect to Barthe's objective) needed for $\vt \in \R^n$ to induce a $(\vc, \eps)$-Forster transform $\mr(\vt)$, as a function of $\eps$ and problem parameters. Our proof makes use of local adjustments and is inspired by a similar technique in \cite{CohenMTV17}.

\begin{lemma}\label{lem:termination}
Let $\eps \in (0, 1)$, and suppose $\vt \in \R^n$ satisfies
\begin{equation}\label{eq:error_bound}f\Par{\vt} - f\Par{\vt^\star} \le \frac{\eps^2\min_{i \in [n]} \vc_i^2} 2 \end{equation}
where $\vt^\star \in \argmin_{\vt \in \R^n} f(\vt)$. Then $\mr(\vt)$ is a $(\vc, \eps)$-Forster transform.
\end{lemma}
\begin{proof}
We prove the contrapositive. Suppose $\mr(\vt)$ is not a $(\vc, \eps)$-Forster transform. By Lemma~\ref{lem:lev_rip}, there are two cases of leverage score violations to consider. We show that both cases contradict \eqref{eq:error_bound}, by designing local improvements to $\vt$ in any coordinate with a violating leverage score.

\textit{Case 1.} Suppose that for some $i \in [n]$, we have $\vtau_i(\ms(\vt) \ma) > \exp(\eps) \vc_i$. Let $\vt' \defeq \vt - \delta \ve_i$ for some choice of $\delta > 0$  that we will optimize later. Then,
\begin{align*}
f\Par{\vt} - f\Par{\vt'} &= \log\Par{\frac{\det\Par{\mz(\vt)}}{\det\Par{\mz(\vt')}}} - \delta \vc_i \\
&= \log\Par{\frac{\det\Par{\mz(\vt)}}{\det\Par{\mz(\vt) + \Par{\exp(-\delta) - 1}\exp\Par{\vt_i}\va_i\va_i^\top}}} - \delta \vc_i \\
&= \log\Par{\frac{\det\Par{\mz(\vt)}}{\det\Par{\mz(\vt)}\Par{1 + \Par{\exp(-\delta) - 1}\exp\Par{\vt_i}\va_i^\top \mz(\vt)^{-1}\va_i}}} - \delta \vc_i \\
&= \log\Par{\frac{\det\Par{\mz(\vt)}}{\det\Par{\mz(\vt)}\Par{1 + \Par{\exp(-\delta) - 1}\vtau_i(\ms(\vt) \ma)}}} - \delta \vc_i \\
&= -\log\Par{1 + \Par{\exp(-\delta) - 1}\vtau_i(\ms(\vt) \ma)} - \delta \vc_i \\
&\ge (1 - \exp(-\delta))\vtau_i(\ms(\vt)\ma) - \delta \vc_i \\
&\ge \Par{\delta - \frac{\delta^2}{2}} \vtau_i(\ms(\vt)\ma) - \delta\vc_i > \delta\eps \vc_i - \frac{\delta^2}{2} \ge \half\Par{\eps\vc_i}^2, 
\end{align*}
where the first two lines expanded definitions, the third line uses the matrix determinant lemma, the fourth line uses the definition of leverage scores \eqref{eq:leverage}, the sixth line uses $-\log(1 + c) \ge -c$ for all $c \in \R$, and the seventh line uses $1 - \exp(-c) \ge c - \frac{c^2}{2}$, $\exp(c) - 1 \ge c$ for $c > 0$ and $\vtau_i(\ms(\vt)\ma) \le 1$ (Fact~\ref{fact:lev}). By choosing the optimal $\delta = \vc_i \eps$, we have a contradiction to \eqref{eq:error_bound}.

\textit{Case 2.} Suppose that for some $i \in [n]$, we have $\vtau_i(\ms(\vt)\ma) < \exp(-\eps)\vc_i \le \frac{\vc_i}{1+\eps}$. Let $\vt' \defeq \vt + \delta \ve_i$ for some choice of $\delta > 0$ that we will optimize later. Then, following analogous derivations as before,
\begin{align*}
f\Par{\vt} - f\Par{\vt'} &= -\log\Par{1 + \Par{\exp\Par{\delta} - 1}\vtau_i(\ms(\vt)\ma) } + \delta\vc_i \\
&> -\log\Par{1 + \Par{\exp\Par{\delta} - 1}\frac{\vc_i}{1+\eps}} + \delta \vc_i \\
&= \log\Par{1 + \eps} - \Par{1 - \vc_i}\log\Par{1 + \frac{\eps}{1 - \vc_i}} \ge \half\Par{\eps\vc_i}^2,
\end{align*}
where the second line uses $\vtau_i(\ms(\vt)\ma) < \frac{\vc_i}{1+\eps}$, the third line chose $\delta = \log(1 + \frac{\eps}{1 - \vc_i})$, and used $\log(1 + c) - b\log(1 + \frac c b) \ge \half((1 - b)c)^2$ for all $b, c \in (0, 1)$. Again this gives a contradiction to \eqref{eq:error_bound}. We remark that the case of $\vc_i = 1$ can be handled using a limiting argument.
\end{proof}

\subsection{Box-constrained Newton's method}\label{ssec:box_newton}

Thus far we have established that $f$ is Hessian stable in $\norm{\cdot}_\infty$ (Proposition~\ref{prop:hessian_stable_barthe}) and needs to be minimized to error $\half \eps^2 \min_{i \in [n]} \vc_i^2$ for our desired application (Lemma~\ref{lem:termination}). We also are given under Assumption~\ref{assume:simplify} that the global minimizer lies inside $\ball_{\infty}(\log(\kappa))$. 

It remains to give an algorithm for efficiently optimizing Hessian stable functions. Fortunately, such a toolkit was provided by \cite{CohenMTV17, ChenPW21}. The former work designed an approximation-tolerant box-constrained Newton's method, tailored towards objectives whose Hessians display a certain combinatorial structure, and the latter work showed how to optimize box-constrained quadratics in these structured Hessians. We can leverage this toolkit due to the next observation.

\begin{lemma}\label{lem:barthe_laplacian}
For all $\vt \in \R^n$, $\nabla^2 f(\vt)$ is a \emph{graph Laplacian}, i.e., $\nabla^2_{ij} f(\vt) \ge 0$ iff $i = j$, and
\[\sum_{j \in [n]} \nabla^2_{ij} f(\vt) = 0 \text{ for all } i \in [n].\]
\end{lemma}
\begin{proof}
The first property is immediate by inspection of \eqref{eq:barthe_derivs}, and using that all the $\mm_i(\vt) \in \PSD^{d \times d}$. The second property follows because $\sum_{j \in [n]} \mm_j(\vt) = \id_d$, so for all $i \in [n]$ we have
\begin{align*}
\sum_{j \in [n]} \nabla^2_{ij} f(\vt) = \Tr\Par{\mm_i(\vt)} - \inprod{\mm_i(\vt)}{\sum_{j \in [n]} \mm_j(\vt)} = \Tr\Par{\mm_i(\vt)} - \inprod{\mm_i(\vt)}{\id_d} = 0.
\end{align*}
\end{proof}

We remark that Lemma~\ref{lem:barthe_laplacian} gives another short proof of $f$'s convexity: it is well-known that graph Laplacians are PSD matrices, which follows e.g., by the Gershgorin circle theorem.

One complication that arises in our algorithm is that computing the Hessian $\nabla^2 f$ is more expensive than providing matrix-vector query access to it, due to a convenient factorization.

\begin{lemma}\label{lem:grad_compute}
Given $\vt \in \R^n$, we can compute $\nabla f(\vt)$ in $O(nd^{\omega - 1})$ time and $\nabla^2 f(\vt)$ in $O(n^2 d^{\omega - 2})$ time. Additionally, given $\vt, \vv \in \R^n$, we can compute  $\nabla^2 f(\vt) \vv$ in $O(nd^{\omega - 1})$ time.
\end{lemma}
\begin{proof}
Recall the formulas for $\nabla f(\vt)$, $\nabla^2 f(\vt)$ in Fact~\ref{fact:barthe_derivs}. For the former claim, following \eqref{eq:notation}, \eqref{eq:mi_def}, we first compute $\ms(\vt) \ma$ in time $O(nd)$, which lets us compute $\mz(\vt)$ in time $O(nd^{\omega - 1})$ by multiplying $d \times n$ and $n \times d$ matrices. We can then compute $\ma \mr(\vt)$ to obtain all of the vectors $\tva_i(\vt)$ in $O(nd^{\omega - 1})$ time. This lets us obtain all $\Tr(\mm_i(\vt)) = \vs_i(\vt)^2 \norms{\tva_i(\vt)}_2^2$ in $O(nd)$ additional time.

It remains to compute the $n \times n$ matrix with $(i, j)^{\text{th}}$ entry $\Tr(\mm_i(\vt)\mm_j(\vt))$. Observe that 
\[\Tr(\mm_i(\vt)\mm_j(\vt)) = \vs_i(\vt)^2\vs_j(\vt)^2\inprod{\tva_i(\vt)}{\tva_j(\vt)}^2.\]
Thus it is enough to form the matrix with $(i, j)^{\text{th}}$ entry $\inprod{\tva_i(\vt)}{\tva_j(\vt)}$, multiply it entrywise by $\vs_i(\vt)\vs_j(\vt)$, and entrywise square it, in $O(n^2)$ time. The former matrix is $\ma \mz(\vt)^{-1} \ma^\top$, which takes time $O(n^2 d^{\omega - 2})$ time to compute by multiplying $n \times d$, $d \times d$, and $d \times n$ matrices.

For the latter claim, we can again first compute
\[\diag{\Brace{\Tr\Par{\mm_i(\vt)}}_{i \in [n]}} \vv \]
in time $O(nd^{\omega - 1})$ using the steps described above. To implement $\nabla^2 f(\vt) \vv$, it remains to compute 
\[\inprod{\mm_i(\vt)}{\sum_{j \in [n]} \vv_j \mm_j(\vt)} = \exp(\vt_i)\Par{\mr(\vt) \Par{\sum_{j \in [n]} \vv_j \mm_j(\vt)} \mr(\vt)}\Brack{\va_i, \va_i} \text{ for all } i \in [n].\]
Observe that
\begin{align*}
\mc \defeq \mr(\vt) \Par{\sum_{j \in [n]} \vv_j \mm_j(\vt)} \mr(\vt) &= \mz(\vt)^{-1} \Par{\sum_{j \in [n] } \vv_j\exp(\vt_j) \va_j\va_j^\top} \mz(\vt)^{-1},
\end{align*}
which can be computed in $O(nd^{\omega - 1})$ time by first forming the middle matrix on the right-hand side via multiplying $d \times n$ and $n \times d$ matrices. Finally to compute all $\mc[\va_i, \va_i]$, we can take the rows of $\ma \mc$ and obtain their dot products with rows of $\ma$ which requires $O(nd^{\omega - 1})$ time to compute. 
\end{proof}

To capitalize on the faster matrix-vector access given by Lemma~\ref{lem:barthe_laplacian}, in Section~\ref{sec:laplacian} we give a proof of Theorem~\ref{thm:implicit_sparsify}, our main result on the implicit sparsification of graph Laplacians. We will apply Theorem~\ref{thm:implicit_sparsify} to sparsify the Hessian of a regularized variant of Barthe's objective. Next, we require a tool from \cite{ChenPW21} for optimizing box-constrained quadratics in a graph Laplacian. 

\begin{proposition}[Theorem 1.1, \cite{ChenPW21}]\label{prop:box_oracle_impl}
Let $\delta \in (0, 1)$, let $\vl, \vu \in \R^n$ have $\vl \le \vu$ entrywise, let $\vb, \vt \in \R^n$, and let $\ml \in \PSD^{n \times n}$ be a graph Laplacian with $\nnz(\ml) \le m$. Let
\[\calW \defeq \Brace{\vw \in \R^n \mid \vl_i \le \vt_i + \vw_i \le \vr_i \text{ for all } i \in [n]}.\]
There is an algorithm $\oracle(\ml, \vb, \calW)$ that runs in time $O((n + m)^{1 + o(1)}\log(\frac 1 \delta))$, and with probability $\ge 1 - \delta$, it returns $\vv \in \calW$, satisfying
\[\inprod{\vb}{\vv} + \half \ml\Brack{\vv, \vv} \le \half \min_{\vw \in \calW}\Brace{\inprod{\vb}{\vw} + \half \ml\Brack{\vw, \vw}}.\]
\end{proposition}

We give additional discussion on Proposition~\ref{prop:box_oracle_impl} in Appendix~\ref{app:boxoracle}, as in \cite{ChenPW21} it was only stated for the case $\vl = \0_n$ and $\vu$ is $\infty$ in each coordinate, i.e., the box constraint is simply the positive orthant $\R^n_{\ge 0}$. However, the techniques extend straightforwardly to general box constraints \cite{ChenPW25}.

We now show how to use Proposition~\ref{prop:box_oracle_impl} to efficiently optimize an $\ell_\infty$-Hessian stable function. The following proof is based on Theorem 3.4, \cite{CohenMTV17}, but adapts it to tolerate multiplicative error in the Hessian computation. We note that a similar multiplicatively-robust generalization appeared earlier as Lemma 20, \cite{AssadiJJST22}, but was too restrictive for our purposes.

\begin{lemma}\label{lem:koracle_step}
Let convex $F: \R^n \to \R$ be $(1, 2)$-Hessian stable with respect to $\norm{\cdot}_\infty$, and let $\vt^\star \in \argmin_{\vt \in \R^n} F(\vt)$ have $\norm{\vt^\star}_\infty \le \log(\kappa)$. For $\vt \in \R^n$, $\alpha \ge 1$, let $\tml \in \PSD^{n \times n}$ be a graph Laplacian with
\[\nabla^2 F\Par{\vt} \preceq \tml \preceq \alpha \nabla^2 F\Par{\vt}.\]
Then for any $\vt \in \ball_\infty(\log(\kappa))$, if $\vt' \gets \vt + \oracle(8\tml, \nabla F(\vt), \ball_\infty(\vt, 1) \cap \ball_\infty(\log(\kappa)))$, where $\oracle$ is as in Proposition~\ref{prop:box_oracle_impl}, we have 
\[F\Par{\vt'} - F\Par{\vt^\star} \le \Par{1 - \frac{1}{240\alpha\log(\kappa)}}\Par{F\Par{\vt} - F\Par{\vt^\star}}.\]
\end{lemma}
\begin{proof}
For any $\vu$ with $\norm{\vt - \vu}_\infty \le 1$, Hessian stability of $F$ yields the bounds
\begin{equation}\label{eq:bounds_stable}
\begin{aligned}
F\Par{\vu} - F\Par{\vt} - \inprod{\nabla F(\vt)}{\vu - \vt} &= \int_0^1 \Par{1 - \lam} \nabla^2 F\Par{\Par{1 - \lam}\vt + \lam \vu }\Brack{\vu - \vt, \vu - \vt} \dd \lam \\
&\le \int_0^1 \Par{1 - \lam} e^2 \tml \Brack{\vu - \vt, \vu - \vt} \dd \lam \\
&\le \frac{e^2}{2}\tml\Brack{\vu - \vt, \vu - \vt} \le 4\tml\Brack{\vu - \vt, \vu - \vt}, \\
F\Par{\vu} - F\Par{\vt} - \inprod{\nabla F(\vt)}{\vu - \vt} &\ge  \frac 1 {2\alpha e^2} \tml\Brack{\vu - \vt, \vu - \vt} \ge \frac 1 {15\alpha} \tml\Brack{\vu - \vt, \vu - \vt}.
\end{aligned}
\end{equation}
Next define
\[\hvdelta \defeq \argmin_{\substack{\norm{\vdelta}_\infty \le 1 \\ \vt + \vdelta \in \ball_\infty(\log(\kappa))}} \inprods{\nabla F(\vt)}{\vdelta} + 4\tml\Brack{\vdelta, \vdelta},\]
and observe that for $\vdelta \defeq \vt' - \vt = \oracle(8\tml, \nabla F(\vt), \ball_\infty(\vt, 1) \cap \ball_\infty(\log(\kappa)))$, we have
\begin{equation}\label{eq:deltastar_sub}
\begin{aligned}\inprod{\nabla F(\vt)}{\vdelta} + 4\tml\Brack{\vdelta, \vdelta} &\le \half\Par{\inprods{\nabla F(\vt)}{\hvdelta} + 4\tml[\hvdelta, \hvdelta]} \\
&\le \half\Par{\inprods{\nabla F(\vt)}{\vdelta^\star} + 4\tml[\vdelta^\star, \vdelta^\star]}, \\
\text{for any }\norm{\vdelta^\star}_\infty &\le 1 \text{ with } \vt + \vdelta^\star \in \ball_\infty(\log(\kappa)),
\end{aligned}
\end{equation}
from the oracle guarantee and definition of $\hvdelta$. Hence, applying the upper bounds in \eqref{eq:bounds_stable} and \eqref{eq:deltastar_sub},
\begin{equation}\label{eq:upper_bound_iter}
\begin{aligned}
F\Par{\vt'} &\le F\Par{\vt} + \inprod{\nabla F\Par{\vt}}{\vdelta} + 4 \tml\Brack{\vdelta, \vdelta} \\
&\le F\Par{\vt} + \half\Par{\inprods{\nabla F(\vt)}{\vdelta^\star} + 4\tml[\vdelta^\star, \vdelta^\star]},
\end{aligned}
\end{equation}
for our choice of $\vdelta^\star$ satisfying the bounds in \eqref{eq:deltastar_sub}. We choose $\vdelta^\star = \frac{c}{2\log(\kappa)} (\vt^\star - \vt)$ where $c = \frac 1 {60\alpha}$. First observe that this is a valid choice of movement, because
\begin{align*}
\vt + \vdelta^\star &= \Par{1 - \frac{c}{2\log(\kappa)}}\vt + \frac{c}{2\log(\kappa)} \vt^\star \in \ball_\infty\Par{\log(\kappa)}, \\
\norm{\vdelta^\star}_\infty &\le \frac 1 {2\log(\kappa)}\Par{\norm{\vt}_\infty + \norm{\vt^\star}_\infty} \le \frac{2\log(\kappa)}{2\log(\kappa)} = 1,
\end{align*}
where both inequalities used that $\vt, \vt^\star \in \ball_\infty(\log(\kappa))$, which is a convex set. Thus,
\begin{align*}
\half\Par{\inprods{\nabla F(\vt)}{\vdelta^\star} + 4\tml[\vdelta^\star, \vdelta^\star]} &= \frac{1}{120\alpha}\Par{\inprod{\nabla F(\vt)}{\frac 1 {c} \vdelta^\star} + \frac{1}{15\alpha} \tml\Brack{\frac 1 {c} \vdelta^\star, \frac 1 {c}\vdelta^\star}} \\
&\le \frac 1 {120\alpha}\Par{F\Par{\vt + \frac 1 {c}\vdelta^\star} - F\Par{\vt}} \\
&\le -\frac{1}{240\alpha \log(\kappa)}\Par{F\Par{\vt} - F\Par{\vt^\star}}.
\end{align*}
The first inequality above used the lower bound in \eqref{eq:bounds_stable}, and the second inequality used convexity of $F$. At this point, combining with \eqref{eq:upper_bound_iter} yields the conclusion.
\end{proof}

\subsection{Proof of Theorem~\ref{thm:main}}\label{ssec:main_proof}

In this section we put together the pieces we have built to obtain our final algorithm. To begin, we note that under Assumption~\ref{assume:simplify}, the following simple initial function error bound holds.

\begin{lemma}\label{lem:initial_error}
Under Assumption~\ref{assume:simplify}, letting $\vt^\star \in \argmin_{\vt \in \R^n} f(\vt) \cap \ball_\infty(\log(\kappa))$, we have that
\[f\Par{\0_n} - f\Par{\vt^\star} \le \frac{d\log^2(\kappa)} 2.\]
\end{lemma}
\begin{proof}
We first note that for all $\vt, \vv \in \R^n$, we have
\begin{align*}
\nabla^2 f(\vt)\Brack{\vv, \vv} \le \diag{\Brace{\Tr\Par{\mm_i(\vt)}}_{i \in [n]}}\Brack{\vv, \vv} = \sum_{i \in [n]} \vtau_i\Par{\ms(\vt) \ma} \vv_i^2 \le d\norm{\vv}_\infty^2,
\end{align*}
i.e., $f$ is $d$-smooth with respect to $\norm{\cdot}_\infty$ (we used Fact~\ref{fact:lev} in the last inequality above). Thus by the second-order Taylor expansion from $\vt^\star$ to $\0_n$, and using that $\nabla f(\vt^\star) = \0_n$,
\begin{align*}f(\0_n) &= f(\vt^\star) + \int_0^1 (1 - \lam) \nabla^2 f\Par{\Par{1 - \lam} \vt^\star}\Brack{\vt^\star, \vt^\star} \dd \lam \\
&\le f(\vt^\star) + \frac{d}{2} \norm{\vt^\star}_\infty^2 \le f(\vt^\star) + \frac{d\log^2(\kappa)}{2}. \end{align*}
\end{proof}

We can now prove Theorem~\ref{thm:main}, the main result of this section.

\restatemain*
\begin{proof}
Throughout this proof, we let $f$ be Barthe's objective, and
\[F(\vt) \defeq f(\vt) + \frac{\eps^2 \vc_{\min}^2}{4\log^2(\kappa)} \vt^\top\mproj\vt,\]
where $\mproj \defeq \id_n - \frac 1 n \1_n\1_n^\top$. Our goal is to optimize $F$ to error $\frac{\eps^2\vc_{\min}^2}{4}$ over $\ball_\infty(\log(\kappa))$. To see why this suffices, note that for all $\vt \in \ball_\infty(\log(\kappa))$, we have $\vt^\top \mproj \vt \le \norm{\vt}_2^2 \le \log^2(\kappa)$, and hence any $\frac{\eps^2\vc_{\min}^2}{4}$-minimizer to $F$ over $\ball_\infty(\log(\kappa))$ satisfies
\[f(\vt) - f(\vt^\star) \le F(\vt) - F(\vt^\star) + \frac{\eps^2 \vc_{\min}^2}{4} \le \frac{\eps^2 \vc_{\min}^2}{2}. \]
Now, applying Lemma~\ref{lem:termination} gives the claim, because computing $\mr(\vt)$ does not dominate the stated runtime (as described in Lemma~\ref{lem:grad_compute}). Note that each time we apply Lemma~\ref{lem:koracle_step} with $\alpha \gets (\frac{n\log(\kappa)}{\eps\vc_{\min}})^{o(1)}$, we improve the function error by a multiplicative $1 - \Omega((\alpha\log(\kappa))^{-1})$. Moreover, the initial function error is bounded as in Lemma~\ref{lem:initial_error}. Thus to obtain the stated runtime, it is enough to show how to implement each call to Lemma~\ref{lem:koracle_step} with $\alpha \gets (\frac{n\log(\kappa)}{\eps\vc_{\min}})^{o(1)}$, in time
\begin{equation}\label{eq:per_iter}O\Par{nd^{\omega - 1}\Par{\frac{n\log(\kappa)}{\delta\eps\vc_{\min}}}^{o(1)}}.\end{equation}
Adjusting $\delta$ by the number of calls and taking a union bound then gives the claim. 

To achieve this runtime we first produce a sparse graph Laplacian matrix $\tml$ satisfying \eqref{eq:sparse_quality} for $\Delta \defeq \frac{\eps^2 \vc_{\min}^2}{4\log^2(\kappa)}$, and $\ml \gets \nabla^2 f(\vt)$ for some iterate $\vt$. Recalling that $\Tr(\ml) = \Tr(\id_d) = d$, Theorem~\ref{thm:implicit_sparsify} and Lemma~\ref{lem:grad_compute} guarantee that we can compute such a $\tml$ with probability $\ge 1 - \delta$ within time \eqref{eq:per_iter}.
Given $\tml$, the per-iteration runtime follows from Proposition~\ref{prop:box_oracle_impl} which does not dominate.
\end{proof}

\begin{remark}\label{rem:log_delta_eps}
For highly-accurate solutions or extremely small failure probabilities (i.e., $\delta, \eps$ smaller than an inverse polynomial in $n$), the subpolynomial dependences on $\frac 1 \delta$, $\frac 1 \eps$ in Theorem~\ref{thm:main} could be dominant factors. However, these subpolynomial factors only arise due to the use of Theorem~\ref{thm:implicit_sparsify} to sparsely approximate Hessians of Barthe's objective. If we instead directly compute the Hessians via Lemma~\ref{lem:grad_compute}, then slightly modifying the proof of Theorem~\ref{thm:main} yields an alternate runtime of
\[O\Par{n^2 d^{\omega - 2} \log\Par{\kappa}\polylog\Par{\frac{n\log(\kappa)}{\delta\eps \vc_{\min}}}}.\]
This runtime gives a worse dependence on $n$, but improves the dependences on other parameters (i.e., $\frac 1 \delta, \frac 1 \eps, \frac 1 {\vc_{\min}}$) from subpolynomial to polylogarithmic.
\end{remark}

\section{Sparsifying Laplacians with Matrix-Vector Queries}\label{sec:laplacian}

In this section, we prove Theorem~\ref{thm:implicit_sparsify}.
Our approach is inspired by \cite{JambulapatiLMSST23}, who showed how to tightly approximate a graph Laplacian using a dictionary of edge Laplacians, given appropriate access (i.e., to its inverse). Our result is incomparable, as it obtains a substantially faster runtime under weaker access, but gives a much looser approximation factor.

In Section~\ref{ssec:soc}, we first define a combinatorial structure that we maintain for the iterates of our methods, supporting efficient matrix-vector products. In Section~\ref{ssec:grid} we next show how to discretize a distance-structured vector in a way that is compatible with our combinatorial structure.

We use our vector approximations in Section~\ref{ssec:packing} to implement an approximate packing SDP oracle for combinatorially structured inputs. Finally, in Sections~\ref{ssec:mdr} and~\ref{ssec:homotopy}, we give an efficient reduction from the two-sided approximation in \eqref{eq:sparse_quality} to our packing SDP oracle.

\subsection{Sum-of-cliques representation}\label{ssec:soc}

Throughout this section, we let $E$ index the (unordered) edges of a graph on $[n] \times [n]$, i.e., $|E| = \binom n 2$ and $E$ has one index $(i, j)$ for each unordered tuple $(i, j) \in [n] \times [n]$. For an index $e = (i, j) \in E$, we also define the corresponding edge Laplacian $\ml_e \in \PSD^{n \times n}$:
\[\ml_e \defeq \Par{\ve_i - \ve_j}\Par{\ve_i - \ve_j}^\top.\]
We also define the graph Laplacian induced by a weight vector $\vv \in \R^E_{\ge 0}$ by
\[\calL(\vv) \defeq \sum_{e \in E} \vv_e \ml_e.\] Additionally, we define the ``dual'' operator $\calL^*$ which takes input $\mm \in \R^{n \times n}$ as
\[
\calL^*(\mm) \defeq \Brace{\inprod{\ml_e}{\mm}}_{e \in E} \in \R^E.
\]
Next, we define a helpful combinatorial structure for maintaining our iterates efficiently.

\begin{definition}[Sum-of-cliques]\label{def:soc}
    For $K \in \N$, we say $\vv \in \R^E_{\ge 0}$ is a \emph{$K$-sum-of-cliques} ($K$-SOC) if there exists $\vw \in \R^K_{\ge 0}$ and $K$ partitions of subsets $\{S_j\}_{j \in [K]}$ of $[n]$, $\{\calP_j\}_{j \in [K]}$, such that 
    \begin{equation}\label{eq:ksoc_def}
\calL(\vv) = \sum_{j \in [K]} \vw_j \sum_{S \in \calP_j} \ml_{S},
    \end{equation}
    where $\ml_{S}$ refers to an unweighted clique Laplacian placed over the set $S \subseteq [n]$, i.e.,
    \begin{equation}\label{eq:clique_lap}\ml_{S} = |S| \id_{S} -  \1_{S}\1_{S}^\top.\end{equation}
    We say $\va \in \{0, 1\}^E$ is an \emph{anti-sum-of-cliques} (ASOC) if there exists a partition $\calA$ of a subset $S \subseteq [n]$ such that
    \begin{equation}\label{eq:asoc_def}
\calL(\va) = \ml_S - \sum_{A \in \calA} \ml_{A}.
    \end{equation}
\end{definition}

The rest of this section shows that sum-of-cliques and anti-sum-of-cliques vectors induce Laplacians that support fast matrix-vector products and sparsification.

\begin{remark}\label{rem:partition}
We briefly discuss our computational model for representing partitions; it will be clear that this representation can be maintained under all operations in this section. We represent a $k$-piece partition $\calP = \{P_i\}_{i \in [k]}$ of a subset $S \subseteq [n]$ using two arrays. The first array has $k$ elements, each of which is an array containing the members of $P_i$. The second array has $n$ elements, each of which marks the partition piece vertex $j \in [n]$ belongs to, or $0$ if $j \not\in \cup_{i \in [k]} P_i$.
\end{remark}

\paragraph{Matrix-vector products.} We now prove that SOCs and ASOCs induce Laplacians with fast matrix-vector products. We require a helper result on computing mutual refinements of partitions.

\begin{lemma}\label{lem:refine_partition}
Let $\calP$, $\calA$ be partitions of $[n]$. Let $\calB$ be the \emph{mutual refinement} of $\calP, \calA$, i.e., $\calB$ is a partition of $[n]$ such that if $S \in \calP$ and $T \in \calA$, then $S \cap T \in \calB$. If $\calP$ and $\calA$ are explicitly given, we can compute $\calB$ in $O(n\log(n))$ time.\footnote{There is a simple randomized $O(n)$ time solution, but it uses hashing, so we include this solution for completeness.}
\end{lemma}
\begin{proof}
In $O(n)$ time, we compute for each element $i \in [n]$ the ordered tuple $(a, p)$, where $a$ is the partition piece in $\calA$ containing $i$, and similarly $p$ is the partition piece in $\calP$. Next we can sort the $\{(a, p)\}$ in lexicographical order to find the indices of the partition pieces in $\calB$. Finally we can assign each element of $[n]$ to the appropriate partition piece in $\calB$.
\end{proof}

\begin{lemma}\label{lem:matvec_soc}
Let $\vv$ be a $K$-SOC satisfying \eqref{eq:ksoc_def} with respect to $\vw \in \R^K_{\ge 0}$ and $\{\calP_j\}_{j \in [K]}$, partitions of $\{S_j\}_{j \in [K]}$,  and let $\va$ be an ASOC satisfying \eqref{eq:asoc_def} with respect to $\calA$, a partition of $S \subseteq [n]$. Suppose $\vw$, $\{\calP_j\}_{j \in [K]}$, and $\calA$ are explicitly given. We can compute
\[\calL(\vv \circ \va) \vu = \sum_{e \in E} \vv_e \va_e \ml_e \vu\]
for any $\vu \in \R^n$ in $O(nK\log(n))$ time.
\end{lemma}
\begin{proof}
Let $\vv_j \in \{0, 1\}^E$ be the $1$-SOC corresponding to a unit weight and a single partition $\calP_j$, so that $\vv = \sum_{j \in [K]} \vw_j \vv_j$. We will show that we can implement the matrix-vector product
\[\calL(\vv_j \circ \va) \vu = \sum_{e \in E} [\vv_j]_e \va_e \ml_e \vu\]
in $O(n\log(n))$ time. Summing over all $j \in [K]$ then gives the result.

In the rest of the proof let $\{T_1, \ldots, T_m\}$ be the partition $\calP_j$ inducing $\vv_j$. Also let $\vb = \1_F - \va$ where $F \subseteq E$ indicates the edges corresponding to the clique on $S$, which is a $1$-SOC. Observe that $\calL(\vv_j \circ \va) \vu = \calL(\vv_j) - \calL(\vv_j \circ \vb)$, and $\calL(\vv_j)\vu $ is computable in $O(n)$ time:
\[\calL(\vv_j) \vu = \sum_{T_i \in \calP_j} \ml_{T_i} \vu. \]
In particular, each $\ml_{T_i} \vu$ can be evaluated in $O(T_i)$ time using the formula \eqref{eq:clique_lap}, and $\sum_{i \in [m]} |T_i| \le n$. 

Moreover, note that $\vb \circ \vv_j$ is itself a vector in $\{0, 1\}^E$ corresponding to a $1$-SOC on the set of vertices $S \cap S_j \subseteq [n]$ (all other coordinates are $0$ in either $\vb$ or $\vv_j$). This is because each pair of vertices corresponds to a weight-$1$ edge in $\vb \circ \vv_j$ iff they lie in the same partition piece in both $\calP_j$ and $\calA$, which is an equivalence relation forming a partition of $S \cap S_j$. In fact, this partition is the mutual refinement of $\calP_j$ and $\calA$, following Lemma~\ref{lem:refine_partition} (which we call with $[n] \gets S \cap S_j$). Thus, we can compute the mutual refinement inducing $\vb \circ \vv_j$ in $O(n\log(n))$ time, and the same logic as above shows we can implement matrix-vector products through $\calL(\vv_j \circ \vb)$ in $O(n)$ time.
\end{proof}

\paragraph{Sparsification.} As a second key result, we show that it is possible to efficiently sparsify the product of a clique and an ASOC. Specifically, let $\vc \in \{0, 1\}^E_{\ge 0}$ be the indicator vector of a clique over vertices $S \subseteq [n]$ (so $\calL(\vc) = \ml_S$ defined in \eqref{eq:clique_lap}), and let $\va$ be an ASOC associated with partition $\calA$ of $A \subseteq [n]$. We give an algorithm for sparsifying $\calL(\vc \circ \va)$ down to $\approx |S|$ edges. 

We next provide some notation for capturing the structure of $\vc \circ \va$. Let the pieces of $\calA$ be denoted $\{A_i\}_{i \in [k]}$. Without loss of generality, if there are $\ell$ pieces with nonempty intersection with $S$, let them be $\{A_i\}_{i \in [\ell]}$, i.e., the first $\ell$ pieces. Finally, denote $S_i \defeq A_i \cap S$ for all $i \in [\ell]$. Then,
\begin{equation}\label{eq:biclique_sum}
\vc \circ \va = \sum_{i \in [\ell]} \sum_{j \in [i - 1]} \1_{K_{ij}},
\end{equation}
where $K_{ij}$ is the \emph{complete unweighted bipartite graph} between $S_i$ and $S_j$, and $\1_{K_{ij}} \subseteq \{0, 1\}^E$ indicates  its edges. Correspondingly we denote each complete unweighted bipartite graph Laplacian by
\[\ml_{ij} \defeq \calL(\1_{K_{ij}}), \text{ for all } 1 \le j < i \le \ell.\]
We next require a basic combinatorial fact on balanced partitions of integers.

\begin{lemma}\label{lem:balanced_partition}
Let $\{n_i\}_{i \in [\ell]} \subset \N$ and let $Z \defeq \sum_{i \in [\ell]} n_i$. Then if there exists no $i \in [\ell]$ with $n_i > \frac{Z}{3}$, there exists a partition of $[\ell]$ into $S, T = [\ell] \setminus S$ such that $\sum_{i \in S} n_i \in [\frac 1 3 Z, \frac 2 3 Z]$.
\end{lemma}
\begin{proof}
Iterate over the elements, adding them to $S$, until the next element would cause the sum of $S$ to exceed $\frac{2Z}{3}$. Because each element is bounded by $\frac Z 3$, $\sum_{i \in S} n_i \in [\frac 1 3 Z, \frac 2 3 Z]$ holds at termination.
\end{proof}

We also will use a fast sparsification algorithm for complete bipartite graphs.

\begin{lemma}\label{lem:sparsify_complete_bipartite}
    Let $G$ be a complete unweighted bipartite graph with bipartition $V = L \cup R$, and let $\delta \in (0, 1)$, and let $\ml_G \in \R^{V \times V}$ denote its graph Laplacian. There is an algorithm that runs in time $O(|V| \log(\frac{|V|}{\delta}))$ and returns a graph Laplacian $\ml_H$ such that with probability $\ge 1 - \delta$,
    \[\ml_H \approx_1 \ml_G,\; \nnz\Par{\ml_H} = O\Par{|V|\log\Par{\frac{|V|}{\delta}}}.\]
\end{lemma}
\begin{proof}
Let $m = O(|V| \log(\frac{|V|}{\delta}))$ denote the allowed sparsity parameter. For a sufficiently large constant in $m$, Theorem 1 of \cite{SpielmanS11} shows that sampling $m$ random edges of $G$ with replacement, and adding them to $H$ with weight $\frac{|L||R|}{m}$, gives a Laplacian $\ml_H$ meeting the desired criteria. Here, we used that all the effective resistances of the edges in $G$ are the same, by symmetry.
\end{proof}

We now give our overall sparsification algorithm for the graph with Laplacian $\ml(\vc \circ \va)$. The main idea is to recursively apply Lemma~\ref{lem:sparsify_complete_bipartite} across balanced partitions whenever possible.

\begin{lemma}\label{lem:sparsify_soc_asoc}
Let $\{S_i\}_{i \in [\ell]}$ be a partition of $S \subseteq [n]$, let $\delta \in (0, 1)$, and let 
\[\ml \defeq \sum_{i \in [\ell]} \sum_{j \in [i - 1]} \ml_{ij}, \]
where $\ml_{ij}$ is the graph Laplacian of the complete unweighted bipartite graph between $S_i$ and $S_j$. There is an algorithm that runs in time $O(|S|\log(n)\log(\frac n \delta))$, and returns a graph Laplacian $\tml$ such that with probability $\ge 1 - \delta$,
\[\tml \approx_1 \ml,\; \nnz(\tml) = O\Par{|S|\log(n)\log\Par{\frac n \delta}}.\]
\end{lemma}
\begin{proof}
Throughout the proof, let $C = \Theta(\log(\frac n \delta))$ with a large enough constant, such that when we apply Lemma~\ref{lem:sparsify_complete_bipartite} at most $n$ times, we can assume that all calls succeed, run in time $\le C|V|$ if called with input vertex set $V$, and return graphs with sparsity at most 
$C|V|$. We condition on this event for the remainder of the proof which gives the failure probability.

Our sparsification algorithm proceeds recursively in one of two ways as follows. Consider the sizes of partition pieces $\{n_i\}_{i \in [\ell]}$, where $n_i \defeq |S_i|$ for all $i \in [\ell]$, and let $Z \defeq |S| \le n$. 

If there is some partition piece $i \in [\ell]$ with $n_i > \frac Z 3$, then we form a bipartition where one side consists of $S_i$ and the other side consists of $\{S_j\}_{j \in [\ell] \setminus \{i\}}$. Note that all edges cross this bipartition are present in $\ml$, due to each term $\ml_{ij}$.
Thus, we sparsify this complete unweighted bipartite graph using Lemma~\ref{lem:sparsify_complete_bipartite}, and recurse on the remaining partition pieces $\{S_j\}_{j \in [\ell] \setminus \{i\}}$.

Otherwise, if there is no such partition piece, then by Lemma~\ref{lem:balanced_partition} there exists a balanced partition into two subsets of $S$ containing between $\frac Z 3$ and $\frac {2Z} 3$ vertices. We again sparsify the complete unweighted bipartite graph across this bipartition using Lemma~\ref{lem:sparsify_complete_bipartite}, and recurse on the two sides separately.

Overall, if we let $\calT(Z)$ denote the time it takes to implement this recursive strategy given an input partitioned set $S \subseteq [n]$ with size $Z = |S|$, we have shown the recursion
\begin{equation}\label{eq:time_recurse}
\calT(Z) \le \max_{Z' \in [\frac Z 3, \frac {2Z} {3}]}\calT\Par{Z'} + \calT\Par{Z - Z'} + CZ,
\end{equation}
because both cases of the recursion are handled by \eqref{eq:time_recurse} (in the first case discussed, there is only one recursion piece, which is dominated by the second case discussed). It is clear that \eqref{eq:time_recurse} solves to $\calT(Z) = O(CZ\log(Z))$, and the conclusion regarding the runtime follows from our choice of $C$ and using that $Z \le n$. A similar recursion holds for the sparsity bound.
\end{proof}

\subsection{Discretizing low-dimensional distances}\label{ssec:grid}

In this section, we show how to discretize a target implicit vector $\vg \in \R^E$ into a small number of SOCs or ASOCs via grid rounding, where $E$ indexes the edge set of the clique on $[n]$. Our algorithm will specifically apply to the case where $\vg$ is implicitly specified by distances between the columns $\{\vq_i\}_{i \in [n]} \subset \R^k$ of a matrix $\mq \in \R^{k \times n}$ in the following way:
\[\vg_e = \norm{\vq_u - \vq_v}_2^2, \text{ for all } e = (u, v) \in E.\]
Here, we think of $k$ as a small parameter, obtained via a low-dimensional embedding.

We give efficient algorithms for finding structured approximations of $\vg$ in the following senses. 

\begin{definition}[SOC approximation]\label{def:soc_approx}
    Let $\alpha, \beta, \gamma > 0$, $m \in \N$, and $\vg \in \R^{E}_{\ge 0}$. We say $\{\tvg^{(i)}\}_{i \in [m]} \subset \R^E_{\ge 0}$ is an $(\alpha,\beta,\gamma, m)$-SOC approximation to $\vg$ if the following hold.
    \begin{itemize}
    \item For all $i \in [m]$, $\tvg^{(i)}$ is a scalar multiple of a SOC.
        \item For $\vx \defeq \sum_{i \in [m]} \tvg^{(i)}$, we have $\0_E \leq \vx \leq \alpha \1_E$ entrywise.
        \item For any $e \in E$ where $\vg_e \leq 1$, $\vx_e \geq \beta$.
        \item For any $e \in E$ where $\vg_e > \gamma$, $\vx_e = 0$. 
    \end{itemize}
\end{definition}

\begin{definition}[ASOC approximation]\label{def:asoc_cost}
Let $\alpha, \beta, \gamma > 0$ and  for $\mq \in \R^{k \times n}$, let $\vg \in \R^E_{\ge 0}$ be defined as
\[\vg_e = \norm{\mq_{:u} - \mq_{:v}}_2^2 \text{ for all } e = (u, v) \in E.\]
We say $\{\tvg^{(i)}\}_{i \in [m]} \subset \R^E_{\ge 0}$ is a $(\alpha, \beta, \gamma, m)$-ASOC approximation to $\vg$ if the following hold.
\begin{itemize}
    \item For all $i \in [m]$, $\tvg^{(i)}$ is a scalar multiple of an ASOC.
    \item $\gamma \ge \max_{(u, v) \in E} \max_{j \in [k]} \Par{\mq_{ju} - \mq_{jv}}^2$.
    \item For all $i \in [m]$, $\tvg^{(i)} \le \beta \vg$ entrywise, and $\sum_{i \in [m]} \tvg^{(i)} \ge \vg^{(\ge \frac \gamma \alpha)}$ entrywise, where for $c > 0$,
    \[\Brack{\vg^{(\ge c)}}_{e} \defeq \sum_{j \in [k]} \Par{\mq_{ju} - \mq_{jv}}^2 \ind_{(\mq_{ju} - \mq_{jv})^2 \ge c}, \text{ for all } e = (u, v) \in E. \]
\end{itemize}
\end{definition}

\paragraph{SOC approximation.} We next describe our SOC approximation algorithm (Definition~\ref{def:soc_approx}), which uses a standard  grid rounding scheme provided below to form partitions of our point set. 

\begin{lemma}\label{lem:grid_rounding}
Let $\{\vq_i\}_{i \in [n]} \subset \R^k$ be given, and let $\rho > 0$. There is a randomized algorithm that in $O(nk\log(n))$ time returns a partition $\calP = \{S_j\}_{j \in [\ell]}$ of $[n]$, satisfying the following properties.
\begin{itemize}
    \item If $u, v \in S_j$ for some $j \in [\ell]$, then $\norms{\vq_u - \vq_v}_2 \le \rho \sqrt k$.
    \item Over the randomness of the algorithm, the probability that each pair $(u, v) \in [n] \times [n]$ do not belong to the same partition piece is at most $\frac{2\sqrt k}{\rho}\norm{\vq_u - \vq_v}_2$.
\end{itemize}
\end{lemma}
\begin{proof}
The algorithm works as follows: we discretize each dimension of $\R^k$ into intervals of length $\rho$, where the interval endpoints are shifted by a random amount chosen uniformly in $[0, \rho)$. We then place two points in the same partition piece if and only if they lie in the same grid box in $\R^k$. It is clear that the distance bound of $\rho\sqrt k$ holds for points in the same grid box. 

We now prove the second claim. First suppose that $\norms{\vq_u - \vq_v}_2 \le \frac \rho 2$. Then by independence, the probability that $\vq_u$ and $\vq_v$ are in the same partition piece is the product of probabilities that they are not separated in any dimension, which we can lower bound by
\begin{align*}\prod_{i \in [k]} \Par{1 - \frac{\Abs{\Brack{\vq_u - \vq_v}_i}}{\rho}} &\ge \prod_{i \in [k]}\exp\Par{-\frac{2\Abs{\Brack{\vq_u - \vq_v}_i}}{\rho}} \\
&= \exp\Par{-\frac 2 \rho \norm{\vq_u - \vq_v}_1} \\
&\ge 1 -\frac 2 \rho \norm{\vq_u - \vq_v}_1 \ge 1 -\frac {2\sqrt k} \rho \norm{\vq_u - \vq_v}_2.\end{align*}
Thus they are separated with probability at most $\frac {2\sqrt k} \rho \norm{\vq_u - \vq_v}_2$. In the other case of $\norm{\vq_u - \vq_v}_2 > \frac \rho 2$, the bound of $\frac{2\sqrt k}{\rho}\norm{\vq_u - \vq_v}_2$ is trivial as it exceeds $1$.

For the runtime, it is enough to first compute the index of each point, i.e., a $k$-tuple of integers specifying grid boxes in each dimension, in $O(nk)$ time. We can then find all unique grid boxes by sorting the indices in lexicographical order, which takes $O(nk\log(n))$ time.
\end{proof}

Given this result, our SOC approximation follows straightforwardly.

\begin{lemma}\label{lem:soc_algo}
Let $\{\vq_i\}_{i \in [n]} \subset \R^k$, $\beta > 0$, $\delta \in (0, 1)$, and for all $e = (u, v) \in E$ suppose
\[\vg_e = \norm{\vq_u - \vq_v}_2^2.\]
Then we can compute $\{\tvg^{(i)}\}_{i \in [m]}$, an $(\alpha, \beta, \gamma, m)$-SOC approximation to $\vg$, for
\[\alpha = \beta m,\; \gamma = 16k^2,\; m = \left\lceil 2\log_2\Par{\frac n \delta}\right\rceil,\]
with probability $\ge 1 - \delta$ in time $O(nk\log(n)\log(\frac n \delta))$.
\end{lemma}
\begin{proof}
Let $\rho = \sqrt{\gamma / k}$. We repeatedly call Lemma~\ref{lem:grid_rounding} with the given value of $\rho$, for $m$ times in total. For the $i^{\text{th}}$ call, we let $\tvg^{(i)}$ be defined as follows: for all $(u, v) \in E$,
\begin{equation}\label{eq:clique_def}
\tvg^{(i)}_{(u, v)} \defeq \begin{cases}
\beta    & (u, v) \text{ belong to the same partition piece}, \\
0   & \text{ otherwise}
\end{cases}.
\end{equation}
Each $\tvg^{(i)}$ defined via \eqref{eq:clique_def} is a scalar multiple of a SOC. Also, because $\vx$ is a sum of $m$ vectors from $\{0, \beta\}^E$, the bound $\alpha = \beta m$ clearly holds. This gives the first two conditions in Definition~\ref{def:soc_approx}.

To obtain the third condition in Definition~\ref{def:soc_approx}, consider some $e = (u, v) \in E$ such that $\vg_e \le 1$, so $\norms{\vq_u - \vq_v}_2 \le 1$. Then $u$ and $v$ belong to the same partition piece with probability at least 
\[1 - \frac{2\sqrt k}{\rho} = 1 - \frac{2k}{\sqrt \gamma} \ge \half.\]
Thus, repeating $m$ times guarantees that some $\tvg^{(i)}_e = \beta$ except with probability $\frac{\delta}{n^2}$, and then union bounding over all edges in $E$ gives the overall failure probability.

To obtain the last condition in Definition~\ref{def:soc_approx}, by the first guarantee in Lemma~\ref{lem:grid_rounding}, if some coordinate $\vg_e$ exceeds $\gamma$, i.e., $\norms{\vq_u - \vq_v}_2 > \sqrt \gamma = \rho \sqrt k$ where $e = (u, v)$, then $u$ and $v$ will never be partitioned together and thus $\vx_e = 0$ always. This does not contribute to the failure probability.
\end{proof}

\paragraph{ASOC approximation.} We now similarly give our ASOC approximation of distance vectors. We begin with a one-dimensional partitioning result.

\begin{lemma}\label{lem:1d_grid}
Let $\{q_i\}_{i \in [n]} \subset \R$ be given, and let $\rho > 0$. There is an algorithm that in $O(n\log(n))$ time returns a set $S \subseteq [n]$, as well as a partition $\{S_j\}_{j \in [\ell]}$ of $S$, satisfying the following properties.
\begin{enumerate}
    \item If $(u, v) \in [n] \times [n]$ have $|q_u - q_v| \in [\frac 3 2 \rho, 2\rho)$, then with probability at least $\frac 1 4$, we have that $\{u, v\} \in S$, and $u, v$ belong to different pieces of the partition.
    \item If $(u, v) \in [n] \times [n]$ have $u \neq v$ and $|q_u - q_v| \le \rho$, then either $\{u, v\} \subseteq S_j$ for some $j \in [\ell]$, or $S$ contains at most one of $u$ or $v$.
\end{enumerate}
\end{lemma}
\begin{proof}
The algorithm is similar to Lemma~\ref{lem:grid_rounding}. We first tile $\R$ with intervals of length $\rho$ (with a random offset in $[0, \rho)$), and we color the intervals either ``black'' or ``white'' in alternating fashion. Then with probability $\half$ we flip black and white intervals. We let $S$ consist of the points in black intervals, and the partition pieces correspond to points falling in the same black interval. The runtime is immediate by sorting the list of interval indices. The second condition is also clear because if $|q_u - q_v| \le \rho$, then the two points lie in either the same interval or adjacent intervals.

For the first condition, it is clearly enough to lower bound the probability that $\{u, v\} \in S$, because they cannot fall in the same piece of the partition if this is the case (as $|q_u - q_v| > \rho$). Without loss of generality, suppose that $q_u = 0$ and $q_v \in [\frac 3 2 \rho, 2\rho)$. Let $s \in [0, \rho)$ be the uniformly random shift, so that $0 \in (s - \rho, s]$. With probability $\half$, this is a black interval. Moreover,
\[\Pr_{s \simunif [0, \rho)}\Brack{q_v \in (s + \rho, s + 2\rho]} = \Pr_{s \simunif [0, \rho)}\Brack{s \in [q_v - 2\rho, q_v - \rho)} \ge \frac 1 2 .\]
Thus, the probability that $(s - \rho, s]$ is black and $q_v$ falls in the next black interval is at least $\frac 1 4$.
\end{proof}

By repeatedly applying Lemma~\ref{lem:1d_grid}, we obtain an ASOC approximation when $k = 1$.

\begin{lemma}\label{lem:1d_asoc}
Let $\{q_i\}_{i \in [n]} \subset \R$ be given, let $\beta \ge 4$, $\gamma \ge \max_{(u, v) \in E} |q_{u} - q_v|^2$, $\alpha \ge 1$, $\delta \in (0, 1)$, and for all $e = (u, v) \in E$ suppose
\[\vg_e = \Par{q_u - q_v}^2.\]
Then we can compute $\{\tvg^{(i)}\}_{i \in [m]}$, an $(\alpha, \beta, \gamma, m)$-ASOC approximation to $\vg$, for
\[m = O\Par{\log\Par{\frac 1 \alpha}\log\Par{\frac n \delta}}\]
with probability $\ge 1 - \delta$ in time $O(n\log(n)\log(\frac 1 \alpha)\log(\frac n \delta))$.
\end{lemma}
\begin{proof}
First, let $\calP = \{\rho_a\}_{a \in [p]} \subset \R_{> 0}$ be a set of candidate distances, such that $p = O(\log(\frac 1 \alpha))$ and for every $\vg_e \in [\frac \gamma \alpha, \gamma]$, there is some $a \in [p]$ such that $\sqrt{\vg_e} \in [\frac 3 2 \rho_a, 2\rho_a)$. This is doable by e.g., setting $\rho_1^2$ to be $\frac {4\gamma} {9\alpha}$, and then increasing $\rho_a$ by multiples of $1.1$ until $\rho_p^2 \ge \frac \gamma 4$.

Next, for each $\rho_a \in \calP$, we repeatedly call Lemma~\ref{lem:1d_grid} with $\rho \gets \rho_a$, for $\frac m p \defeq \lceil 8\log(\frac n \delta) \rceil$ times in total. Let $\{S_j\}_{j \in [\ell]}$ denote the partition pieces from the $i^{\text{th}}$ call. Then we define for all $(u, v) \in E$,
\[\tvg^{(i, a)}_{(u, v)} \defeq \begin{cases}
\beta\rho_a^2 & \{u, v\} \subset S \text{ and } u \in S_j,\; v \in S_{j'} \text{ where } j \neq j' \\
0 & \text{ else}
\end{cases}.\]
In other words, we assign a weight of $\beta \rho_a^2$ to all vertex pairs  in different partition pieces, and zero out all other weights. We return the collection of $\{\tvg^{(i, a)}\}_{a \in [p], i \in [\frac m p]}$ as our ASOC approximation. Then every $\tvg^{(i, a)}$ is a scalar multiple of an ASOC, and for all $i \in [m]$, $a \in [p]$, $\tvg^{(i, a)} \le \beta \vg$ edgewise, because the only way $\tvg^{(i, a)}_e > 0$ is if $\vg_e > \rho_a^2$ by the second claim in Lemma~\ref{lem:1d_grid}.

Finally we need to show that with probability $\ge 1 - \delta$, we have that 
\begin{equation}\label{eq:coverage}\sum_{(i, a) \in [\frac m p] \times [p]}\tvg^{(i, a)} \ge \vg^{(\ge \frac \gamma \alpha)}.\end{equation}
We argue this edgewise. Let $e = (u, v) \in E$ have $\vg_e \ge \frac \gamma \alpha$, and let $a$ be the corresponding index satisfying $\sqrt{\vg_e} \in [\frac 3 2 \rho_a, 2\rho_a]$. By the first claim in Lemma~\ref{lem:1d_grid}, after running $\frac m p$ independent trials, there is at most a $\frac \delta {n^2}$ chance that $u, v$ never fall in different black intervals. Moreover, if $u, v$ fall in different black intervals in the $i^{\text{th}}$ trial at distance $\rho_a$, then
\[\tvg^{(i, a)}_e = \beta \rho_a^2 \ge 4\rho_a^2 \ge \vg_e. \]
Thus, by union bounding over all $|E| \le n^2$ edges, \eqref{eq:coverage} holds except with probability $\delta$.
\end{proof}

We can now apply Lemma~\ref{lem:1d_asoc} multiple times to generalize our result to distance vectors in $\R^k$.

\begin{lemma}\label{lem:kd_asoc}
Let $\{\vq_i\}_{i \in [n]} \subset \R^k$, $\beta \ge 4$, $\delta \in (0, 1)$, and for all $e = (u, v) \in E$ suppose
\[\vg_e = \norm{\vq_u - \vq_v}_2^2.\]
Also, let
\[\gamma \ge \max_{(u, v) \in E} \max_{j \in [k]} \Par{\Brack{\vq_u}_j - \Brack{\vq_v}_j}^2,\; \alpha \ge 1.\]
Then we can compute $\{\tvg^{(i)}\}_{i \in [m]}$, an $(\alpha, \beta, \gamma, m)$-ASOC approximation to $\vg$, for
\[m = O\Par{k\log\Par{\frac 1 \alpha}\log\Par{\frac {nk} \delta}},\]
with probability $\ge 1 - \delta$ in time $O(nk\log(n)\log(\frac 1 \alpha)\log(\frac {nk} \delta))$.
\end{lemma}
\begin{proof}
We apply Lemma~\ref{lem:1d_asoc} separately to each coordinate, adjusting the failure probability by a $k$ factor, and return the union of the $k$ resulting sets as our ASOC approximation. The condition that $\tvg^{(i)} \le \beta\vg$ coordinatewise always holds because $\vg_{(u, v)} \ge |[\vq_u]_j - [\vq_v]_j|^2$ for any coordinate $j \in [k]$. Similarly, $\sum_{i \in [m]} \tvg^{(i)} \ge \vg^{(\ge \frac \gamma \alpha)}$ holds by summing over the contributions of each coordinate.
\end{proof}

\subsection{Long step packing SDP}\label{ssec:packing}

In this section, we give an algorithm for approximately solving packing semidefinite programs. Our algorithm is inspired by existing counterparts in the literature \cite{Allen-ZhuLO16, PengTZ16, JambulapatiLT20}, but differs in a few ways. On the one hand, it only gives very coarse multiplicative approximation guarantees instead of $1 + \eps$ factors for arbitrarily small $\eps$. On the other hand, it only uses highly-structured steps to update iterates, which support the efficient access tools from Sections~\ref{ssec:soc} and~\ref{ssec:grid}.

Throughout this section, we consider the following optimization problem for semidefinite matrices $\{\ma_e\}_{e \in E} \subset \PSD^{n \times n}$ for an index set $E$, and $\vc \in \{0, 1\}^E$:
\begin{equation}
\label{eq:packing}
\max_{\substack{\vx \in \R^E_{\ge 0} \\ \sum_{e \in E} \vx_e \ma_e \preceq \mi_n}} \vc^\top \vx.
\end{equation}
We are specifically interested in the specialization of \eqref{eq:packing} to
\begin{equation}\label{eq:packing_lap}
\ma_e = \mmp \ml_e \mmp \text{, where } \ml_e \defeq \Par{\ve_u - \ve_v}\Par{\ve_u - \ve_v}^\top \text{ for all } e = (u, v) \in E,
\end{equation}
where $\mmp$ supports matrix-vector query access, and $E$ indexes unordered $(u, v) \in [n] \times [n]$.
Analogously to Section~\ref{ssec:soc}, we define $\calA: \R^E \to \Sym^{n \times n}$ and $\calA^*: \Sym^{n \times n} \to \R^E$ as 
\[
\calA(\vv) = \sum_{e \in E} \vv_e \ma_e \quad \text{and} \quad \calA^*(\my) =\{ \inprod{\ma_e }{\my} \}_{e \in E}.
\] 

We split our discussion into three parts. We first propose an algorithm for solving the \emph{decision variant} of \eqref{eq:packing}, that guesses an optimal value and certifies whether it is approximately achievable. Our algorithm uses a certain step oracle (Definition~\ref{def:step}), which we then discuss how to implement using Lemma~\ref{lem:soc_algo}. Finally, we reduce the optimization variant \eqref{eq:packing} to the decision variant.

\paragraph{Decision variant: correctness.} In Algorithm~\ref{alg:packing}, we state our general solver framework for the decision variant of \eqref{eq:packing}. It uses the following type of oracle to implement its steps.

\begin{definition}[Step oracle]\label{def:step}
We say that $\mathcal{O}: \PSD^{n \times n} \to \R^E$ is an \emph{$(\alpha, \beta, \gamma,m)$-step oracle} for $\{\ma_e\}_{e \in E}$ if on input $\mz \in \PSD^{n \times n}$, $\vx = \mathcal{O}(\mz)$ is a $m$-SOC and the following hold.
    \begin{itemize}
        \item $\0_E \leq \vx \leq \alpha \1_E$ entrywise.
        \item For any $e \in E$ where $[\calA^*(\mz)]_e \leq 1$, $\vx_e \geq \beta$.
        \item For any $e \in E$ where $[\calA^*(\mz)]_e > \gamma$, $\vx_e = 0$. 
    \end{itemize}
\end{definition}

\begin{algorithm2e}
	\caption{$\SOCPack(\oracle, \vc, p, T)$}
	\label{alg:packing}
	\DontPrintSemicolon
		\codeInput $\oracle$, an $(\alpha, \beta, \gamma, m)$-step oracle for $\{\ma_e\}_{e \in E}$, $\vc \in \{0, 1\}^E$, $p \ge 2$, $T \in \N$\;
  $\vw_0 \gets \vc$\;
  \For{$0 \le t < T$}{
   $\my_t \gets \frac{\calA(\vw_t)}{\norm{\calA(\vw_t)}_p}$ \; 
   $\step_t \gets \oracle(\my_t^{p-1})$\;\label{line:oracle_step}
   $\vw_{t+1} = \vw_t \circ (\1_E +\step_t)$\;
   \If{$\vc^\top \vw_{t+1} \geq \beta^{\frac T 2}$}{
        \codeReturn $\frac{\vw_{t+1}}{\vc^\top \vw_{t+1}}$\label{line:packing_return}
   }
  }
\end{algorithm2e} 

We proceed to analyze the correctness of Algorithm~\ref{alg:packing}, discussing how to implement $\oracle$ later in this section. We begin with a helper lemma.

\begin{lemma}\label{lem:taylor_step}
    For any $p \ge 2$ and $\mm \in \PSD^{n \times n}$ with $\mm \preceq \alpha \id_n$, $(\mi_n + \mm)^p \preceq \id_n + p (1+\alpha)^{p-1} \mm$. 
\end{lemma}
\begin{proof}
Because $\id_n$ commutes with $\mm$, we may diagonalize both sides; the claim reduces to showing $(1+x)^p \leq 1 + p (1 + \alpha)^{p-1} x$ for $x \in [0,\alpha]$. This follows from the intermediate value theorem.
\end{proof}

We now make a simple observation on the sparsity pattern of the iterates $\vw_t$ in Algorithm~\ref{alg:packing}.
\begin{lemma}
\label{lemma:sparsity}
For any $0 \le t < T$ and $e \in E$, $[\vw_{t}]_e$ is nonzero if and only if $\vc_e$ is. 
\end{lemma}
\begin{proof}
We proceed by induction on $t$: the base case is trivial since $\vw_0 = \vc$. For any other iteration $t$, observe $[\vw_{t+1}]_e = [\vw_{t}]_e (1 + [\step_t]_e)$ for all $e \in E$, and since $\step_t \geq 0$ the inductive claim is obvious.
\end{proof}

With this, we show that a certain potential function is nonincreasing throughout the algorithm.
\begin{lemma}
\label{lemma:potential}
In the setting of Algorithm~\ref{alg:packing}, define the function
\[
\Phi_t = \norm{\calA(\vw_t)}_p - (1+\alpha)^{p-1} \gamma \vc^\top \vw_t, \text{ for all } 0 \le t < T.
\]
Then for any $0 \le t < T$, we have $\Phi_{t+1} \leq \Phi_t$.
\end{lemma}
\begin{proof}
 For simplicity, we drop the index $t$ and let $\vw \defeq \vw_t$, $\my \defeq \my_t$. Define the matrices
\[\mm_0 \defeq \calA(\vw) ,\; \mm_1 \defeq \calA(\step \circ \vw).\]
By the Lieb-Thirring inequality, i.e., $\Tr((\ma \mb\ma)^p) \le \Tr(\ma^{2p}\mb^p)$, we have
\[\norm{\mm_0 + \mm_1}_p^p = \Tr\Par{\left(\mm_0 + \mm_1\right)^p} \le \Tr\Par{\mm_0^p\left(\id_n + \mm_0^{-\half}\mm_1\mm_0^{-\half}\right)^p}.\]

Next, noting that $\step \leq \alpha \1_E$ by definition of the step oracle, applying Lemma~\ref{lem:taylor_step} gives 
\[\norm{\mm_0 + \mm_1}_p^p \le \Tr\left(\mm_0^p + p (1+\alpha)^{p-1} \mm_0^{p - 1} \mm_1 \right). \]
If we define $\my = \frac{\mm_0}{\norm{\mm_0}_p}$, we observe
\begin{align*}\norm{\mm_0 + \mm_1}_p^p &\le \Tr\Par{\mm_0^p + p (1+\alpha)^{p-1} \mm_0^{p - 1} \mm_1 }\\
&=\norm{\mm_0}_p^p\left(1 + p(1+\alpha)^{p-1}\inprod{\my^{p - 1}}{\frac{\mm_1}{\norm{\mm_0}_p} }\right). \end{align*}
By using $1 + px \le (1 + x)^p$ for all $p \ge 2$ and $x \ge 0$, taking $p^{\text{th}}$ roots we thus have
\[\norm{\mm_0 + \mm_1}_p \le \norm{\mm_0}_p + (1+\alpha)^{p-1}\inprod{\my^{p - 1}}{\mm_1}. \]
Now via linearity of trace,
\begin{equation}
\label{eq:potential1}
\inprod{\my^{p - 1}}{\mm_1} = \sum_{e \in E} \inprod{\my^{p - 1}}{\ma_e} \step_e \vw_e \le \sum_{e \in E} \gamma \step_e \vw_e
\end{equation}
as $\inprod{\my^{p - 1}}{\ma_j} > \gamma$ implies $\step_e = 0$. We now claim 
\[
\sum_{e \in E} \step_e \vw_e = \sum_{e \in E} \vc_e \step_e \vw_e.
\]
This follows immediately from the observation that $\vw_e \neq 0$ implies $\vc_e =1$ by Lemma~\ref{lemma:sparsity} and the fact that all $\vc_e \in \{0,1\}$. Combining this with \eqref{eq:potential1} gives the desired result.
\end{proof}
We augment this potential bound by showing that  after $T$ steps of the algorithm we must be able to certify either a primal or dual bound for the quality of a packing solution.
\begin{lemma}
\label{lemma:dual}
   Let $p \ge 2$ and $q \defeq \frac p {p - 1}$. After $T$ steps of Algorithm~\ref{alg:packing}, define $\bmy = \frac{1}{T} \sum_{0 \le t < T} \my_t^{p - 1}$ and $\vx = \frac{\vw_T}{\vc^\top \vw_T}$. Then at least one of the following is true.
    \begin{itemize}
        \item $\calA^*(\bar{\my}) \geq \frac{1}{2} \vc$ and $\norms{\bmy}_q \le 1$.
        \item 
        The algorithm terminates on Line~\ref{line:packing_return} in some iteration.
    \end{itemize} 
\end{lemma}
\begin{proof}
First we observe that all $\norm{\my_t}_q = 1$ by inspection, so by convexity of norms, we indeed have $\norms{\bmy}_q \le 1$. Similarly, we have $\vc^\top \vx = 1$ by definition.

Assume that the first claim is false, and there exists some $e \in E$ such that $\langle \ma_e, \bar{\my} \rangle < \frac{1}{2} \vc_e$. We show that the second claim must hold. Because $\langle \ma_e, \my_j^{p-1} \rangle \ge 0$ for any $e \in E, 0 \le t < T$, 
\[
0 \leq  \langle \ma_e, \bar{\my} \rangle < \frac{1}{2} \vc_e,
\]
and thus $\vc_e = 1$ since $\vc \in \{0, 1\}^E$. Applying Markov's inequality, we have that at least $\frac T 2$ iterations $0 \le t < T$ have $\langle \ma_e, \my_t^{p-1} \rangle \leq  \vc_i$ by Markov's inequality. By the guarantee of $\oracle$, every such step $t$ must grow $[\vw_t]_e$ by a factor of $\beta$ since $\langle \ma_e, \my_t^{p-1} \rangle =  \calA^*(\my_t^{p-1})_e = \vg_e$. With this, we have 
\begin{align*}
\vc^\top \vw_T \geq \vc_e [\vw_T]_e \geq \beta^{\frac T 2} \vc_e [\vw_0]_e = \beta^{\frac T 2},
\end{align*}
so if the algorithm was allowed to run for $T$ iterations, it must terminate as claimed.
\end{proof}

We provide an alternative characterization of the first case in Lemma~\ref{lemma:dual} to help apply the result.

\begin{lemma}\label{lem:dual_certify}
Let $p \ge 2$ and $q \defeq \frac p {p - 1}$. If there exists $\mz \in \PSD^{n \times n}$ with $\norm{\mz}_q \le 1$ and $\calA^*(\mz) \ge \half \vc$,
\[\max_{\substack{\vx \in \R^E_{\ge 0} \\ \calA(\vx) \preceq \id_n}} \vc^\top \vx \le 2n^{\frac 1 p}.\]
\end{lemma}
\begin{proof}
Assume $\norm{\mz}_q \le 1$ and $\calA^*(\mz) \ge \half \vc$, and suppose that there exists some $\vx \in \R^E_{\ge 0}$ satisfying $\vc^\top \vx > 2$ and $\norm{\calA(\vx)}_p \le 1$. Then by H\"older's inequality, we have a contradiction:
\[ 1 < \half \inprod{\vx}{\vc} \le \inprod{\vx}{\calA^*(\mz)} = \inprod{\mz}{\calA(\vx)} \le \norm{\calA(\vx)}_p \le 1.\]
Thus, all $\vx \in \R^E_{\ge 0}$ with $\norm{\calA(\vx)}_p \le 1$ must have $\vc^\top \vx \le 2$. This implies the conclusion upon performing a norm conversion from $\norm{\cdot}_p$ to $\norm{\cdot}_\infty$ in the constraint, losing a $n^{\frac 1 p}$ factor.
\end{proof}

\paragraph{Decision variant: implementation.} We now turn to implementation. We begin by providing an efficient step oracle $\oracle$ for the matrices $\my_t^{p - 1}$ in Line~\ref{line:oracle_step}. It will first be helpful to collapse our true oracle input to a smaller dimension. To motivate this, Algorithm~\ref{alg:packing} requires approximating
\begin{equation}\label{eq:distance_p}\inprod{\my^{p - 1}}{\mmp \ml_e \mmp} = \norm{\my^{\frac{p - 1}{2}}\mmp \Par{\ve_u - \ve_v}}_2^2, \end{equation}
for all $e = (u, v) \in E$. This makes it clear that this is a squared distance between columns of $\my^{\frac{p - 1}{2}}\mmp $, which are natively $n$-dimensional vectors that we can only implicitly access. Using the Johnson-Lindenstrauss lemma, we can embed our vectors into $k \approx \log(n)$ dimensions at a constant factor approximation loss, which drastically improves parameters when applying tools from Section~\ref{ssec:grid}.

\begin{lemma}\label{lem:jl_embed_schatten}
Following notation in Algorithm~\ref{alg:packing}, let $\delta \in (0, 1)$, suppose that $p \ge 2$ is an odd integer, and let $\my = \frac{\mmp \calL(\vw) \mmp}{\norm{\mmp \calL(\vw) \mmp}_p}$ for $\vw \in \R^E$. In time
\[O\Par{\Par{\tmv\Par{\mmp} + \tmv\Par{\calL(\vw)}}\cdot p\log\Par{\frac n \delta}}\]
we can output $\mq \in \R^{k \times n}$ where $k = O(\log(\frac n \delta))$, such that if we define $\vg'_e \defeq \norm{\mq_{:u} - \mq_{:v}}_2^2$ for all $e = (u, v) \in E$, we have with probability $\ge 1 - \delta$ that
\[ \half \calA^*(\my^{p - 1}) \le \vg' \le \calA^*(\my^{p - 1}) \text{ entrywise.}\]
\end{lemma}
\begin{proof}
First, we claim that we can output a number $Z$ such that
\[Z \le \norm{\mmp \calL(\vw) \mmp}_p^p = \Tr\Par{\Par{\mmp \calL(\vw) \mmp}^p}\le 1.1 Z,\]
within the allotted time, with failure probability $\le \frac \delta 2$. This is immediate from the Johnson-Lindenstrauss lemma, see e.g., its application in Lemma 2 of \cite{JambulapatiLMSST23}. This also implies that
\[ Z^{\half - \frac 1 {2p}} \le \norm{\mmp \calL(\vw)\mmp}_p^{\frac{p - 1} 2} \le 1.1 Z^{\half - \frac 1 {2p}}.\]
Next, by the Johnson-Lindenstrauss lemma (see e.g., Lemma 3 of \cite{JambulapatiLMSST23}), letting $k = O(\log(\frac n \delta))$ with a sufficiently large constant, and $\mg \in \R^{k \times n}$ have i.i.d.\ entries $\sim \Nor(0, \frac 1 k)$, we have that
\begin{align*}0.9 \norm{\Par{\mmp \calL(\vw) \mmp}^{\frac{p - 1}{2}} \mmp \Par{\ve_u - \ve_v}}_2^2 &\le \norm{\mg \Par{\mmp \calL(\vw) \mmp}^{\frac{p - 1}{2}} \mmp \Par{\ve_u - \ve_v}}_2^2 \\ 
&\le 1.1 \norm{\Par{\mmp \calL(\vw) \mmp}^{\frac{p - 1}{2}} \mmp \Par{\ve_u - \ve_v}}_2^2, \end{align*}
with probability $\ge 1 - \frac \delta 2$, for all $e = (u, v) \in E$. Thus combining our claims and recalling \eqref{eq:distance_p}, we have that with probability $\ge 1 - \delta$, it is enough to output $\mq = (1.1)^{-\half} Z^{\frac 1 {2p} - \half} \mg (\mmp \calL(\vw) \mmp)^{\frac{p - 1}{2}} $, which takes the stated amount of time, as it uses $O(pk)$ matrix-vector products through $\mmp \calL(\vw) \mmp$.
\end{proof}

We can now directly apply Lemma~\ref{lem:soc_algo} to implement our step oracle.
\begin{lemma}
\label{lem:steporacle}
Following notation in Algorithm~\ref{alg:packing}, let $\beta > 0$, $\delta \in (0, 1)$, suppose that $p \ge 2$ is an odd integer, and let $\my = \frac{\mmp \calL(\vw) \mmp}{\norm{\mmp \calL(\vw) \mmp}_p}$ for $\vw \in \R^E$. We can implement $\oracle$, an $(\alpha, \beta, \gamma, m)$-step oracle for $\my^{p - 1}$, for
\[\alpha = 2\beta m,\; \gamma = O\Par{\log^2\Par{\frac n \delta}},\; m = O\Par{\log\Par{\frac n \delta}},\]
with probability $\ge 1 - \delta$ in time
\[O\Par{\Par{\tmv\Par{\mmp} + \tmv\Par{\calL(\vw)}}\cdot p\log\Par{\frac n \delta} + n\log(n)\log^2\Par{\frac n \delta}}.\]
\end{lemma}
\begin{proof}
We apply Lemma~\ref{lem:soc_algo} (with $\delta \gets \frac \delta 2$) to $\mq$ where $\mq$ is the output of Lemma~\ref{lem:jl_embed_schatten} (with $\delta \gets \frac \delta 2$). The resulting oracle parameters at most grow by a factor of $2$ due to the lossiness of Lemma~\ref{lem:jl_embed_schatten}.
\end{proof}

We next show inductively that our iterates $\{\vw_t\}_{0 \le t < T}$ induce Laplacians $\calL(\vw_t)$ that support efficient matrix-vector queries by using tools for $K$-SOCs from Section~\ref{ssec:soc}.

\begin{lemma}\label{lem:soc_iterates}
Following notation in Algorithm~\ref{alg:packing}, for each $0 \le t < T$, $\vw_t = \vy_t \circ \vc$ where $\vy_t$ is a $(m + 1)^t$-SOC where $m = O(\log(\frac n \delta))$. Moreover, we can compute $\vy_t$ in $O(n\log(n) (m + 1)^t)$ time.
\end{lemma}
\begin{proof}
We proceed by induction on $t$. For the base case, $\vw_0 = \vc = \1_E \circ \vc$ and $\1_E$ is a $1$-SOC. Next, for some $t \ge 0$, assume that $\vw_t$ can be written as 
\[\sum_{i \in [(m + 1)^t]} \vv_i \circ \vc,\]
where each $\vv_i$ is a $1$-SOC. Further, note that by the step oracle guarantee, each $\step_t = \sum_{j \in [m]} \vw_j$ is always a $m$-SOC, where each $\vw_j$ is a $1$-SOC. This proves the inductive guarantee, as
\begin{align*}\Par{\sum_{i \in [(m + 1)^t]} \vv_i \circ \vc} \circ \Par{\1_E + \sum_{j \in [m]} \vw_j} = \sum_{i \in [(m + 1)^t]} \vv_i \circ \vc + \sum_{\substack{i \in [(m + 1)^t] \\ j \in [m]}} \vv_i \circ \vw_j \circ \vc,\end{align*}
and the latter expression clearly has $(m + 1)^{t + 1}$ terms, each of which is either $\vv_i \circ \vc$ (where $\vv_i$ is a $1$-SOC), or $\vv_i \circ \vw_j \circ \vc$ (where $\vv_i \circ \vw_j$ is a $1$-SOC because of Lemma~\ref{lem:refine_partition}). Finally, note that we can compute each $\vv_i \circ \vw_j$ as a $1$-SOC using Lemma~\ref{lem:refine_partition} in time $O(n\log(n))$.
\end{proof}

We can now put together all the pieces developed thus far to give our analysis for Algorithm~\ref{alg:packing}.

\begin{lemma}\label{lem:decision_packing}
In the setting of Algorithm~\ref{alg:packing}, suppose $\normsop{\calA(\vc)} \le \rho$ for $\rho \ge 1$, let $\delta \in (0, 1)$, and let
\[p = \Theta\Par{\log^{\frac 1 3}\Par{n\rho}},\; T = \Theta\Par{\log^{\frac 2 3}\Par{n\rho}},\]
where $p$ is an odd integer. We have with probability $\ge 1 - \delta$ that if Algorithm~\ref{alg:packing} returns on Line~\ref{line:packing_return}, 
\begin{equation}\label{eq:opt_not_small}\max_{\substack{\vx \in \R^E_{\ge 0} \\ \calA(\vx) \preceq \id_n}} \vc^\top \vx \ge \exp\Par{-O\Par{\log^{\frac 2 3}(n\rho) \log\log\Par{\frac{n\rho}{\delta}}}},\end{equation}
and otherwise,
\begin{equation}\label{eq:opt_not_large}\max_{\substack{\vx \in \R^E_{\ge 0} \\ \calA(\vx) \preceq \id_n}} \vc^\top \vx \le \exp\Par{O\Par{\log^{\frac 2 3}(n\rho)}},\end{equation}
for appropriate constants.
The algorithm can be implemented to run in time
\begin{align*} 
O\Par{\tmv\Par{\mmp} \cdot \log^{\frac 5 3}\Par{\frac{n\rho}{\delta}} } + n\log(n)\log\Par{\frac {n\rho} \delta}^{O(\log^{2/3}(n\rho))},\end{align*}
and if it returns $\vx = \frac{\vw_{t + 1}}{\vc^\top \vw_{t + 1}}$ on Line~\ref{line:packing_return}, then
$\vx = \vc \circ \vy$ where $\vy$ is a $O(\log(\frac {nT} \delta))^T$-SOC.
\end{lemma}
\begin{proof}
By induction, Lemma~\ref{lem:soc_iterates} shows that we can maintain each $\vw_t$ within the allotted time as a $(m + 1)^T$-SOC, for some $m = O(\log(\frac{nT}{\delta}))$ (where we adjusted the failure probability by a $O(T)$ factor). Thus, Lemma~\ref{lem:steporacle} shows we can obtain an $(\alpha, \beta, \gamma, m)$-step oracle to implement Line~\ref{line:oracle_step} in all $T$ iterations in the stated time, with failure probability $\frac \delta 2$, and parameters
\[\beta = \exp\Par{\Theta\Par{\log^{\frac 1 3}(n\rho)}},\; \alpha = 2\beta m,\; \gamma = O\Par{\log^2\Par{\frac{nT}{\delta}}},\; m = O\Par{\log\Par{\frac{nT}{\delta}}}.\]
In the remainder of the proof, assume that all calls on Lemma~\ref{lem:steporacle} succeeded. It only remains to establish correctness of \eqref{eq:opt_not_small} and \eqref{eq:opt_not_large}. First, if the algorithm ever successfully returns on Line~\ref{line:packing_return} for some iteration $t + 1$, we have by inducting on Lemma~\ref{lemma:potential} that for $\vx \defeq \frac{\vw_{t + 1}}{\vc^\top \vw_{t + 1}}$,
\begin{align*}
\normop{\calA\Par{\vx}} &\le  \frac 1 {\vc^\top \vw_{t + 1}} \norm{\calA\Par{\vw_{t + 1}}}_p = \frac 1 {\vc^\top \vw_{t + 1}} \Phi_{t + 1} + \Par{1 + \alpha}^{p - 1} \gamma \\
&\le \frac 1 {\beta^{\frac T 2}} \Phi_0 + \exp\Par{O\Par{\log^{\frac 2 3}(n\rho)\log\log\Par{\frac{n\rho}{\delta}}}} = \exp\Par{O\Par{\log^{\frac 2 3}(n\rho)\log\log\Par{\frac{n\rho}{\delta}}}}, 
\end{align*}
once we plug in our choice of parameters. Using that $\vc^\top \vx = 1$ and normalizing by the final expression above proves \eqref{eq:opt_not_small}. Conversely, if the algorithm never returns on Line~\ref{line:packing_return}, then \eqref{eq:opt_not_large} follows from our choice of parameters, the first case in Lemma~\ref{lemma:dual}, and the characterization in Lemma~\ref{lem:dual_certify}.
\end{proof}

\paragraph{Optimization variant.} Lemma~\ref{lem:decision_packing} lets us certify whether the value of a problem \eqref{eq:packing} is very large or very small, at a given scale. We complete this section by wrapping this decision problem result in a binary search to solve the optimization variant of \eqref{eq:packing}.

\begin{proposition}\label{prop:fast_packing}
In an instance of \eqref{eq:packing} where \eqref{eq:packing_lap} holds, let $\delta \in (0, 1)$, let $\opt$ denote the value of  \eqref{eq:packing}, and suppose we know $\opt \in [\ell, u] \subset \R_{> 0}$. There is an algorithm running in time
\[O\Par{\tmv\Par{\mmp} \cdot \log^{2}\Par{\frac{nu}{\delta\ell}} } + n\exp\Par{O\Par{\log^{\frac 2 3}\Par{\frac{nu}{\ell}}\log\log\Par{\frac{n u}{\delta \ell}}}}. \]
that returns $\vx$ satisfying
\[\inprod{\vc}{\vx} \ge \exp\Par{-O\Par{\log^{\frac 2 3}\Par{\frac{nu}{\ell}}\log\log\Par{\frac{n u}{\delta \ell}}}}\opt,\; \calA\Par{\vx} \preceq \id_n.\]
Moreover, $\vx = \vv \circ \vc$ where $\vv$ is a $K$-SOC, for
\[K = \exp\Par{O\Par{\log^{\frac 2 3}\Par{\frac{nu}{\ell}}\log\log\Par{\frac{n u}{\delta \ell}}}}.\]
\end{proposition}
\begin{proof}
The proof follows by combining Lemma~\ref{lem:decision_packing} with the binary search in Proposition 5 of \cite{JambulapatiLMSST23}. In particular, there are $\log\log(\frac u \ell)$ phases of binary search, each calling Lemma~\ref{lem:decision_packing} once with a multiple of $\vc$. We remark that we never need to pass any scaling that does not satisfy the premise $\normsop{\calA(\vc)} \ge 1$, which we can certify up to a constant factor via the power method. This is because if $\normsop{\calA(\vc)} \le 1$, then \eqref{eq:opt_not_small} immediately holds by using the choice $\vx = \vc$. By the same logic, no scaling considered will ever have $\rho \ge \frac u \ell$, because we only ever truncate the range of interest.
\end{proof}

\subsection{Matrix dictionary recovery}\label{ssec:mdr}

In this section, we build upon Section~\ref{ssec:packing} to give a two-sided matrix dictionary recovery result. Our result approximates an unknown graph Laplacian, given very weak access in the forms of matrix-vector products and a preconditioned packing oracle (with one-sided guarantees).

\paragraph{JL embedding.} Our algorithm is based on the matrix multiplicative weights updates, which produces iterates of the form
\begin{equation}\label{eq:my_def}\my = \frac{\exp(-\ms)}{\Tr\exp(-\ms)}.\end{equation}
In our setting, the matrix $\ms$ has a special structure. Assume that we have matrix-vector product access to $\mmp \in \PSD^{n \times n}$, and let $\{\va_s, \vv_s\}_{s \in [S]} \in \R^E$ be such that all of the $\{\va_s\}_{s \in [S]}$ are ASOCs, and all of the $\{\vv_s\}_{s \in [S]}$ are $K$-SOCs. We are interested in supporting access to $\my$ in \eqref{eq:my_def}, for
\begin{equation}\label{eq:mmw_step_form}\begin{aligned}
\ms = \mmp \calL(\vx) \mmp ,
\text{ where } \vx &= \sum_{s \in [S]} \va_s \circ \vv_s.
\end{aligned}\end{equation}
To simplify notation, let $\tmy \defeq \mmp \my \mmp$. 
Our algorithm requires applying the techniques of Section~\ref{ssec:grid} to the vector $\vg \in \R^E$ defined by
\begin{equation}\label{eq:gdef}\vg_e = \inprod{\my}{\mmp \ml_e \mmp} = \inprod{\tmy}{\ml_e} \text{ for all } e = (u, v) \in E.\end{equation}
As in Section~\ref{ssec:packing}, we will treat each $\vg_e$ as a squared distance between columns of $\tmy^{\half}$, which we can embed into low dimensions. We first require the following guarantee on approximating $\Tr\exp(-\ms)$.

\begin{lemma}\label{lem:approx_trace}
Let $\vx \in \R^E$ have the form \eqref{eq:mmw_step_form} where all $\{\va_s\}_{s \in [S]}$ are ASOCs and all $\{\vv_s\}_{s \in [S]}$ are $K$-SOCs, and define $\ms, \my$ as in \eqref{eq:mmw_step_form}, \eqref{eq:my_def}. Assume $\normsop{\ms} \le R$, and let $\delta \in (0, 1)$. In time
\[O\Par{\Par{\tmv(\mmp) + nKS\log(n)}R\log\Par{\frac n \delta}},\]
we can output $Z \in \R$ satisfying $\frac 9 {10} \Tr\exp(-\ms) \le Z \le \Tr\exp(-\ms)$ with probability $\ge 1 - \delta$.
\end{lemma}
\begin{proof}
This follows from Lemma 2 of \cite{JambulapatiLMSST23} with $\eps \gets \frac 1 {30}$, $\kappa \gets 1$, and $R \gets R$, because $\tmv(\ms) = 2\tmv(\mmp) + \tmv(\calL(\vx))$, and $\tmv(\calL(\vx)) = O(nKT\log(n))$ using \eqref{eq:mmw_step_form} and Lemma~\ref{lem:matvec_soc}.
\end{proof}

We now give our main embedding result, which approximates $\vg$ entrywise in low dimensions.

\begin{lemma}\label{lem:jl_embed}
Let $\vx \in \R^E$ have the form \eqref{eq:mmw_step_form} where all $\{\va_s\}_{s \in [S]}$ are ASOCs and all $\{\vv_s\}_{s \in [S]}$ are $K$-SOCs, and define $\ms, \my, \vg$ as in \eqref{eq:mmw_step_form}, \eqref{eq:my_def}, \eqref{eq:gdef}. Assume $\normsop{\ms} \le R$, and let $\delta \in (0, 1)$. In time
\[O\Par{\Par{\tmv(\mmp) + nKS\log(n)}R\log\Par{\frac n \delta}},\]
we can output $\mq \in \R^{k \times n}$ where $k = O(\log(\frac n \delta))$, such that if we define $\vf_e \defeq \norms{\mq_{:u} - \mq_{:v}}_2^2$ for all $e = (u, v) \in E$, we have with probability $\ge 1 - \delta$ that $\half \vg \le \vf \le \vg$ entrywise.
\end{lemma}
\begin{proof}
As our first step, we obtain an estimate satisfying $\frac 9 {10} \Tr\exp(-\ms) \le Z \le \Tr\exp(-\ms)$ with probability $\ge 1 - \frac \delta 2$ from Lemma~\ref{lem:approx_trace}, within the allotted runtime.

Our next step is to approximate $\exp(-\ms)$ with a polynomial in $\ms$. Specifically, Theorem 4.1 of \cite{SachdevaV14} shows that because $\0_{n \times n} \preceq \ms \preceq R \id_n$, there is a matrix $\mm$ such that
\begin{equation}\label{eq:m_approx}\frac 9 {10} \mm \preceq \exp\Par{-\half\ms} \preceq \mm,\end{equation}
and $\mm$ is a degree-$O(R)$ polynomial in $\ms$. Because $\mm, \ms$ commute, this also shows that
\[\frac 4 5 \mm^2 \preceq \exp\Par{-\ms} \preceq \mm^2.\]
Combining these pieces shows that
\begin{equation}\label{eq:approx_before_jl} \frac 3 4 \inprod{\tmy}{\ml_e} \le \frac 1 Z \norm{\mm \mmp \Par{\ve_u - \ve_v}}_2^2 \le \frac 5 4 \inprod{\tmy}{\ml_e}, \text{ for all } e = (u, v) \in E.\end{equation}
Next, by the Johnson-Lindenstrauss lemma, letting $k = O(\log(\frac n \delta))$ for an appropriately large constant, and letting $\mg \in \R^{k \times n}$ have i.i.d.\ entries $\sim \Nor(0, \frac 1 k)$, we have
\begin{equation}\label{eq:approx_after_jl}\frac 2 3 \inprod{\tmy}{\ml_e} \le \frac 1 Z \norm{\mg \mm \mmp\Par{\ve_u - \ve_v}}_2^2 \le \frac 4 3 \inprod{\tmy}{\ml_e}, \text{ for all } e = (u, v) \in E,\end{equation}
with probability $\ge 1 - \frac \delta 2$. This implies that to get $\half \vg \le \vf \le \vg$, it is enough to set
\[\mq = \sqrt{\frac 3 {4Z}} \mg \mm \mmp,\]
which takes the stated time to compute. We can see this because outputting $\mq$ requires applying $\mm$ and $\mmp$ to $k$ different vectors, and $\tmv(\mm) = O(R \cdot \tmv(\ms))$.
\end{proof}

\paragraph{Preconditioned matrix dictionary recovery.} We next give a coarse approximation to a Laplacian using a method inspired by \cite{JambulapatiLMSST23}. Let $\ml \in \PSD^{n \times n}$ be an unknown graph Laplacian, and assume that we have matrix-vector query access to $\mmp \in \PSD^{n \times n}$ satisfying
\begin{equation}\label{eq:precon_def} \ml^\dagger \preceq \mmp^2 \preceq 2 \ml^\dagger .\end{equation}
We let $\mproj \defeq \id_n - \frac 1 n \1_n\1_n^\top$ be the graph Laplacian corresponding to the clique on $[n]$, as in Theorem~\ref{thm:implicit_sparsify}.
We also require the following notion of a \emph{SOC packing oracle}.

\begin{definition}\label{def:soc_packing}
Let $\ml \in \PSD^{n \times n}$ be an unknown graph Laplacian, and let $Q \ge 1$, $K \in \N$, $\delta \in (0, 1)$. We say $\oracle: \R^E_{\ge 0} \times \PSD^{n \times n} \to \R^E_{\ge 0}$ is a $(Q, K, \delta)$-SOC packing oracle for $\ml$ if on inputs $\vg \in \R^E_{\ge 0}$ and $\mmp \in \PSD^{n \times n}$ satisfying \eqref{eq:precon_def}, with probability $\ge 1 - \delta$, $\oracle$ returns $\vx \in \R^E_{\ge 0}$ satisfying
\begin{equation}\label{eq:feasible_region}\vx \in \calF \text{ and } \inprod{\vg}{\vx} \ge \frac 1 Q\Par{\max_{\vy \in \calF} \inprod{\vg}{\vy}} , 
\text{ where } \calF \defeq \Brace{\vy \in \R^E_{\ge 0} \mid \mmp \calL(\vy) \mmp \preceq \mproj}.\end{equation}
Moreover, the only access used by $\oracle$ to its input $\mmp$ is matrix-vector query access, and its output $\vx$ has the form $\vx = \vv \circ \va$ where $\vv$ is a $K$-SOC and $\va$ is an ASOC.
\end{definition}

Indeed, note that Proposition~\ref{prop:fast_packing} exactly provides an oracle of this form, when $\vg$ is an ASOC. We next show that packing problems with a well-conditioned solution are not substantially affected by small entries, and that we can approximate certain distances encountered in our algorithm. These helper observations are for using our ASOC approximation scheme (Lemma~\ref{lem:kd_asoc}).

\begin{lemma}\label{lem:truncate_small}
Let $\vw \in \R^E_{> 0}$ have $\max_{e \in E} \vw_e \le \rho \min_{e \in E} \vw_e$, let $\mq \in \R^{k \times n}$, $\vf \in \R^E$ be such that
\[\vf_e = \norm{\mq_{:u} - \mq_{:v}}_2^2 \text{ for all } e = (u, v) \in E.\]
Also, suppose $\gamma, \alpha > 0$ satisfy
\[\gamma \ge \max_{(u, v) \in E} \max_{j \in [k]} \Par{\mq_{ju} - \mq_{jv}}^2,\; \frac \gamma \alpha \le \frac 1 {2\rho n^2 k} \max_{(u, v) \in E} \max_{j \in [k]} \Par{\mq_{ju} - \mq_{jv}}^2.\]
Then, following notation from Definition~\ref{def:asoc_cost}, if $\half \vg \le \vf \le \vg$ entrywise,
\[\inprod{\vf^{(\ge \frac \gamma \alpha)}}{\vw} \ge \frac 1 4 \inprod{\vg}{\vw}. \]
\end{lemma}
\begin{proof}
It is enough to show that $\inprod{\vf^{(\ge \frac \gamma \alpha)}}{\vw} \ge \half \inprod{\vf}{\vw}$, because we know that $\vf \ge \half \vg$ entrywise. Throughout the proof let $(a, b) \in E$ and $\ell \in [k]$ be the maximizing arguments in 
\[\max_{(u, v) \in E} \max_{j \in [k]} \Par{\mq_{ju} - \mq_{jv}}^2.\]
Because
\[\inprod{\vf}{\vw} = \sum_{(u, v) \in E} \sum_{j \in [k]} \vw_{(u, v)} \Par{\mq_{ju} - \mq_{jv}}^2 \ge \vw_{(a, b)} \Par{\mq_{\ell a} - \mq_{\ell b}}^2,\]
for any $(u, v) \in E$ and $j \in [k]$ that is a pair such that the corresponding squared distance in $\mq$ does not participate in $\vf^{(\frac \gamma \alpha)}$, we can bound the contribution of the pair to the objective value, i.e.,
\begin{align*} \vw_{(u, v)} \Par{\mq_{ju} - \mq_{jv}}^2 &\le \rho \vw_{(a, b)} \cdot \frac \gamma \alpha \le   \frac 1 {2 n^2 k} \cdot \vw_{(a, b)}  \Par{\mq_{\ell a} - \mq_{\ell b}}^2 .\end{align*}
By combining the above two displays, for any such non-contributing $(u, v) \in E$, $j \in [k]$,
\[\vw_{(u, v)} \Par{\mq_{ju} - \mq_{jv}}^2 \le \frac 1 {2n^2 k} \inprod{\vf}{\vw}.\]
Summing over the contributions of all $\le kn^2$ possible pairs of $(u, v) \in E$ and $j \in [k]$, we have shown that the dropped coordinates contribute at most half the value of $\inprod{\vf}{\vw}$ as claimed.
\end{proof}

\begin{lemma}\label{lem:trace_overestimate}
Let $\vx \in \R^E$ have the form \eqref{eq:mmw_step_form} where all $\{\va_s\}_{s \in [S]}$ are ASOCs and all $\{\vv_s\}_{s \in [S]}$ are $K$-SOCs, and define $\ms, \my$ as in \eqref{eq:mmw_step_form}, \eqref{eq:my_def}. Assume $\normsop{\ms} \le R$, and let $\delta \in (0, 1)$. In time
\[O\Par{\Par{\tmv(\mmp) + nKS\log(n)}R\log\Par{\frac n \delta}},\]
we can output $Z \in \R$ satisfying $\frac 9 {10} \Tr(\mmp \my \mmp) \le Z \le \Tr(\mmp \my \mmp)$ with probability $\ge 1 - \delta$.
\end{lemma}
\begin{proof}
This is standard in the literature, e.g., it is implicit in Lemma 2 of \cite{JambulapatiLMSST23}, which shows that to obtain a constant factor multiplicative estimate of the trace of $\mmp \my \mmp \in \PSD^{n \times n}$ with probability $\ge 1 - \delta$, it is enough to compute $O(\log(\frac n \delta))$ many matrix-vector multiplications through $\mmp \my \mmp$. Here we use the polynomial approximation from Lemma~\ref{lem:jl_embed} and the trace approximation from Lemma~\ref{lem:approx_trace} to simulate constant factor multiplicative approximate access to $\my$.
\end{proof}

Finally, we recall a standard regret bound on the matrix multiplicative weights algorithm.

\begin{proposition}[Theorem 3.1, \cite{ZhuLO15}]\label{prop:mmw}
Consider a sequence of gain matrices $\{\mg_t\}_{0 \le t < T} \subset \PSD^{n \times n}$, which all satisfy for step size $\eta > 0$, that $\normop{\eta \mg_t} \le 1$. Letting $\ms_0 \defeq \0_{n \times n}$ and iteratively defining
\[\my_t \defeq \frac{\exp(-\ms_t)}{\Tr\exp(-\ms_t)},\; \ms_{t + 1} \defeq \ms_t + \eta \mg_t \text{ for all } 0 \le t < T,\]
we have the bound for any $\mmu \in \PSD^{n \times n}$ with $\Tr(\mmu) = 1$,
\[\frac 1 T \sum_{0 \le t < T} \inprod{\mg_t}{\my_t - \mmu} \le \frac{\log(n)}{\eta T} + \frac 1 T \sum_{t \in [T]} \eta \normop{\mg_t} \inprod{\mg_t}{\my_t}. \]
Moreover, if all $\mg_t$ have spans contained in a subspace $S \subseteq \R^n$, then the above bound holds for all $\mmu \in \PSD^{n \times n}$ with $\Tr(\mmu)$ whose spans are contained in $S$.
\end{proposition}

We can now state our matrix dictionary recovery method and its guarantees.

\begin{algorithm2e}
	\caption{$\OMDR(\mmp, \rho, Q, K, \delta,  m, \oracle)$}	\label{alg:oracle_mdr}
	\DontPrintSemicolon
		\codeInput Matrix-vector query access to $\mmp \in \PSD^{n \times n}$ satisfying \eqref{eq:precon_def} for unknown $\ml = \calL(\vw)\in \PSD^{n \times n}$, where $\max_{e \in E} \vw_e \le \rho \min_{e \in E} \vw_e$,  $(Q, K, m, \delta) \in \R_{\ge 1} \times \N \times \N \times (0, 1)$, $\oracle$ a $(Q, K, \frac{\delta}{900 m^2 Q\log(n)})$-SOC packing oracle for $\ml$ \;
        \codeOutput $\bvx \in \R^E_{\ge 0}$ such that with probability $\ge 1 - \delta$, 
        \[\frac 1 {512 m Q} \ml \preceq \calL(\bvx) \preceq \ml\]
  \;
  $(\eta, \ms_0, T, \alpha) \gets (\frac 1 {2}, \0_{n \times n}, \lceil 256mQ\log(n) \rceil, O(\rho n^4 \log^2(\frac{nT}{\delta})))$  \;
  \For{$0 \le t < T$}{
   $\my_t \gets \frac{\exp(-\ms_t)}{\Tr\exp(-\ms_t)}$\;
   $\gamma \gets $ value satisfying $\Tr(\mmp \my_t \mmp) \le \gamma \le \frac {10} 9 \Tr(\mmp \my_t \mmp)$ with probability $\ge 1 - \frac{\delta}{900 m Q \log(n)}$\;\label{line:gamma_approx}
   $\{\tvg_t^{(i)}\}_{i \in [m]} \gets (8, m)$-ASOC approximation to $\vg_t$ with probability $\ge 1 - \frac{\delta}{900 m Q\log(n)}$, where
   \begin{equation}\label{eq:grad_approx} [\vg_t]_e \defeq \inprod{\mmp \my_t \mmp}{\ml_e} \text{ for all } e \in E
   \end{equation}\;\label{line:asoc_approx}
   \For{$i \in [m]$}{
   $\vx_t^{(i)} \gets \oracle(\tvg_t^{(i)}, \mmp)$\;
    }
    $\vx_t \gets \frac 1 m \sum_{i \in [m]} \vx_t^{(i)}$\;
   $\mg_t \gets \mmp \calL(\vx_t) \mmp$\;
   $\ms_{t + 1} \gets \ms_t + \eta \mg_t$\;
  }
  \codeReturn $\bvx \defeq \frac 1 T \sum_{0 \le t < T} \vx_t$
\end{algorithm2e}

\begin{proposition}\label{prop:mdr_analysis}
Following notation in Algorithm~\ref{alg:oracle_mdr}, suppose that we use Proposition~\ref{prop:fast_packing} as our SOC packing oracle, that we use Lemma~\ref{lem:trace_overestimate} to implement Line~\ref{line:gamma_approx}, and that we use Lemmas~\ref{lem:kd_asoc} and~\ref{lem:jl_embed} to implement Line~\ref{line:asoc_approx}. Then, 
Algorithm~\ref{alg:oracle_mdr} returns $\bvx \in \R^E_{\ge 0}$ such that with probability $\ge 1 - \delta$,
\begin{equation}\label{eq:quality}\begin{gathered} \exp\Par{-O\Par{\log^{\frac 2 3}\Par{n\rho}\log\log\Par{\frac{n\rho}{\delta}}}}\ml \preceq \calL(\bvx) \preceq \ml.
\end{gathered}
\end{equation}
Moreover, $\bvx = \sum_{s \in [S]} \vv_s \circ \va_s$ where all $\{\va_s\}_{s \in [S]}$ are ASOCs and all $\{\vv_s\}_{s \in [S]}$ are $K$-SOCs, where
\[K = S = \exp\Par{O\Par{\log^{\frac 2 3}(n\rho)\log\log\Par{\frac{n\rho}{\delta}}}},\]
and the algorithm can be implemented to run in time
\[\Par{\tmv\Par{\mmp} + n} \cdot \exp\Par{O\Par{\log^{\frac 2 3}\Par{n\rho}\log\log\Par{\frac{n\rho}{\delta}}}}.\]
\end{proposition}
\begin{proof}
Throughout this proof let $\beta \defeq 8$, and let
\[k = O\Par{\log\Par{\frac{nT}{\delta}}},\; m = O\Par{k^2\log\Par{\frac 1 \alpha}},\; Q = \exp\Par{O\Par{\log^{\frac 2 3}\Par{\frac{nu}{\ell}}\log\log\Par{\frac{n T u}{\delta \ell}}}},\]
be the parameters from Lemmas~\ref{lem:kd_asoc} and~\ref{lem:jl_embed} and Proposition~\ref{prop:fast_packing}, where the last expression is for bounds $\ell, u$ on the value of our packing objectives that we will specify. As all above expressions depend logarithmically on $T$ (which depends linearly on $mQ$), there are no conflicts for an appropriate choice of constants. Further, since $T \le 300 m Q\log(n)$, by a union bound, we may assume that all computations on Lines~\ref{line:gamma_approx},~\ref{line:asoc_approx} are correct, and that all calls to $\oracle$ succeed, except with probability $\delta$.

Condition on these events for the rest of the proof. We next specify valid bounds $\ell, u$ that hold with probability $1$ for all packing problems encountered in the algorithm. Specifically we let
\begin{equation}\label{eq:packing_bounds}
\ell \defeq \frac{1}{2} \min_{e \in E} \vw_e,\;  u \defeq \frac {n^4} 4 \max_{e \in E} \vw_e \implies \frac u \ell \le n^4\rho.
\end{equation}
To see why the bounds in \eqref{eq:packing_bounds} are valid, first, observe that
\begin{equation}\label{eq:lap_bound}
\calL\Par{\vw} \preceq 2\sum_{e \in E} \vw_e \mproj \preceq \Par{n^2 \max_{e \in E} \vw_e} \mproj,
\end{equation}
where we used that $2\mproj$ dominates any $\ml_e$. Thus for any packing problem of the form \eqref{eq:packing} encountered by the algorithm, by using the assumption \eqref{eq:precon_def}, the optimal $\vx \in \R^E_{\ge 0}$ cannot have
\[\vx_e > \frac {n^2 \max_{e \in E} \vw_e} {2} \ge \frac{1}{2\vlam_n\Par{\mmp}^2} \text{ for any } e \in E,\]
because then this coordinate alone would violate the feasibility constraint:
\[\vx_e \ma_e = \vx_e \mmp \ml_e \mmp \not\preceq \mproj. \]
This implies $u$ in \eqref{eq:packing_bounds} is a valid upper bound on $\inprod{\vc}{\vx}$ for feasible $\vx$ and $\vc \in \{0, 1\}^E$, because $\norm{\vc}_1 \le \frac{n^2}{2}$.
Next, to obtain the lower bound,
by pre-multiplying and post-multiplying \eqref{eq:precon_def} by $\ml^{\half}$, all eigenvalues of $\mmp \ml \mmp$ are in $[1, 2]$ (other than a zero eigenvalue in the $\1_n$ direction), i.e.,
\[\mproj \preceq \mmp \ml \mmp \preceq 2\mproj. \]
Thus, $\half \vw \in \calF$ for the feasible region $\calF$ defined in \eqref{eq:feasible_region}, because
\[\mmp \calL\Par{\half \vw} \mmp = \half \mmp \ml \mmp \preceq \mproj.\]
Thus the optimal $\inprod{\vc}{\vx}$ is always at least $\half \min_{e \in E} \vw_e$, by plugging in the feasible choice of $\vx = \half \vw$. This concludes the proof that $\ell$ in \eqref{eq:packing} is also valid. We simplify our parameters accordingly:
\begin{equation}\label{eq:simpler_parameters}
\begin{aligned}
Q &= \exp\Par{O\Par{\log^{\frac 2 3}\Par{n\rho} \log\log\Par{\frac{n\rho}{\delta}}}},\; T = \exp\Par{O\Par{\log^{\frac 2 3}\Par{n\rho} \log\log\Par{\frac{n\rho}{\delta}}}}, \\
\alpha &= O\Par{n^4\rho\log^2\Par{\frac{n\rho}{\delta}}},\; k = O\Par{\log\Par{\frac{n\rho}{\delta}}},\; m = O\Par{\log^3\Par{\frac{n\rho}{\delta}}}.
\end{aligned}
\end{equation}
Continuing, we have for any $\my$ such that $\Tr(\my) = 1$ and $\Span(\my) = \Span(\mproj)$:
\begin{equation}\label{eq:tracebound}
\inprod{\mmp \my \mmp}{\calL(\vw)} = \inprod{\my}{\mmp \ml \mmp} \succeq \inprod{\my}{\mproj} = 1.
\end{equation}
By the guarantee on $\oracle$ (Definition~\ref{def:soc_packing}), we have that for all $i \in [m]$ and $0 \le t < T$,
\[\inprod{\tvg_t^{(i)}}{\vx^{(i)}} \ge \frac 1 {2Q} \inprod{\tvg_t^{(i)}}{\vw} \implies \inprod{\vg_t}{\vx^{(i)}_t} \ge \frac 1 {2\beta Q} \inprod{\tvg_t^{(i)}}{\vw}.\]
The above implication used that $\vg_t \ge \frac 1 \beta \tvg_t^{(i)}$ for all $i \in [m]$ by the ASOC approximation guarantee. Averaging this bound for all $i \in [m]$,
\begin{equation}\label{eq:quality_per_round}
\begin{aligned}
\inprod{\vg_t}{\vx_t} &= 
\frac 1 m \sum_{i \in [m]} \inprod{\vg_t}{\vx^{(i)}_t} \ge \frac 1 {2\beta m Q} \sum_{i \in [m]}\inprod{\tvg_t^{(i)}}{\vw} \ge \frac 1 {8\beta m Q} \inprod{\vg_t}{\vw}.
\end{aligned}
\end{equation}
The last inequality in \eqref{eq:quality_per_round} is derived as follows. Recall that $\sum_{i \in [m]} \tvg_t^{(i)} \ge \vf_t^{(\ge \frac \gamma \alpha)}$ by combining our ASOC approximation guarantee with $\vf_t \ge \half \vg_t$ for our estimate $\vf_t$ constructed in Lemma~\ref{lem:jl_embed}. Note that $\gamma$ is at most a $\frac {20} 9 n^2 k$ factor overestimate of the largest entry in Lemma~\ref{lem:truncate_small}, because
\[\sum_{e \in E} \vg_e = \inprod{\mmp \my \mmp}{\sum_{e \in E} \ml_e} = \inprod{\mmp \my \mmp}{n\mproj} = n\Tr\Par{\mmp \my \mmp}, \]
so our choice of $\alpha = \frac{40}{9} \rho n^4 k^2$ shows that Lemma~\ref{lem:truncate_small} applies. Thus,
\[\sum_{i \in [m]} \inprod{\tvg_t^{(i)}}{\vw} \ge \inprod{\vf_t^{(\frac \gamma \alpha)}}{\vw} \ge \frac 1 4 \inprod{\vg_t}{\vw}.\]

However, we also have by the definition of $\vg_t$ in \eqref{eq:grad_approx} that
\[\inprod{\vg_t}{\vw} = \inprod{\mmp \my_t \mmp}{\calL(\vw)} \ge 1,\]
where we used \eqref{eq:tracebound} in the last inequality. Combining with \eqref{eq:quality_per_round}, we have shown
\begin{equation}\label{eq:gain_per_round}
\inprod{\mg_t}{\my_t} = \inprod{\mmp \calL(\vx_t) \mmp}{\my_t} = \inprod{\vg_t}{\vx_t} \ge \frac 1 {8\beta m Q},
\end{equation}
for all $0 \le t < T$. Also, by the packing oracle guarantee $\vx_t \in \calF$,
\begin{equation}\label{eq:gain_bound}\mg_t = \mmp \calL(\vx_t) \mmp \preceq \mproj.\end{equation}
Now, rearranging Proposition~\ref{prop:mmw} gives for all $\mmu \in \PSD^{n \times n}$ with $\Span(\mmu) = \Span(\mproj)$ and $\Tr(\mmu) = 1$,
\begin{equation}\label{eq:covering_bound}
\begin{aligned}
\inprod{\mmp \calL(\bvx) \mmp}{\mmu} &= \frac 1 T \sum_{0 \le t < T}\inprod{\mg_t}{\mmu} \\
&\ge \frac 1 {2T} \sum_{0 \le t < T} \inprod{\mg_t}{\my_t} - \frac{\log(n)}{\eta T} \\
&= \frac 1 {16\beta m Q} - \frac 1 {32\beta m Q} = \frac 1 {32 \beta m Q}.
\end{aligned}
\end{equation}
The first line above used the definition of $\bvx$, the second used Proposition~\ref{prop:mmw} with $\eta = \half$ and $\normsop{\mg_t} \le 1$ by \eqref{eq:gain_bound}, and the third used \eqref{eq:gain_per_round} and plugged in our choices of $\eta$ and $T$. Minimizing \eqref{eq:covering_bound} over all valid choices of $\mmu$ shows that
\[\frac 1 {32\beta m Q} \mproj \preceq \mmp \calL(\bvx) \mmp \preceq \mproj, \]
where the upper bound holds by averaging \eqref{eq:gain_bound} across all iterations. Thus, we indeed have
\[ \frac 1 {64\beta m Q} \ml \preceq \calL(\bvx) \preceq \ml,\]
which proves \eqref{eq:quality} upon plugging in $\beta = 8$ and our parameters from \eqref{eq:simpler_parameters}.

Next we discuss the maintenance of our iterates. Because of the SOC packing oracle guarantee, every $\vx_t$ for $0 \le t < T$ is the average of $m$ products of a $K$-SOC and an ASOC, where $m$ is as in \eqref{eq:simpler_parameters} and $K$ is as stated in Proposition~\ref{prop:fast_packing}. Thus, the average of all $T$ iterates has $S = mT$ components, each a product of a $K$-SOC and an ASOC, in its representation as claimed. 

There are three runtime bottlenecks in our algorithm: for each of $T$ iterations, we call Lemma~\ref{lem:trace_overestimate} to implement Line~\ref{line:gamma_approx}, we call Lemmas~\ref{lem:jl_embed} and~\ref{lem:kd_asoc} to implement Line~\ref{line:asoc_approx}, and we call Proposition~\ref{prop:fast_packing} $m$ times to implement our SOC packing oracle. Moreover, $\normsop{\ms_t} \le \frac T 2$ for all iterations $0 \le t < T$ by inspection. The result follows by combining all of these runtimes with our parameter choices.
\end{proof}

\subsection{Homotopy method}\label{ssec:homotopy}

In this section, we show how recursively calling Proposition~\ref{prop:mdr_analysis} in phases grants us access to $\mmp$ in \eqref{eq:precon_def} needed by the next phase.
We use an implicit approximation provided in \cite{JambulapatiLMSST23}.

\begin{lemma}\label{lem:invsqrt}
Suppose that for $\bvx \in \R^E$, $\Delta, Q > 0$, and unknown graph Laplacian $\ml \in \PSD^{n \times n}$, we have 
\[\frac 1 Q \Par{\ml + \Delta \mproj}\preceq \calL\Par{\bvx} \preceq \ml + \Delta \mproj,\]
where $\mproj \defeq \id_n - \frac 1 n \1_n \1_n^\top$.
Further, suppose that $\bvx = \sum_{s \in [S]} \va_s \circ \vv_s$ where all $\{\va_s\}_{s \in [S]}$ are ASOCs and all $\{\vv_s\}_{s \in [S]}$ are $K$-SOCs. Then we can provide matrix-vector product access to a matrix $\mmp$ that satisfies the following with probability $\ge 1 - \delta$:
\begin{gather*}\Par{\ml + \Delta \mproj}^\dagger \preceq \mmp^2 \preceq 2\Par{\ml + \Delta\mproj}^\dagger,\\
\tmv\Par{\mmp} = O\Par{\Par{\tmv\Par{\ml} +  nKS\log^2(n)\log^2\Par{\frac{nKS\Tr(\ml)}{\Delta\delta}}} \sqrt Q \log^6\Par{\frac{Q\Tr(\ml)}{\Delta\delta}}}.
\end{gather*}
\end{lemma}
\begin{proof}
We first sparsify $\calL(\bvx)$ using Lemma~\ref{lem:sparsify_soc_asoc} to obtain a graph Laplacian $\tml$ that satisfies
\[\frac 1 {Q\exp(2)} \Par{\ml + \Delta \mproj} \preceq \tml \preceq \ml + \Delta \mproj,\; \nnz\Par{\tml} = O\Par{nKS\log(n)\log\Par{\frac {nKS} \delta}}, \] except with probability $\frac \delta 2$. Next, we can provide access to constant-factor approximations of $(\tml + \lam \mproj)^\dagger$ for any $L$ values of $\lam \ge 0$ with failure probability $\frac \delta 2$, in time $O(\nnz(\tml)\log(n) \log(\frac {nL} \delta))$ by using the Laplacian system solver of \cite{KoutisMP14}. The conclusion follows from Lemma 13, \cite{JambulapatiLMSST23}, which requires this primitive with $L = O(\log(\frac{\Tr(\ml)}{\Delta}))$ (see Proposition 3 in \cite{JambulapatiLMSST23}). We remark that the approximation factor worsens by a $\frac{\Tr(\ml)}{\Delta}$ factor due to Lemma 14 in \cite{JambulapatiLMSST23}, but this is accounted for by the polylogarithmic terms above, where the accuracy dependence lies.
\end{proof}

We are now ready to prove Theorem~\ref{thm:implicit_sparsify}.

\restateimplicitsparsify*
\begin{proof}
We closely follow the outline of the homotopy method outlined in Section 3.2 of \cite{JambulapatiLMSST23}. The method proceeds in $p = O(\log \frac {\Tr(\ml)} \Delta)$ phases, and in each phase $q \in [p]$, we apply Proposition~\ref{prop:mdr_analysis} to the unknown regularized Laplacian
\[\ml_q \defeq \ml + \Delta 2^{p - q} \mproj.\]
Note that in each phase, $\ml_q = \calL(\vw_q)$ for a vector $\vw_q \in \R^E_{\ge 0}$ that is entrywise at least $2^{p - q} \cdot \frac \Delta n$, and whose sum of entries is at most $2^{p - q} \cdot \Delta n + \Tr(\ml)$. Thus, in all calls to Proposition~\ref{prop:mdr_analysis}, we have
\[\rho = O\Par{n^2 + \frac{n\Tr(\ml)}{\Delta}}.\]
To initialize each phase, we require a matrix $\mmp_q$ satisfying \eqref{eq:precon_def}. In the first phase $q = 1$, for a large enough constant in the definition of $p$, it is enough to choose $\mmp_1$ to be a known multiple of $\mproj$. In every phase $q$ after this, $\mmp_q$ results from applying Lemma~\ref{lem:invsqrt}. Because
\[\ml_q \preceq \ml_{q - 1} \preceq 2\ml_q\]
for all $q \in [p]$, the multiplicative approximation factor given by Proposition~\ref{prop:mdr_analysis} only worsens by $2$. The runtime follows from Proposition~\ref{prop:mdr_analysis} and our matrix-vector query access in Lemma~\ref{lem:invsqrt}. Finally, the sparsity of the final output follows by applying the sparsification described in the proof of Lemma~\ref{lem:invsqrt} to the output of the final phase $p$, and scaling it up to obtain the approximation guarantee.
\end{proof}
\section{Conditioning of Smoothed Matrices}\label{sec:smoothed}

In this section, we provide an $\ell_\infty$ diameter bound for a minimizing vector of Barthe's objective, when computing a Forster transform of a smoothed matrices of the form 
\[\tma = \ma + \mg, \text{ where } \ma \in \R^{n \times d}, \text{ and } \mg \in \R^{n \times d} \text{ has entries } \simiid \Nor(0, \sig^2). \]
We first define a notion of deepness in Section~\ref{ssec:deep} and derive a diameter bound for deep vectors, patterned off \cite{ArtsteinKS20}. In Sections~\ref{ssec:wide} and \ref{ssec:tall}, we show tail bounds for the singular values of a smoothed matrix. We combine these results to prove Theorem~\ref{thm:smoothed} in Section~\ref{ssec:smoothed_forster}, our main result on the conditioning of Forster transforms of smoothed matrices.

Throughout this section, we follow notation \eqref{eq:notation}, \eqref{eq:barthe_obj}, \eqref{eq:mi_def} from Section~\ref{ssec:rip} and make the following simplifying, and somewhat mild, assumption on the relationship between $n$ and $d$.

\begin{assumption}\label{assume:nd}
    $n\geq Cd$ for some constant $C>1$.
\end{assumption}

For example, our results apply if $n \ge 1.1d$. Lifting the restriction in Assumption~\ref{assume:nd} significantly complicates our approach, as explained by Remark~\ref{rem:superlinear}, and we leave it as an interesting open question to establish similar bounds in the parameter regime $n \in [d, d + o(d)]$.

Finally, we mention that our result in Theorem~\ref{thm:smoothed} is stated for the case of $\vc = \frac d n \1_n$, which is the most interesting setting we are aware of in typical applications. However, we discuss the case of general $\vc$ in Section~\ref{ssec:gen_c}, where our techniques readily apply under a strengthening of Assumption~\ref{assume:nd}.

\subsection{Diameter bound for deep vectors}\label{ssec:deep}

Our strategy for our diameter bound follows an analysis in \cite{ArtsteinKS20}, based on the assumption that $\vc$ lies nontrivially inside the interior of the basis polytope (cf.\ Proposition~\ref{prop:scaling_polytope}). 

We start by extending Definition 1.4 and proving a stronger version of Lemma 4.4 from \cite{ArtsteinKS20}.

\begin{definition}[Deepness]\label{def:deep}
    Let $\vc\in\R_{>0}^n$ satisfy $\norm{\vc}_1=d$, $\eta\in[0,1]$, and $\Delta\geq0$. We say that $\vc$ lies \emph{$(\eta,\Delta)$-deep} inside the basis polytope of $\{\va_i\}_{i\in[n]}\subset\R^d$ if for all $k\in[d-1]$ and subspaces $E \subseteq \R^d$ with dimension $k$, \[ \sum_{i\in[n]}\vc_i\ind_{\norm{\va_i-\proj_E\va_i}_2\leq\Delta}\leq (1-\eta)k. \] When $\vc=\frac{d}{n}\1_n$, this is equivalent to the following: for all $k\in[d-1]$ and subspaces $E$ with dimension $k$, at most $\frac{(1-\eta)kn}{d}$ of the $\va_i$ satisfy $\norm{\va_i-\proj_E\va_i}_2\leq\Delta$.
\end{definition}

We remark that every vector in the basis polytope is $(0,0)$-deep by Proposition~\ref{prop:scaling_polytope}. Furthermore, increasing $\Delta$ potentially increases the number of $\vc_i$ considered in the sum, and increasing $\eta$ tightens the inequality, so $(\eta,\Delta)$-deepness becomes a stronger condition for larger values of $\eta$ and $\Delta$.

We now show, following \cite{ArtsteinKS20}, that deepness in the basis polytope implies a diameter bound.

\begin{lemma}\label{lem:diam_bound}
    Let $\mu ,M,\eta,\Delta>0$ and $\ma\in\R^{n\times d}$ with rows $\{\va_i\}_{i\in[n]}$ satisfying $\mu \leq\norm{\va_i}_2^2\leq M$ for all $i\in[n]$. If $\vc\in\R_{>0}^n$ has minimum entry $\vc_{\min}$ and lies $(\eta,\Delta)$-deep inside the basis polytope of $\{\va_i\}_{i\in[n]}$, then there exists $\vt^\star\in\argmin_{\vt\in\R^n}f(\vt)$ satisfying \[ \norm{\vt^\star}_\infty\leq\half\log\Par{\frac{M}{\mu \vc_{\min}}\Par{\frac{4M}{\eta\Delta^2}}^{d-1}}. \]
\end{lemma}

\begin{proof}
    Let $\vt^\star\in\argmin_{\vt\in\R^n}f(\vt)$ such that $\min_{i\in[n]}\vt^\star_i=0$, which exists by Proposition~\ref{prop:scaling_polytope}, Proposition~\ref{prop:barthe}, and \eqref{eq:invariant_ones}. Let $\vlam_1\geq\vlam_2\geq\dots\geq\vlam_d$ be the eigenvalues of $\mr\defeq\mr(\vt^\star)$.

    Fix $k\in[d-1]$. Let $E$ be a $k$-dimensional subspace spanned by eigenvectors of $\mr$ with eigenvalues $\vlam_d,\vlam_{d-1},\dots,\vlam_{d-k+1}$. By Proposition~\ref{prop:barthe}, $\mr$ is a $\vc$-Forster transform of $\ma$, so \[ \sum_{i\in[n]}\vc_i\cdot\frac{\Par{\mr\va_i}\Par{\mr\va_i}^\top}{\norm{\mr\va_i}_2^2}=\id_d. \] Projecting onto $E^\perp$ and taking a trace of both sides, \begin{equation}\label{eq:traces}
        \sum_{i\in[n]}\vc_i\cdot\frac{\norm{\proj_{E^\perp}\mr\va_i}_2^2}{\norm{\mr\va_i}_2^2}=\Tr\Par{\sum_{i\in[n]}\vc_i\cdot\frac{\Par{\proj_{E^\perp}\mr\va_i}\Par{\proj_{E^\perp}\mr\va_i}^\top}{\norm{\mr\va_i}_2^2}}=\Tr\Par{\proj_{E^\perp}}=d-k.
    \end{equation} Now, consider the $\va_i$ such that $\norm{\va_i-\proj_E\va_i}_2>\Delta$, and decompose these $\va_i$ as $\va_i=\vy_i+\vz_i$, where $\vy_i\in E$ and $\vz_i\in E^\perp$. Then $\norm{\vz_i}_2>\Delta$ and $\norm{\vy_i}_2<\sqrt{\norm{\va_i}_2^2-\Delta^2}\leq\sqrt{M-\Delta^2}$. Since $E$ and $E^\perp$ are both spanned by eigenvectors of $\mr$, $\mr\vy_i$ and $\mr\vz_i$ are orthogonal, and $\norm{\mr\va_i}_2^2=\norm{\mr\vy_i}_2^2+\norm{\mr\vz_i}_2^2$. Furthermore, since $E$ is spanned by eigenvectors with eigenvalues at most $\vlam_{d-k+1}$ and $E^\perp$ is spanned by eigenvectors with eigenvalues at least $\vlam_{d-k}$, \[ \norm{\mr\vy_i}_2\leq\vlam_{d-k+1}\norm{\vy_i}_2\leq\vlam_{d-k+1}\sqrt{M-\Delta^2}\text{ and }\norm{\mr\va_i}_2\geq\norm{\mr\vz_i}_2\geq\vlam_{d-k}\norm{\vz_i}_2\geq\vlam_{d-k}\Delta. \] It follows that \[ \frac{\norm{\proj_{E^\perp}\mr\va_i}_2^2}{\norm{\mr\va_i}_2^2}=\norm{\proj_{E^\perp}\frac{\mr\va_i}{\norm{\mr\va_i}_2}}_2^2=1-\norm{\proj_E\frac{\mr\va_i}{\norm{\mr\va_i}_2}}_2^2=1-\frac{\norm{\mr\vy_i}_2^2}{\norm{\mr\va_i}_2^2}\geq1-\frac{\vlam_{d-k+1}^2(M-\Delta^2)}{\vlam_{d-k}^2\Delta^2}. \] Combining this with \eqref{eq:traces} and the $(\eta,\Delta)$-deepness of $\vc$, \begin{align*}
        d-k&=\sum_{i\in[n]}\vc_i\cdot\frac{\norm{\proj_{E^\perp}\mr\va_i}_2^2}{\norm{\mr\va_i}_2^2}\geq\sum_{i\in[n]}\vc_i\ind_{\norm{\va_i-\proj_E\va_i}_2>\Delta}\cdot\frac{\norm{\proj_{E^\perp}\mr\va_i}_2^2}{\norm{\mr\va_i}_2^2}\\
        &\geq\Par{1-\frac{\vlam_{d-k+1}^2(M-\Delta^2)}{\vlam_{d-k}^2\Delta^2}}\sum_{i\in[n]}\vc_i\ind_{\norm{\va_i-\proj_E\va_i}_2>\Delta}\\
        &\geq\Par{1-\frac{\vlam_{d-k+1}^2(M-\Delta^2)}{\vlam_{d-k}^2\Delta^2}}(d-(1-\eta)k),
    \end{align*} which rearranges to \begin{equation}\label{eq:lam_bound}
        \frac{\vlam_{d-k}}{\vlam_{d-k+1}}\leq\Par{\frac{d-k+\eta k}{\eta k}\cdot\frac{M-\Delta^2}{\Delta^2}}^\half.
    \end{equation} Since \eqref{eq:lam_bound} holds for all $k\in[d-1]$, 
    \begin{equation}\label{eq:condition_number}
        \begin{aligned}
        \frac{\vlam_1}{\vlam_d}&=\prod_{k\in[d-1]}\frac{\vlam_{d-k}}{\vlam_{d-k+1}}\leq\prod_{k\in[d-1]}\Par{\frac{d-k+\eta k}{\eta k}\cdot\frac{M-\Delta^2}{\Delta^2}}^\half\\
        &=\Par{\frac{M-\Delta^2}{\Delta^2}}^{\frac{d-1}{2}}\prod_{k\in[d-1]}\Par{\frac{d-k+\eta k}{\eta k}}^\half\\
        &=\Par{\frac{M-\Delta^2}{\Delta^2}}^{\frac{d-1}{2}}\prod_{k=1}^{\left\lfloor\frac{d}{1+\eta}\right\rfloor}\Par{\frac{d-k+\eta k}{\eta k}}^\half\prod_{k=\left\lfloor\frac{d}{1+\eta}\right\rfloor+1}^{d-1}\Par{\frac{d-k+\eta k}{\eta k}}^\half\\
        &\leq\Par{\frac{M-\Delta^2}{\Delta^2}}^{\frac{d-1}{2}}\prod_{k=1}^{\left\lfloor\frac{d}{1+\eta}\right\rfloor}\Par{\frac{2(d-k)}{\eta k}}^\half\prod_{k=\left\lfloor\frac{d}{1+\eta}\right\rfloor+1}^{d-1}\Par{\frac{2}{\eta}}^\half\\
        &=\Par{\frac{2(M-\Delta^2)}{\eta\Delta^2}}^{\frac{d-1}{2}}\prod_{k=1}^{\left\lfloor\frac{d}{1+\eta}\right\rfloor}\Par{\frac{d-k}{k}}^\half=\Par{\frac{2(M-\Delta)^2}{\eta\Delta^2}}^{\frac{d-1}{2}}\binom{d-1}{\left\lfloor\frac{d}{1+\eta}\right\rfloor}^\half\\
        &\leq\Par{\frac{4(M-\Delta^2)}{\eta\Delta^2}}^{\frac{d-1}{2}}\leq\Par{\frac{4M}{\eta\Delta^2}}^{\frac{d-1}{2}},
    \end{aligned}
    \end{equation} where the fourth line uses $\frac{d-k+\eta k}{\eta k}\leq\frac{2(d-k)}{\eta k}$ for $k\leq\frac{d}{1+\eta}$ and $\frac{d-k+\eta k}{\eta k}\leq2\leq\frac{2}{\eta}$ for $k\geq\frac{d}{1+\eta}$ and the sixth line uses $\binom{a}{b}\leq2^a$. 
    
    Let $j\in[n]$ such that $\vt^\star_j=\norm{\vt^\star}_\infty$, and note that the largest eigenvalue of $\mz(\vt^\star)$ is $\frac{1}{\vlam_d^2}$. Thus \begin{equation}\label{eq:lammin_bound}
        \frac{1}{\vlam_d^2}=\max_{\norm{\vx}_2=1}\sum_{i\in[n]}\exp(\vt^\star_i)\inprod{\va_i}{\vx}^2\geq\max_{\norm{\vx}_2=1}\exp(\vt^\star_j)\inprod{\va_j}{\vx}^2=\norm{\va_j}_2^2\exp(\vt^\star_j)\geq \mu \exp(\norm{\vt^\star}_\infty).
    \end{equation} Let $\ell\in[n]$ such that $\vt^\star_\ell=0$. Since $\vt^\star$ minimizes $f$, we have \begin{equation}\label{eq:lammax_bound} \vc_\ell=\Tr(\mm_\ell(\vt^\star))=\norm{\mr\va_\ell}_2^2\leq\vlam_1^2\norm{\va_\ell}_2^2\leq M\vlam_1^2
    \end{equation} by Fact~\ref{fact:barthe_derivs}. Combining \eqref{eq:lammin_bound} and \eqref{eq:lammax_bound} gives \[ \frac{\vlam_1^2}{\vlam_d^2}\geq\frac{\mu \vc_{\min}}{M}\exp(\norm{\vt^\star}_\infty), \] and combining with \eqref{eq:condition_number} and rearranging gives \[ \norm{\vt^\star}_\infty\leq\log\Par{\frac{M}{\mu \vc_{\min}}\Par{\frac{4M}{\eta\Delta^2}}^{d-1}}. \] Since $f$ is invariant to translations by $\1_n$ and $\min_{i\in[n]}\vt^\star_i=0$ by assumption, we can shift $\vt^\star$ to obtain a minimizer that has extreme coordinates averaging to 0, which gives the result.
\end{proof}

With Lemma~\ref{lem:diam_bound} in hand, we must show $\frac M \mu$ and $\frac 1 \Delta$ are polynomially bounded for an appropriate choice of $\eta$, in the smoothed setting. The bulk of our remaining analysis establishes this result. 

\begin{remark}\label{rem:superlinear}
Our strategy in this section is to restrict $\eta$ to be a constant, e.g., $\eta = 0.1$, in which case our goal is to show that at most $\frac{0.9 k}{d} n$ of the vectors in $\{\va_i\}_{i \in [n]}$ lie close to any $k$-dimensional subspace. If, for instance, $\frac n d < 1.1$, this is clearly impossible, because choosing $k = 1$ implies that not even a single vector lies close to any $1$-dimensional subspace, which is false just by taking $\Span(\va_i)$ for any $i \in [n]$. Thus, the restriction $\frac n d \ge 1.1$ (or more generally, a constant bounded away from $1$) is somewhat inherent in the regime $\eta = \Omega(1)$. It is possible that our strategy can be modified to extend to even smaller $n$, but for simplicity, we focus on the setting of Assumption~\ref{assume:nd}.
\end{remark}

To frame the rest of the section, we provide a helper lemma relating the condition in Definition~\ref{def:deep} to appropriate submatrices having small singular values.

\begin{lemma}\label{lem:small_sv}
    Let $\ma=\{\va_i\}_{i\in[m]}\in\R^{d\times m}$, $k<m$, and $\Delta>0$. Suppose there exists a $k$-dimensional subspace $E$ of $\R^d$ such that $\norm{\va_i-\proj_E\va_i}_2\leq\Delta$ for all $i\in[m]$. Then $\vsig_{k+1}(\ma)\leq \sqrt{m}\Delta$.
\end{lemma}

\begin{proof}
    Let $\mv\in\R^{d\times(d-k)}$ have columns that form an orthonormal basis for $E^\perp$. By assumption, \[ \norm{\mv\mv^\top\va_i}_2^2=\norm{\mv^\top\va_i}_2^2\leq\Delta^2\text{ for all }i\in[m]. \] By the min-max theorem, \begin{align*}
        \vsig_{k+1}(\ma)&=\min_{\substack{E\subseteq\R^{d}\\\dim(E)=d-k}}\max_{\substack{\vx\in E\\\norm{\vx}_2=1}}\norm{\ma^\top\vx}_2\\
        &\leq\max_{\substack{\vx\in\Span(\mv)\\\norm{\vx}_2=1}}\norm{\ma^\top\vx}_2=\max_{\substack{\vv\in\R^{d-k}\\\norm{\vv}_2=1}}\norm{\ma^\top\mv\vv}_2.
    \end{align*}
    Observe that $\ma^\top\mv$ has rows $\{\mv^\top\va_i\}_{i\in[m]}$, so by the Cauchy--Schwarz inequality, \[ \norm{\ma^\top\mv\vv}_2^2=\sum_{i\in[m]}\inprod{\mv^\top\va_i}{\vv}^2\leq\sum_{i\in[m]}\Delta^2=m\Delta^2 \] for all $\vv\in\R^{d-k}$ with $\norm{\vv}_2=1$. It follows that $\vsig_{k+1}(\ma)\leq\sqrt{m}\Delta$.
\end{proof}

In Sections~\ref{ssec:wide} and~\ref{ssec:tall}, we show that with high probability, the conclusion of Lemma~\ref{lem:small_sv} is violated for all appropriately-sized smoothed submatrices. This shows that the premise of Lemma~\ref{lem:small_sv} is also violated, which we establish in Section~\ref{ssec:smoothed_forster}, so that $\vc = \frac d n \1_n$ is indeed $(\eta, \Delta)$-deep.

\subsection{Conditioning of wide and near-square smoothed matrices}\label{ssec:wide}

Throughout this section, we let 
\[\eta\defeq1-\frac{1}{\sqrt{C}},\; K\defeq\sqrt[3]{C},\; k\in[d-1],\text{ and } m\defeq\left\lceil\frac{(1-\eta)kn}{d}\right\rceil.\]

\begin{remark}\label{rem:eta}
    We note that these definitions imply $0<\eta<1-\frac{1}{C}$ and $1<K<(1-\eta)C$. Indeed, nothing in our analysis relies on our choices of $\eta$ and $K$ other than the fact that they are constants that satisfy these inequalities, which limit our analysis in the following places.
    \begin{itemize}
        \item In Lemma~\ref{lem:wide}, we require $(1-\eta)C>1$ (equivalently, $\eta<1-\frac{1}{C}$) so that $m-k=\Omega(m)$.
        \item In Lemma~\ref{lem:square}, we require $K<(1-\eta)C$ so that $d-k=\Omega(d)$.
        \item In Lemma~\ref{lem:square} and Lemma~\ref{lem:tall}, we require $K>1$, and in Theorem~\ref{thm:smoothed}, we require $\eta>0$.
    \end{itemize}
\end{remark}

We first use the following results from the literature to establish tail bounds for the singular values of smoothed matrices in the cases $m\leq d$ and $d<m\leq Kd$, i.e., $m$ that are at most a constant factor larger than $d$. In Section~\ref{ssec:tall}, we use a different argument to handle $m > Kd$.

\begin{fact}[Theorem 1.2, \cite{Szarek91}]\label{fact:gaussian_sv}
    Let $\mg\in\R^{d\times d}$ have entries $\simiid\Nor(0,1)$. Then for all $j\in[d]$, \[ \Pr\Brack{\vsig_{d-j+1}(\mg)<\frac{\alpha j}{\sqrt{d}}}\leq\Par{\sqrt{2e}\alpha}^{j^2}. \]
\end{fact}

\begin{fact}[Theorem 2.4, \cite{BanksKMS21}]\label{fact:smoothed_sv}
    Let $\mm,\mn\in\R^{d\times d}$ such that $\vsig_i(\mm)\geq\vsig_i(\mn)$ for all $i\in[d]$. Then for every $t\geq0$, there exists a joint distribution on pairs of matrices $(\mg,\mh)\in\R^{d\times d}\times\R^{d\times d}$ such that the marginals $\mg$ and $\mh$ have entries $\simiid\Nor(0,1)$ and \[ \Pr[\vsig_i(\mm+t\mg)\geq\vsig_i(\mn+t\mh)]=1\text{ for all }i\in[d]. \]
\end{fact}

We note that these results imply tail bounds for the singular values of square smoothed matrices: given $\sig>0$, $\ma\in\R^{d\times d}$, and $\mg\in\R^{d\times d}$ with entries $\simiid\Nor(0,1)$, we let $\mm\gets\ma$, $\mn\gets\0$, $t\gets\sig$, and $(\mg,\mh)$ have the distribution in Fact~\ref{fact:smoothed_sv}. Then by Fact~\ref{fact:smoothed_sv} and Fact~\ref{fact:gaussian_sv}, \begin{equation}\label{eq:smoothed_sv}
    \begin{aligned}
        \Pr\Brack{\vsig_{d-j+1}(\ma+\sig\mg)<\frac{\alpha\sig j}{\sqrt{d}}}&\leq\Pr\Brack{\vsig_{d-j+1}(\sig\mh)<\frac{\alpha\sig j}{\sqrt{d}}}\\
        &=\Pr\Brack{\vsig_{d-j+1}(\mh)<\frac{\alpha j}{\sqrt{d}}}\leq\Par{\sqrt{2e}\alpha}^{j^2}.
    \end{aligned}
\end{equation}

\begin{lemma}\label{lem:wide}
    Under Assumption~\ref{assume:nd}, let $\delta\in(0,1)$, $\sig\in(0,\half)$, $\eta\defeq1-\frac{1}{\sqrt{C}}$, $k\in[d-1]$, $m\defeq\lceil\frac{(1-\eta)kn}{d}\rceil$, $\ma\in\R^{n\times d}$, $\mg\in\R^{n\times d}$ have entries $\simiid\Nor(0,\sig^2)$, $\tma\defeq\ma+\mg$, and suppose $m\leq d$. Then for any $S\subseteq[n]$ with $|S|=m$, \[ \Pr\Brack{\vsig_{k+1}(\tma_{S:})\leq\sqrt{m}\Delta}\leq\frac{\delta}{(d-1)\binom{n}{m}},\text{ where }\Delta = \Par{\frac{\delta\sig}{n}}^{O(1)}. \]
\end{lemma}

\begin{proof}
    Remove any $d-m$ columns of $\tma_{S:}$ to obtain $\mb\in\R^{m\times m}$. We have $\vsig_{k+1}(\tma_{S:})\geq\vsig_{k+1}(\mb)$ by the min-max theorem: \[ \vsig_{k+1}(\tma_{S:})=\max_{\substack{E\subseteq\R^{d}\\\dim(E)=k+1}}\min_{\substack{\vx\in E\\\norm{\vx}_2=1}}\norm{\tma_{S:}\vx}_2\geq\max_{\substack{E\subseteq\R^{m}\\\dim(E)=k+1}}\min_{\substack{\vx\in E\\\norm{\vx}_2=1}}\norm{\mb\vx}_2=\vsig_{k+1}(\mb). \] Then, by \eqref{eq:smoothed_sv} with $d\gets m$ and $j\gets m-k$, \begin{equation}\label{eq:wide_bound}
        \Pr\Brack{\vsig_{k+1}(\tma_{S:})\leq\frac{\alpha\sig(m-k)}{\sqrt{m}}}\leq\Pr\Brack{\vsig_{k+1}(\mb)\leq\frac{\alpha\sig(m-k)}{\sqrt{m}}}\leq\Par{\sqrt{2e}\alpha}^{(m-k)^2}.
    \end{equation} Let $C'\defeq\half\Par{1-\frac{1}{(1-\eta)C}}>0$, and note that \[ m-k\geq\Par{\frac{(1-\eta)n}{d}-1}k=\Par{1-\frac{d}{(1-\eta)n}}\Par{\frac{(1-\eta)kn}{d}}\geq C'm. \] Setting $\alpha=\frac{m\Delta}{\sig(m-k)}\leq\frac{\Delta}{C'\sig}$ in \eqref{eq:wide_bound}, \[ \Pr\Brack{\vsig_{k+1}(\tma_{S:})\leq\sqrt{m}\Delta}\leq\Par{\sqrt{2e}\alpha}^{(m-k)^2}\leq\Par{\frac{\sqrt{2e}\Delta}{C'\sig}}^{(m-k)^2}. \] Now, we can set $\Delta=\frac{C'}{\sqrt{2e}}(\frac{\delta\sig}{n})^{\frac{2}{C'}}\leq\frac{C'\sig}{\sqrt{2e}}(\frac{\delta}{n})^{\frac{2}{C'}}$ to give \[ \Par{\frac{\sqrt{2e}\Delta}{C'\sig}}^{(m-k)^2}\leq\Par{\frac{\delta}{n}}^{\frac{2(m-k)^2}{C'}}\leq\Par{\frac{\delta}{n}}^{2m(m-k)}\leq\Par{\frac{\delta}{n}}^{m+1}\leq\frac{\delta^{m+1}}{n\binom{n}{m}}\leq\frac{\delta}{(d-1)\binom{n}{m}}, \] which establishes the claim.
\end{proof}

\begin{lemma}\label{lem:square}
    In the setting of Lemma~\ref{lem:wide}, the result holds if we suppose instead that $d<m\leq Kd$, where $K\defeq\sqrt[3]{C}$.
\end{lemma}

\begin{proof}
    Similarly to the proof of Lemma~\ref{lem:wide}, remove any $m-d$ rows of $\tma_{S:}$ to obtain $\mb\in\R^{d\times d}$. Then \begin{equation}\label{eq:square_bound}
        \Pr\Brack{\vsig_{k+1}(\tma_{S:})\leq\frac{\alpha\sig(d-k)}{\sqrt{d}}}\leq\Pr\Brack{\vsig_{k+1}(\mb)\leq\frac{\alpha\sig(d-k)}{\sqrt{d}}}\leq\Par{\sqrt{2e}\alpha}^{(d-k)^2}.
    \end{equation}
    Since $m\leq Kd$, we have $\frac{(1-\eta)kn}{d}\leq Kd$, which implies $k\leq\frac{Kd^2}{(1-\eta)n}\leq\frac{Kd}{(1-\eta)C}<d$. Thus $d-k\geq K'd$, where $K'\defeq1-\frac{K}{(1-\eta)C}>0$. Setting $\alpha=\frac{\sqrt{md}\Delta}{\sig(d-k)}\leq\frac{\sqrt{K}\Delta}{K'\sig}$ in \eqref{eq:square_bound}, \[ \Pr\Brack{\vsig_{k+1}(\tma_{S:})\leq\sqrt{m}\Delta}\leq\Par{\sqrt{2e}\alpha}^{(d-k)^2}\leq\Par{\frac{\sqrt{2eK}\Delta}{K'\sig}}^{(d-k)^2}. \] Now, we can set $\Delta=\frac{K'}{\sqrt{2eK}}\Par{\frac{\delta\sig}{n}}^{\frac{2K}{K'}}\leq\frac{K'\sig}{\sqrt{2eK}}\Par{\frac{\delta}{n}}^{\frac{2K}{K'}}$ to give \[ \Par{\frac{\sqrt{2eK}\Delta}{K'\sig}}^{(d-k)^2}\leq\Par{\frac{\delta}{n}}^{\frac{2K(d-k)^2}{K'}}\leq\Par{\frac{\delta}{n}}^{2Kd}\leq\Par{\frac{\delta}{n}}^{m+1}\leq\frac{\delta^{m+1}}{n\binom{n}{m}}\leq\frac{\delta}{(d-1)\binom{n}{m}}, \] which establishes the claim.
\end{proof}

\subsection{Conditioning of tall smoothed matrices}\label{ssec:tall}

In this section we provide tools for lower bounding the smallest singular value of a random $\Omega(d) \times d$ smoothed matrix $\tma = \ma + \mg$, where $\mg$ has entries $\simiid \Nor(0, \sig^2)$. Specifically we provide estimates, for sufficiently small $\alpha$, on the quantity
\begin{equation}\label{eq:square_tail_bound}
\Pr\Brack{\vsig_d\Par{\tma} < \alpha}.
\end{equation}
We first use the following standard result bounding $\normsop{\tma}$.
\begin{lemma}\label{lem:gaussian_opnorm}
Let $\sig\in(0,1)$, $m\geq d$, $\ma\in\R^{m\times d}$ have rows $\{\va_i\}_{i\in[m]}$ such that $\norm{\va_i}_2=1$ for all $i\in[m]$, $\mg\in\R^{m\times d}$ have entries $\simiid\Nor(0,\sig^2)$, and $\tma\defeq\ma+\mg$. Then there exists a constant $\cop>0$ such that for all $\delta\in(0,1)$,
\[\Pr\Brack{\normsop{\tma} > \cop \sqrt{m + \log\Par{\frac 1 \delta}}} \le \delta.\]
\end{lemma}
\begin{proof}
By Theorem 4.4.5, \cite{Vershynin24}, we have that for some constant $\cop > 2$,
\[\Pr\Brack{\normop{\mg} > \frac{\cop} 2 \sqrt{m + \log\Par{\frac 1 \delta}}} \le \delta,\]
where we used $\sig \le 1$.
The conclusion follows as $\normsop{\ma} \le \normsf{\ma} \le \sqrt m$. 
\end{proof}

We also require the definition of an $\eps$-net and a standard bound on its size.

\begin{definition}\label{def:net}
Let $S \subseteq \R^d$, $\net \subset S$ be finite, and $\eps \in (0, 1)$. We say that $\net$ is an $\eps$-net of $S$ if 
\[\sup_{\vu \in S} \min_{\vv \in \net} \norm{\vv - \vu}_2 \le \eps. \]
\end{definition}

\begin{fact}[Corollary 4.2.13, \cite{Vershynin24}]\label{fact:ball_net_size}
Let $\partial \ball_2(1)$ denote the boundary of the unit norm ball in $\R^d$. For all $\eps \in (0, 1)$, there exists an $\eps$-net of $\partial \ball_2(1)$ with $|\net| \le (\frac 3 \eps)^d$.
\end{fact}

We next observe that it suffices to provide estimates on a net, given an operator norm bound.

\begin{lemma}\label{lem:net_suffices}
Let $m \ge d$, $\eps \in (0, 1)$, and $\net$ be an $\eps$-net of $\partial \ball_2(1) \subset \R^d$, and suppose that $\tma \in \R^{m \times d}$ satisfies $\normsop{\tma} \le \rho$. Then $\vsig_d(\tma) \ge \min_{\vv \in \net} \norms{\tma \vv}_2 - \eps\rho$.
\end{lemma}
\begin{proof}
Let $\vu$ realize $\vsig_d(\tma)$ in the definition $\vsig_d(\tma) = \min_{\vu \in \partial \ball_2(1)} \norms{\tma\vu}_2$.
Then if we define $\vv \defeq \argmin_{\vv \in \net}\norms{\vu - \vv}_2$, the result follows from the triangle inequality:
\[\norms{\tma \vu}_2 \ge \norms{\tma \vv}_2 - \normsop{\tma}\norm{\vu-\vv}_2 \ge \min_{\vv \in \net}\norms{\tma \vv}_2 - \eps\rho.\]
\end{proof}

Finally, we provide tail bounds on the contraction given by $\tma$ on a single fixed vector.

\begin{lemma}\label{lem:single_vector_small}
In the setting of Lemma~\ref{lem:gaussian_opnorm}, let $\vv \in \R^d$ have $\norms{\vv}_2 = 1$. Then,
\[\Pr\Brack{\norms{\tma \vv}_2 < \alpha} \le \Par{\frac \alpha \sig}^m \text{ for all } \alpha \in (0, 1).\]
\end{lemma}
\begin{proof}
Observe that if $\norms{\tma \vv}_2 \le \alpha$, then every coordinate of $\tma \vv$ is bounded by $\alpha$. Each coordinate of $\tma \vv$ is distributed independently as $\Nor(\inprod{\va_i}{\vv}, \sig^2)$. Now the claim follows because for all $i \in [m]$,
\begin{align*}\Pr_{\xi \sim \Nor(\inprod{\va_i}{\vv}, \sig^2)}\Brack{\Abs{\xi} < \alpha} &= \Pr_{\xi \sim \Nor(\frac 1 \sig \inprod{\va_i}{\vv}, 1)}\Brack{\Abs{\xi} < \frac \alpha \sig} \\
&= \frac 1 {\sqrt {2\pi}} \int_{-\frac \alpha \sig}^{\frac \alpha \sig} \exp\Par{-\frac{(\xi-\frac 1 \sig \inprod{\va_i}{\vv})^2}{2}}\dd \xi \le \frac \alpha \sig.\end{align*}
\end{proof}

We can now prove our main tail bound on $\vsig_d(\tma)$.

\begin{lemma}\label{lem:sigd_rectangular}
In the setting of Lemma~\ref{lem:gaussian_opnorm}, suppose that $m\geq Kd$ for a constant $K>1$. Then there exists a constant $\beta>0$ such that
\[\Pr\Brack{\vsig_d\Par{\tma} < \alpha} \le 2\Par{\frac{2\alpha}{\sig}}^{\frac{(K-1)m}{2K}} \text{ for all } \alpha \in \Par{0, \frac {\sig^\frac{K+3}{K-1}} {\beta m^{\frac{2}{K-1}}}}.\]
\end{lemma}
\begin{proof}
Let $\delta\defeq2\Par{\frac{2\alpha}{\sig}}^{\frac{(K-1)m}{2K}}$. By Lemma~\ref{lem:gaussian_opnorm}, there exists $\cop>0$ such that \[ \Pr\Brack{\normsop{\tma}>\rho}\leq\frac{\delta}{2}\text{ for } \rho \defeq \cop\sqrt{m + \log\Par{\frac 2 \delta}}. \]
Let $L>0$ be a constant such that $\sqrt{\log(\frac{c}{2})}\leq Lc^{\frac{K-1}{4}}$ for all $c\geq2$, and let \[ \beta\defeq\max\Par{2e,\Par{2^{\frac{K+1}{2}}\cdot6\cop L}^{\frac{4}{K-1}}}\geq2e. \] Then for the stated range of $\alpha$, \[ \alpha\leq\frac{\sig}{2e}\implies\rho=\cop\sqrt{m+\frac{(K-1)m}{2K}\log\Par{\frac{\sig}{2\alpha}}}\le 2\cop\sqrt{m\log\Par{\frac{\sig}{2\alpha}}}. \] Let $\eps \defeq \frac \alpha {\rho}$, and let $\net$ be an $\eps$-net of $\partial \ball_2(1) \in \R^d$ with size $|\net| \le (\frac 3 \eps)^d$, as guaranteed by Fact~\ref{fact:ball_net_size}. Then by taking a union bound over Lemma~\ref{lem:single_vector_small} applied to each $\vv \in \net$, \begin{align*}
    \Pr\Brack{\min_{\vv\in\net}\norms{\tma\vv}_2<2\alpha}&\leq|\net|\Par{\frac{2\alpha}{\sig}}^m\leq\Par{\frac{6\cop\sqrt{m\log(\frac{\sig}{2\alpha})}}{\alpha}}^{\frac{m}{K}}\Par{\frac{2\alpha}{\sig}}^m\\
    &=\Par{\frac{6\cop\sqrt{m\log(\frac{\sig}{2\alpha})}}{\alpha}}^{\frac{m}{K}}\Par{\frac{2\alpha}{\sig}}^{\frac{(K+1)m}{2K}}\Par{\frac{2\alpha}{\sig}}^{\frac{(K-1)m}{2K}}\\
    &\leq\Par{\frac{2^{\frac{K+1}{2}}\cdot6\cop\sqrt{m\log(\frac{\sig}{2\alpha})}}{\sig^{\frac{K+1}{2}}}\cdot\frac{\sig^{\frac{K+3}{4}}\alpha^{\frac{K-1}{4}}}{2^{\frac{K+1}{2}}\cdot6\cop L\sqrt{m}}}^{\frac{m}{K}}\Par{\frac{2\alpha}{\sig}}^{\frac{(K-1)m}{2K}}\\
    &\leq\Par{\frac{2^{\frac{K+1}{2}}\cdot6\cop\sqrt{m\log(\frac{\sig}{2\alpha})}}{\sig^{\frac{K+1}{2}}}\cdot\frac{\sig^{\frac{K+1}{2}}}{2^{\frac{K+1}{2}}\cdot6\cop\sqrt{m\log(\frac{\sig}{2\alpha})}}}^{\frac{m}{K}}\Par{\frac{2\alpha}{\sig}}^{\frac{(K-1)m}{2K}}\\
    &=\Par{\frac{2\alpha}{\sig}}^{\frac{(K-1)m}{2K}}=\frac{\delta}{2},
\end{align*} where the third line uses \[ \alpha^{\frac{K-1}{4}}\leq\frac{\sig^{\frac{K+3}{4}}}{2^{\frac{K+1}{2}}\cdot6\cop L\sqrt{m}}\implies\alpha^{\frac{K-1}{2}}\leq\frac{\sig^{\frac{K+3}{4}}\alpha^{\frac{K-1}{4}}}{2^{\frac{K+1}{2}}\cdot6\cop L\sqrt{m}} \] for the stated range of $\alpha$ and the fourth line uses $\sqrt{\log(\frac{c}{2})}\leq Lc^{\frac{K-1}{4}}$ for $c\geq2$. The claim follows from a union bound on the above two events and Lemma~\ref{lem:net_suffices}.
\end{proof}

Applying Lemma~\ref{lem:sigd_rectangular} then gives our extension to tall matrices.

\begin{lemma}\label{lem:tall}
    In the setting of Lemma~\ref{lem:wide}, the result holds if we suppose instead that $m>Kd$, where $K\defeq\sqrt[3]{C}$, and in addition that $\ma$ has rows $\{\va_i\}_{i\in[n]}$ satisfying $\norm{\va_i}_2=1$ for all $i\in[n]$.
\end{lemma}

\begin{proof}
    Let $\beta$ be the constant in Lemma~\ref{lem:sigd_rectangular}, and let $\alpha=\sqrt{m}\Delta$, where \[ \Delta=\frac{1}{2}\Par{\frac{\delta\sig}{\beta n}}^{\frac{4K}{K-1}+1}\leq\frac{\sig}{2\sqrt{m}}\Par{\frac{\delta\sig}{\beta n}}^{\frac{4K}{K-1}}\in \Par{0, \frac {\sig^\frac{K+3}{K-1}} {\beta m^{\frac{2}{K-1}+\half}}}. \] By Lemma~\ref{lem:sigd_rectangular}, \begin{align*}
        \Pr\Brack{\vsig_{k+1}(\tma_{S:})\leq\sqrt{m}\Delta}&\leq\Pr\Brack{\vsig_d(\tma_{S:})\leq\sqrt{m}\Delta}\leq2\Par{\Par{\frac{\delta\sig}{\beta n}}^{\frac{4K}{K-1}}}^{\frac{(K-1)m}{2K}}\\
        &\leq\Par{\frac{\delta}{n}}^{2m}\leq\Par{\frac{\delta}{n}}^{m+1}\leq\frac{\delta^{m+1}}{n\binom{n}{m}}\leq\frac{\delta}{(d-1)\binom{n}{m}},
    \end{align*} which establishes the claim.
\end{proof}

\subsection{Assumption~\ref{assume:simplify} for smoothed matrices}\label{ssec:smoothed_forster}

We can now put together the previous results to give a diameter bound for smoothed matrices. To begin, we show simple norm bounds on the rows of a smoothed matrix.

\begin{lemma}\label{lem:norm_bound}
    Under Assumption~\ref{assume:nd}, let $\delta\in(0,1)$, $\frac 1 \sig \ge 10(d + \log(\frac n \delta))$, $\ma\in\R^{n\times d}$ have rows $\{\va_i\}_{i\in[n]}$ such that $\norm{\va_i}_2=1$ for all $i\in[n]$, $\mg\in\R^{n\times d}$ have entries $\simiid\Nor(0,\sig^2)$, and $\tma\defeq\ma+\mg$ have rows $\{\tva_i\}_{i\in[n]}$. Then with probability $\geq1-\delta$, $\frac{1}{6}\leq\norm{\tva_i}_2^2\leq2$ for all $i \in [n]$.
\end{lemma}

\begin{proof}
By using the inequalities
\begin{align*}
\half \norm{\va_i}_2^2 - \norm{\vg_i}_2^2 \le \norm{\va_i + \vg_i}_2^2 \le \frac 3 2 \norm{\va_i}_2^2 + 3\norm{\vg_i}_2^2,
\end{align*}
it is enough to show that for all $i \in [n]$, the probability that $\norm{\vg_i}_2^2 \ge \frac 1 6$ is bounded by $\frac \delta n$. This follows from a standard $\chi^2$ tail bound, e.g., Lemma 1, \cite{LaurentM00}, for our choice of $\sig$.
\end{proof}

It remains to show that $\frac{d}{n}\1_n$ is deep inside the basis polytope of the rows of $\tma$, so that we can use Lemma~\ref{lem:diam_bound} to obtain a diameter bound for minimizing Barthe's objective.

\begin{lemma}\label{lem:dn_deep}
    In the setting of Lemma~\ref{lem:norm_bound}, $\frac{d}{n}\1_n$ is $(\eta,\Delta)$-deep inside the basis polytope of the rows of $\tma$ with probability $\geq1-\delta$, where $\eta\defeq1-\frac{1}{\sqrt{C}}$ and $\Delta = \Par{\frac{\delta\sig}{n}}^{O(1)}$.
\end{lemma}

\begin{proof}
    Let $k\in[d-1]$. By Lemma~\ref{lem:small_sv}, if some $m\defeq\lceil\frac{(1-\eta)kn}{d}\rceil$ rows of $\tma$ indexed by $S$ violate the condition for $(\eta,\Delta)$-deepness, then $\vsig_{k+1}(\tma_{S:})\leq\sqrt{m}\Delta$. By Lemma~\ref{lem:wide}, Lemma~\ref{lem:square}, and Lemma~\ref{lem:tall}, \[ \Pr\Brack{\vsig_{k+1}(\tma_{S:})\leq\sqrt{m}\Delta}\leq\frac{\delta}{(d-1)\binom{n}{m}} \] for any $S\subseteq[n]$ with $|S|=m$, in every range of $k \in [d - 1]$. By a union bound over all $S$, the failure probability is at most $\frac{\delta}{d-1}$.     The result follows by a union bound over all $k\in[d-1]$.
\end{proof}

At this point, we have all the tools necessary to prove Theorem~\ref{thm:smoothed}.

\restatesmoothed*
\begin{proof}
By Lemma~\ref{lem:norm_bound}, in the relevant range of $\sig$, the conclusion $\frac{1}{6}\leq\norm{\tva_i}_2^2\leq2$ holds for all $i\in[n]$ with probability $\geq 1-\frac{\delta}{2}$. Moreover, let $\Delta=(\frac{\delta\sig}{n})^{O(1)}$ so that $\frac{d}{n}\1_n$ is $(\eta,\Delta)$-deep inside the basis polytope of the rows of $\tma$ with probability $\geq1-\frac{\delta}{2}$ by Lemma~\ref{lem:dn_deep}. By a union bound on these events, with probability $\geq1-\delta$, we can apply Lemma~\ref{lem:diam_bound} and conclude that there exists $\vt^\star\in\argmin_{\vt\in\R^n}f(\vt)$ satisfying \[ \norm{\vt^\star}_\infty\leq\half\log\Par{\frac{12n}{d}\Par{\frac{8}{\eta\Delta^2}}^{d-1}} = O\Par{d\log\Par{\frac 1 \sig}}. \]
\end{proof}

\subsection{Extension to non-uniform \texorpdfstring{$\vc$}{}}\label{ssec:gen_c}

Our smoothed diameter bound in Theorem~\ref{thm:smoothed} is stated with respect to uniform marginals $\vc = \frac d n \1_n$. However, the analysis in this section can be straightforwardly extended to hold for $\vc$ with nonuniform entries by a reduction, as long as $\vc$ is sufficiently bounded away from $\1_n$ entrywise.

\begin{corollary}\label{cor:nonuniform}
    In the setting of Theorem~\ref{thm:smoothed}, the result holds if we suppose instead that $\vc\in(0,1]^n$ satisfies $\norm{\vc}_1=d$ and $\vc\leq c\cdot\frac{d}{n}\1_n$ entrywise, where $1<c<C\leq\frac{n}{d}$ for constants $c, C$.
\end{corollary}

\begin{proof}
    Our proof of Lemma~\ref{lem:dn_deep} shows that with probability $\ge 1 - \delta$ in the setting of Theorem~\ref{thm:smoothed}, $\frac d n \1_n$ is $(\eta, \Delta)$-deep for a constant $\eta$ arbitrarily close to $1-\frac{1}{C}$, where $C\leq\frac{n}{d}$ (see Remark~\ref{rem:eta}). However, this also implies that for all $k \in [d - 1]$ and $k$-dimensional subspaces $E$, recalling Definition~\ref{def:deep},
    \[\sum_{i\in[n]}\vc_i\ind_{\norm{\va_i-\proj_E\va_i}_2\leq\Delta}\leq c\sum_{i\in[n]}\frac{d}{n}\ind_{\norm{\va_i-\proj_E\va_i}_2\leq\Delta}\leq c(1-\eta)k=(1-(1-c(1-\eta)))k.\]
    Thus, we have shown that $\vc$ is also $(1 - c(1 - \eta), \Delta)$-deep. Since $\eta$ can be arbitrarily close to $1-\frac{1}{C}$ and $c<C$, we can verify that the new parameter $1 - c(1 - \eta)$ satisfies $0<1-c(1-\eta)<1-\frac{1}{C}$, so the rest of our proof applies (propagating constant changes appropriately) by Remark~\ref{rem:eta}.
\end{proof}
\section*{Acknowledgments}

We would like to thank Mehtaab Sawhney for providing several pointers to the random matrix theory literature, in particular suggesting the use of Fact~\ref{fact:smoothed_sv}. We would also like to thank Li Chen for providing clarifications on \cite{ChenPW21}, that were used in Proposition~\ref{prop:box_oracle_impl} and Appendix~\ref{app:boxoracle}, as well as Richard Peng and Di Wang for providing other helpful comments on \cite{ChenPW21}.

\bibliographystyle{alpha}
\bibliography{ref.bib}

\newpage
\appendix

\section{Discussion of Proposition~\ref{prop:box_oracle_impl}}\label{app:boxoracle}

In this section, we describe the modifications to \cite{ChenPW21} that are needed to prove Proposition~\ref{prop:box_oracle_impl}. It is clear that by reparameterizing $\vl \gets \vl - \vt$ and $\vr \gets \vr - \vt$ for the vectors $(\vt, \vl, \vr)$ defining the input set $\calW$ in Proposition~\ref{prop:box_oracle_impl}, it is enough to describe how to obtain the stated guarantee with $\vt = \0_n$. In other words, we need to obtain $\vv$ satisfying $\vl \le \vv \le \vr$ entrywise, and
\[\inprod{\vb}{\vv} + \half \ml\Brack{\vv, \vv} \le \half \min_{\vl \le \vw \le \vr}\Brace{\inprod{\vb}{\vw} + \half \ml\Brack{\vw, \vw}}.\]
\paragraph{Structure of \cite{ChenPW21}.} We first provide an overview of the components of \cite{ChenPW21}. Their core algorithmic result (which proves Proposition~\ref{prop:box_oracle_impl} when the box constraint is $[0, \infty)^n$) consists of three sub-pieces, Lemmas 4.3, 4.4, and 4.5. These pieces respectively correspond to an accelerated proximal point method, the construction of $j$-tree sparsifiers, and solving a diffusion instance on a $j$-tree. Of the three, Lemma 4.4 is independent of the constraint on the optimization problem, and thus can be left untouched, so we focus on the other two components.

\paragraph{Modifying Lemma 4.5.} Lemma 4.5 is proven in Sections 6 and 7 of \cite{ChenPW21}, but the only piece of the proof that interacts with the constraint set explicitly is Lemma 6.6, shown in Section 7.2. This result provides a data structure that supports efficient modifications to a VWF (vertex weighting function, i.e., a convex piecewise-quadratic function with concave continuous derivative). The data structure maintains the coefficients and cutoff points for each piece of the VWF.

The data structure is implemented using a segment tree, where each leaf corresponds to a cutoff point $s_i$ in a piecewise quadratic with $k$ pieces. Implicitly, the data structure in \cite{ChenPW21} sets the rightmost cutoff point $s_k = \infty$ and never modifies it. Instead, we can augment the data structure to maintain an explicit cutoff $s_k$. It is straightforward to check that all updates to $s_k$ (i.e., in the operations given by Claims 7.9, 7.10, 7.11) in \cite{ChenPW21} can be handled in $O(1)$ time, as there are closed-form formulas in each case. 
We remark that similar segment tree-based data structures supporting changes to a dynamic piecewise-polynomial function, but with an explicit right endpoint, have appeared in the recent literature, see e.g., Section 3 of \cite{HuJTY24}.

\paragraph{Modifying Lemma 4.3.} Lemma 4.3 is proven in Section 8 of \cite{ChenPW21}, and is an approximation-tolerant variant of the classical accelerated proximal point algorithm \cite{Guler92}. The only explicit property about the constraint set used in this section is convexity, as efficient projection is handled by Lemma 4.5. This property holds for axis-aligned boxes, so Lemma 4.3 extends to our setting.
\section{Numerical Precision Considerations}\label{app:numerical}

For brevity, we specialize our discussion in this section to the setting of Theorem~\ref{thm:main}, where $\vc_{\min}, \eps, \delta = \poly(\frac 1 n)$, and $\log(\kappa)$ is either $O(\log(n))$ (the \emph{well-conditioned regime}), or $O(d\log(n))$ (the \emph{smoothed analysis regime}). The latter name is justified by applying Theorem~\ref{thm:smoothed} with $\sigma = \poly(\frac 1 n)$. We also assume that entries of $\ma$ are represented with $b$-bit numbers, where $b = O(\log(n))$.

From the perspective of numerical stability, the bottleneck in both regimes is computing derivatives of Barthe's objective up to $\poly(\frac 1 n)$ additive error. All other operations performed by the box-constrained Newton's method in Theorem~\ref{thm:main}, either using explicit Hessian computations (Remark~\ref{rem:log_delta_eps}) or implicit Hessian sparsification (Theorem~\ref{thm:implicit_sparsify}), is tolerant to $\poly(\frac 1 n)$ additive error in entries of vectors and matrices. Thus, we can simply truncate our bit representations to $O(\log(n))$-sized words, as long as we can compute gradients, Hessians, or Hessian-vector products stably.

\paragraph{Bottleneck operations.} Recall Fact~\ref{fact:barthe_derivs}, and define for $\vt \in \R^n$,
\[\mz(\vt) \defeq \sum_{i \in [n]} \exp\Par{\vt_i} \va_i\va_i^\top.\]
For an instance of Definition~\ref{def:rip} where Assumption~\ref{assume:simplify} holds, the bottleneck operations are evaluating, for $\vt \in \R^n$ with $\norm{\vt}_\infty \le \log(\kappa)$, gradients $\nabla f(\vt)$:
\begin{align*}
\Brace{\Tr\Par{\exp(\vt_i)\va_i\va_i^\top\Par{\mz(\vt)}^{-1}}}_{i \in [n]}, 
\end{align*}
Hessians $\nabla^2 f(\vt)$:
\begin{align*}
\Brace{\Tr\Par{\exp(\vt_i)\va_i\va_i^\top\Par{\mz(\vt)}^{-1}} \ind_{i = j} - \Tr\Par{\exp(\vt_i + \vt_j)\va_i\va_i^\top\Par{\mz(\vt)}^{-1}\va_j\va_j^\top \Par{\mz(\vt)}^{-1}}}_{(i, j) \in [n] \times [n]},
\end{align*}
and Hessian-vector products with the above form.
Lemma~\ref{lem:grad_compute} gives methods for performing these operations without consideration of numerical stability. 

In the well-conditioned regime $\norm{\vt}_\infty = O(\log(n))$, all of the numbers in the above expressions are representable up to $\poly(\frac 1 n)$ error using $O(\log(n))$-bit words. Thus, we believe that our algorithms are numerically stable at constant overhead to the bit complexity, as these bottleneck operations reduce to standard matrix manipulations on $O(\log(n))$-bit entries.

In the smoothed analysis regime $\norm{\vt}_\infty \gtrsim d$, it is possible that the computation of gradients and Hessians described here is not numerically stable using $O(\log(n))$-bit representations of numbers. As a result, our algorithms as stated may require larger bit complexity, and suffer in runtime.

We note that all other work on optimizing Barthe's objective, based on cutting-plane methods \cite{HardtM13} or gradient descent \cite{ArtsteinKS20}, also suffers from the same numerical stability challenges, because they also require computing gradients. Thus, under any reasonable cost model our algorithm is the state-of-the-art by at least a $\approx \frac n d$ factor in the smoothed analysis regime.

We think the strategy of optimizing Barthe's objective for computing Forster transforms is a natural one, and hence evaluating the numerical stability of gradient and Hessian computation under finite bit precision is an important goal (for any derivative-based algorithm, not just ours). Alternatively, can we prove tighter conditioning bounds than Theorem~\ref{thm:smoothed}, in natural statistical models?

\paragraph{Strongly polynomial methods.} Computing Forster transforms in finite-precision arithmetic was explicitly studied by \cite{DiakonikolasTK23, DadushR24}, who considered a definition of strongly polynomial introduced by \cite{GrotschelLS88} that requires polynomial space complexity. Each proved the types of rounding result described at the end of Section~\ref{ssec:results}; however, the resulting bit complexities are rather large. For example, Theorem 5.1 of \cite{DiakonikolasTK23} proves magnitude bounds of $\approx \exp(d^3 b)$ where $b$ is an initial word size, with an unspecified intermediate bit complexity. Corollary 4.10 of \cite{DadushR24} yields the improved estimate of $\approx ndb$ bits required, for stable implementation on worst-case instances. However, it is unclear how these rounding procedures affect progress on Barthe's objective.

\end{document}